\documentclass[10pt,a4paper]{article}
\usepackage[utf8]{inputenc}
\usepackage{geometry}
\usepackage[titletoc,toc,title]{appendix}
\usepackage{graphicx}
\usepackage{amssymb}
\usepackage{epstopdf}
\usepackage{mathrsfs}
\usepackage{amsmath}
\usepackage{amsthm}
\begin{document}
\newtheorem{theoreme}{Theorem}[section]
\newtheorem{ex}{Example}[section]
\newtheorem{definition}{Definition}[section]
\newtheorem{lemme}{Lemma}[section]
\newtheorem{remarque}{Remark}[section]
\newtheorem{exemple}{Example}[section]
\newtheorem{proposition}{Proposition}[section]
\newtheorem{corolaire}{Corollary}[section]
\newtheorem{hyp}{Hypothesis}[section]
\newtheorem*{rec}{Recurrence Hypothesis}
\newtheorem*{prop}{Proposition}
\newtheorem*{theo}{Theorem}
\newtheorem*{Conj}{Conjecture}
\newcommand\Bound{\partial\overline{X}}
\newcommand\bdf{b_\xi}
\newcommand\Euc{\mathbb{R}^d}
\newcommand\COIn{\tilde{\chi}}
\newcommand\Comp{\overline{X}}
\newcommand\Dir{\varphi}
\newcommand\Lag{\mathcal{L}_0}
\newcommand\ins{\gamma_{uns}}
\newcommand\ext{\mathcal{W}_{0}}
\newcommand\brut{A_7}
\newcommand\inter{\mathcal{W}_2}
\newcommand\close{\mathcal{W}_1}
\newcommand\tins{N_{uns}}
\newcommand\din{\delta_{in}}
\newcommand\inco{\mathcal{DE}_-}
\newcommand\se{\epsilon_{sec}}
\newcommand\Nlag{N_{lag}}
\newcommand\cmin{c_{min}}
\newcommand\ind{J_a^{h,r}}
\newcommand\sor{\mathcal{T}}
\newcommand\secur{\varepsilon_0}
\newcommand\diam{\varepsilon_6}
\newcommand\tim{T_0}
\newcommand\zone{\overline{a}}
\newcommand\exit{N_\epsilon}
\newcommand\wait{N_1}
\newcommand\cli{\varepsilon_7}
\newcommand\nuag{\chi_{7}}
\newcommand\rain{\Pi_7}
\newcommand\sui{A_1^{\leq N}}
\newcommand\symp{\mathcal{C}_0}
\newcommand\zeit{t_2}
\newcommand\petit{\varepsilon_2}
\newcommand\temps{t_5}
\newcommand\set{\mathcal{W}_3}
\newcommand\taille{\varepsilon_2}
\newcommand\bel{>}
\newcommand\R{\mathbb{R}}
\newcommand\Sp{\mathbb{S}^{d-1}}
\newcommand\spt{\mathrm{supp}}
\makeatletter
\renewcommand\theequation{\thesection.\arabic{equation}}
\@addtoreset{equation}{section}
\makeatother
\title{Distorted plane waves on manifolds of nonpositive curvature}
\author{Maxime Ingremeau}

\maketitle
\begin{abstract}
We will consider the high frequency behaviour of distorted plane waves on manifolds of nonpositive curvature which are Euclidean or hyperbolic near infinity, under the assumption that the curvature is negative close to the trapped set of the geodesic flow and that the topological pressure associated to half the unstable Jacobian is negative.

We obtain a precise expression for distorted plane waves in the high frequency limit, similar to the one in \cite{GN} in the case of convex co-compact manifolds. In particular, we will show $L_{loc}^\infty$ bounds on distorted plane waves that are uniform with frequency. We will also show a small-scale equidistribution result for the real part of distorted plane waves, which implies sharp bounds for the volume of their nodal sets.
\end{abstract}

\section{Introduction}

Consider a Riemannian manifold $(X,g)$ such that there exists $X_0\subset X$ and $R_0\bel 0$ such that $(X\backslash X_0,g)$ and $(\mathbb{R}^d\backslash B(0,R_0),g_{eucl})$ are isometric (we shall say that such a manifold is \emph{Euclidean near infinity}). The distorted plane waves on $X$ are a family of functions $E_h(x;\xi)$ with parameters $\xi\in\mathbb{S}^{d-1}$ (the direction of propagation of the incoming wave) and $h$ (a semiclassical parameter corresponding to the inverse of the square root of the energy) such that
\begin{equation}\label{eigen}
(-h^2\Delta-1) E_h(x;\xi)=0, 
\end{equation}
and which can be put in the form
\begin{equation}\label{jeanne}
E_h(x;\xi)= (1-\chi)e^{i x\cdot \xi /h} + E_{out}.
\end{equation}
Here, $\chi\in C_c^\infty$ is such that $\chi\equiv 1$ on $X_0$, and $E_{out}$ is \emph{outgoing} in the sense that it is the image of a function in
$C_c^\infty(\mathbb{R}^d)$ by the outgoing resolvent
$(-h^2\Delta- 1+i0)^{-1}$. It can be shown (cf. \cite{Mel}) that there is only one function $E_h(\cdot;\xi)$ such that (\ref{eigen}) is satisfied and which can be put in the form (\ref{jeanne}).

In \cite{Ing}, the author studied the behaviour as $h\rightarrow 0$ of $E_h$, under some assumptions on the geodesic flow $\Phi^t: S^*X\rightarrow S^*X$. The \emph{trapped set} is defined as
\begin{equation*}
K:=\{(x,\xi)\in S^*X; \Phi^t(x,\xi) \text{ remains in a bounded set for all } t\in \mathbb{R}\}.
\end{equation*}

The main result of \cite{Ing} was that provided that $K$ is non-empty, that the dynamics is hyperbolic close to $K$, and that some topological pressure is negative, then $E_h(\cdot;\xi)$ can be written as a convergent sum of Lagrangian states associated to Lagrangian manifolds which are close to the unstable directions of the hyperbolic dynamics.

In this paper, we will show that this result can be made more precise if we work on a manifold of nonpositive curvature. In this framework, we shall show that all the Lagrangian states which make up $E_h(\cdot;\xi)$ are associated to Lagrangian manifolds which can be projected smoothly on the base manifold $X$. As a consequence, we will deduce the following estimates on the $C^\ell$ norms of the distorted plane waves.

We refer to section \ref{assumptions} for the general assumptions we need, and to section \ref{newyork} for the theorem concerning the decomposition of distorted plane waves into a sum of Lagrangian states. In this section, we will also be able to describe the semiclassical measures associated to the distorted plane waves.

\begin{theoreme}\label{linfini}
Let $(X,g)$ be a Riemannian manifold which is Euclidean near infinity. We suppose that $(X,g)$ has nonpositive sectional curvature, and that it has strictly negative curvature near $K$. We also suppose that
Hypothesis \ref{Husserl} on topological pressure is satisfied.

Let $\xi\in\mathbb{S}^{d-1}$, $\ell\in \mathbb{N}$ and $\chi\in C_c^\infty(X)$. Then there exists $C_{\ell,\chi}\bel 0$ such that, for any $h\bel 0$, we have
\begin{equation*}
\|\chi E_h(\cdot,\xi)\|_{C^\ell}\leq \frac{C_{\ell,\chi}}{h^\ell}.
\end{equation*}
\end{theoreme}

\paragraph{Small-scale equidistribution}
If $x_0\in X$ and $r>0$, let us write $B(x_0,r)$ for the geodesic ball centred at $x_0$ of radius $r$. The following result, which tells us that $\Re E_h$ has $L^2$ norm bounded from below on any ball of radius larger than $C h$ for $C$ large enough, can be seen as a “small-scale equidistributuon” result. Note that the upper bound in (\ref{smallscale21}) is just a consequence of Theorem \ref{linfini}.

\begin{theoreme}\label{smallscale7}
Let $(X,g)$ be a Riemannian manifold which is Euclidean near infinity. We suppose that $(X,g)$ has nonpositive sectional curvature, and that it has strictly negative curvature near $K$. We also suppose that
Hypothesis \ref{Husserl} on topological pressure is satisfied.

Let $\xi\in\mathbb{S}^{d-1}$, and $\chi\in C_c^\infty(X)$. There exist constants $C,C_1, C_2\bel 0$ such that the following holds. For any $x_0\in X$ such that $\chi(x_0)=1$, for any sequence $r_h$ such that $1\geq  r_h \bel C h$, we have for $h$ small enough:
\begin{equation}\label{smallscale21}
C_1 r_h^d\leq \int_{B(x_0,r_h)} |\Re E_h|^2(x) \mathrm{d}x\leq C_2 r_h^d.
\end{equation}

In particular, for any bounded open set $U\subset X$, there exists $c(U)\bel 0$ and $h_U\bel 0$ such that for all $0<h<h_U$, we have
\begin{equation}\label{bientolafin2}
\int_U |\Re E_h|^2\geq c(U).
\end{equation}
\end{theoreme}

Here, we could have considered the imaginary part of $E_h$ instead, or even $E_h$ itself, and we would have obtained the same result.

\paragraph{Nodal sets}
Theorem \ref{smallscale7} has interesting applications to the study of the nodal volume of the real part of the distorted plane waves. Namely, let us fix a compact set $\mathcal{K}$ and a $\xi\in \mathbb{S}^{d-1}$ and consider
$$N_{\mathcal{K},h}:=\{x\in \mathcal{K}; \Re(E_h)(x,\xi)=0\}.$$

Then we have the following estimate. We refer once again to section \ref{newyork} for the precise assumptions we make.
\begin{corolaire}\label{hydre}
We make the same assumptions as in Theorem \ref{linfini}. Then there exist $C_{\mathcal{K}}, C'_{\mathcal{K}}>0$ such that
\begin{equation}\label{yaudistorted}
\frac{C_\mathcal{K}}{h} \leq Haus_{d-1} (N_{\mathcal{K},h}) \leq \frac{C'_\mathcal{K}}{h},
\end{equation}
where $Haus_{d-1}$ denotes the $(d-1)$-dimensional Hausdorff measure.
\end{corolaire}

Here, again we could have considered the imaginary part of $E_h$ instead, and we would have obtained the same result.

The lower bound in (\ref{yaudistorted}) could be deduced from \cite{logunov2016lowernodal}, but we will give a proof of this fact in section \ref{lowernod} since it is easy to deduce from (\ref{bientolafin2}). As for the upper bound, it follows from the recent work \cite{hezari2016applications}, which says precisely that the upper bound in (\ref{yaudistorted}) can be deduced from small-scale equidistribution (\ref{smallscale21}). We refer to this paper for more details. For other applications of small-scale equidistribution, see \cite{hezari2016inner} and \cite{hezari2016quantum}.

\paragraph{Comparison of the results with the case of compact manifolds}
Nodal sets, small-scale behaviour and $L^p$ norms of eigenfunctions of the Laplace-Beltrami operator on compact manifolds have been actively studied recently. Let us recall what is known in this framework.

Let $(X,g)$ be a $d$-dimensional closed (compact, without boundary) manifold. Then there exists a sequence $h_n$ of positive numbers going to zero and an $L^2$ orthonormal basis of (real-valued) eigenfunctions $\phi_n$ such that
\begin{equation*}
-h_n^2\Delta \phi_n = \phi_n.
\end{equation*}

The following estimate on the $L^p$ norm of $\phi_n$ was proven in \cite{sogge1988concerning}:
\begin{equation*}
\|\phi_n\|_{L^p}\leq \frac{C}{h_n^{\sigma(p)}},
\end{equation*}
where $\sigma(p)= \frac{d-1}{2}\Big{(}\frac{1}{2}-\frac{1}{p}\Big{)}$
 if $2\leq p \leq \frac{2(d+1)}{d-1}$, and  $\sigma(p)= \frac{d-1}{2}-\frac{d}{p}$ if $\frac{2(d+1)}{d-1}\leq p \leq \infty$.
 These estimates are sharp if no further assumption is made on the manifold. However,
if $(X,g)$ has negative curvature, these bounds were slightly improved in \cite{HeRi}, \cite{hassell2015improvement} and \cite{SoggeNegativ}. We refer to these papers and to the references therein for precise statements and for more historical background. These estimates are far from showing that $\phi_n$ is bounded uniformly in $L^\infty$: actually, it is not clear if such a bound should hold.

The estimates given by
Theorem \ref{linfini} for distorted plane waves are therefore much better than what is available in the case of compact manifolds.

Small scale equidistribution results, similar to (\ref{smallscale21}) were obtained in \cite{Han} and in \cite{HeRi} for a density one sequence of eigenfunctions on the Laplace-Beltrami operator on compact manifolds of negative curvature and for sequences of the form $r_h\leq \frac{C_\alpha}{|\log h|^\alpha}$ for any $0<\alpha < 1/(2d)$. Small scale equidistribution results similar to (\ref{smallscale21}) were also obtained on the torus (see \cite{lester2015small} and the references therein). We may conjecture that on compact manifolds of negative curvature, an inequality like (\ref{smallscale21}) should hold for any sequence $r_h\bel \bel h$, for some density $1$ subsequence of eigenfunctions of the Laplace-Beltrami operator.

As for nodal sets, let us write $N_{h_n}:=\{x\in X; \phi_n(x)=0\}$. The following conjecture was made in \cite{Yau}.
\begin{Conj}[Yau]
There exists $c_1,c_2\bel 0$ such that
\begin{equation*}
\frac{c_1}{h_n}\leq Haud_{d-1}(N_{h_n})\leq \frac{c_2}{h_n}.
\end{equation*}
\end{Conj}

This conjecture was proven for analytic manifolds in \cite{DF}, and the lower bound was proven in dimension 2 in \cite{Bru}. A breakthrough was made recently in \cite{logunov2016lowernodal} and \cite{logunov2016uppernodal}, where the lower bound was proved and a polynomial upper bound was established respectively, in any dimension. 
Corollary \ref{hydre} can be seen as an analogue of Yau's conjecture in the case of non-compact manifolds.

\paragraph{Relation to other works}
An important part of this paper consists in describing the semiclassical measures associated to distorted plane waves. It can thus be considered as a (partial) generalization of \cite{GN}, where the authors describe the semiclassical measures associated to eigenfunctions of the Laplace-Beltrami operator on manifolds of infinite volume, with sectional curvature constant equal to $-1$ (convex co-compact hyperbolic manifolds). While the proofs in \cite{GN} rely heavily on the quotient structure of constant curvature hyperbolic manifolds, our proofs are based on the properties of the classical dynamics, being thus more versatile.

Semiclassical measures associated to distorted plane waves were studied in \cite{DG} under very general assumptions, and we follow their approach on many points. Semiclassical measures for Eisenstein series associated to specral parameters away from the spectrum of the Laplacian were also studied in \cite{Dcusp} and \cite{Yannick} on noncompact manifolds with finite volume (manifolds with cusps) using methods similar to ours. In all these papers, distorted plane waves are seen as the propagation in the long-time limit of usual plane waves (or hyperbolic waves). However, the reason for the convergence in the long time limit is very different in these papers: in \cite{Dcusp} and \cite{Yannick}, the convergence happens because the energy parameter is away from the real axis. In \cite{DG}, convergence occurs because the authors average on all directions and on a small energy layer. In this paper, just as in \cite{Ing}, convergence takes place because of a topological pressure assumption, as in \cite{NZ}. We will often use the methods and results of \cite{NZ} to take advantage of the hyperbolicity and topological pressure assumptions.

Nodal sets of distorted plane waves on manifolds of infinite volume were studied for the first time in \cite{JaNa} in the framework of Eisenstein series on convex co-compact hyperbolic surfaces. Since convex co-compact manifolds are analytic, the proof of \cite{DF}, which is purely local, gives an analogue of Yau's conjecture. The main results in \cite{JaNa} concern the counting of the number of intersections of the nodal sets $N_{\mathcal{K},h}$ with a given geodesic. Some of the results in \cite{JaNa} should still work on a manifold on nonpositive curvature under the assumptions made in this paper. This will be pursued elsewhere.

\paragraph{Organisation of the paper}
In section 2, we will state the general assumptions we need on the manifold $(X,g)$ and on the generalised eigenfunctions $E_h$. In section 3, we will recall the main results from \cite{Ing}, and state the new results we obtain. In section 4, we will prove results about the propagation of Lagrangian manifolds. In section 5, we will prove results about distorted plane waves, including Theorem \ref{linfini}. Finally, section 6 is devoted to the proof of Theorem \ref{smallscale7}, and of Corollary \ref{hydre}.

Appendix A simply recalls a few classical facts from semiclassical analysis. In Appendix B, we show that the general hypotheses formulated in \cite{Ing} and used in the present paper are fulfilled on manifolds that are hyperbolic near infinity.

\paragraph{Acknowledgements}
The author would like to thank Stéphane Nonnenmacher for his supervision and advice during this work.
He would also like to thank Frédéric Naud for suggesting to study nodal sets of distorted plane waves, as well as for fruitful discussion during the writing of this paper. Finally, the author would like to thank the anonymous referee for his many suggestions and comments.

The author is partially supported by the Agence Nationale de la Recherche project GeRaSic (ANR-13-BS01-0007-01).

\section{General assumptions}\label{assumptions}
In this section, we will state the main assumptions under which our results apply. The assumptions in section \ref{hypvariete} concern the background manifold $(X,g)$, while the assumptions in section \ref{hypdistorted} concern the distorted plane waves.
Most of these assumptions were already made in \cite{Ing}, in the framework of potential scattering. The additional assumptions which allow us to obtain more precise results than those of \cite{Ing} were regrouped in sections \ref{NewManifold} and \ref{NewDistorted}.

\subsection{Assumptions on the manifold}\label{hypvariete}

Let $(X,g)$ be a noncompact complete Riemannian manifold of dimension $d$, and let us denote by $p$ the classical Hamiltonian $p: T^*X \ni (x,\xi)\mapsto \|\xi\|_x^2\in \mathbb{R}$.

For each $t\in \mathbb{R}$, we denote by $\Phi^t:T^*X\longrightarrow T^*X$
the geodesic flow generated by $p$ at time $t$. We will write by the same letter its restriction $\Phi^t : S^*X \longrightarrow S^*X$ to the energy layer $p(x,\xi)=1$.

Given any smooth function $f : X \longrightarrow \mathbb{R}$, it may be
lifted to a function $f : T^*X \longrightarrow \mathbb{R}$, which we
denote by the same letter. We may then define $\dot{f}, \ddot{f}\in
C^\infty (T^*X)$ to be the derivatives of $f$ with respect to the
geodesic flow.
\begin{equation*}\dot{f}(x,\xi):= \frac{d}{dt}\Big{|}_{t=0} f(\Phi^t(x,\xi)),~~ \ddot{f}(x,\xi):=  \frac{d^2}{dt^2}\Big{|}_{t=0} 
f(\Phi^t(x,\xi)).
\end{equation*}

\subsubsection{Hypotheses near infinity} \label{Hector}

We suppose the following conditions are fulfilled.

\begin{hyp}[Structure of $X$ near infinity] \label{Guepard}
We suppose that the manifold $(X,g)$ is such that the following holds:

(1) There exists a compactification $\Comp$ of $X$, that is, a compact
manifold with boundaries $\Comp$ such that $X$ is diffeomorphic to the
interior of $\Comp$. The boundary $\Bound$ is called the boundary at
infinity.

(2) There exists a boundary defining function $b$ on $X$, that is, a
smooth function $b : \Comp \longrightarrow [0,\infty)$ such that $b>0$ on
$X$, and $b$ vanishes to first order on $\Bound$.

(3) There exists a constant $\epsilon_0>0$ such that for any point
$(x,\xi)\in S^*X$,
\begin{equation*}\text{if } b(x,\xi)\leq \epsilon_0 \text{ and } \dot{b}(x,\xi)=0 \text{
then } \ddot{b}(x,\xi)<0.
\end{equation*}
\end{hyp}

Part (3) in the hypothesis implies that any geodesic ball with a large enough radius is \emph{geodesically convex}.

\begin{ex} \label{nokia}
$\mathbb{R}^d$ fulfils the Hypothesis \ref{Guepard}, by taking the
boundary defining function $b(x)=(1+|x|^2)^{-1/2}$. $\overline{X}$ can then be identified with the closed unit ball in $\mathbb{R}^d$.
\end{ex}
\begin{ex} \label{samsung}
The Poincaré space $\mathbb{H}^{d}$
also fulfils the Hypothesis \ref{Guepard}. Indeed, in the ball model
$B_0(1)=\{x\in \mathbb{R}^d; |x|<1\}$, where $|\cdot|$ denotes the
Euclidean norm, then $\mathbb{H}^{d}$ compactifies to the closed unit ball,
and the boundary defining function $b(x)=2\frac{1-|x|}{1+|x|}$ fulfils
conditions (2) and (3).
\end{ex}
\begin{ex}
Let $e_1\in \mathbb{R}^2$, and consider $X= \mathbb{R}^2/ (\mathbb{Z}e_1)$, the flat two-dimensional cylinder. It may be compactified in the $e_2$-direction by setting $b(x)=(1+|x_2|^2)^{-1/2}$. $\overline{X}$ may then be identified with $S^1\times S^1$, and $b$ is a boundary defining function.

However, part (3) of the hypothesis is not satisfied. Indeed, for any $\epsilon\bel 0$, the set ${b=\epsilon}$ contains a closed geodesic (whose trajectory is just a circle). On this geodesic, we have $\dot{b}=0$ and $\ddot{b}=0$.
\end{ex}
We will write $X_0:=\{x\in X; b(x)\geq\epsilon_0/2\}$.
We will call $X_0$ the \emph{interaction region}. We will also write
\begin{equation}\label{frite}
V_0:=T^*(X\backslash X_0) = \{\rho\in T^*X; b(\rho) < \epsilon_0/2\}.
\end{equation}

By possibly taking $\epsilon_0$ smaller, we may ask that 
\begin{equation}\label{descendance}
\forall \rho \in S^*(X\backslash X_0), b(\Phi^1(\rho)) < \epsilon_0.
\end{equation}

\begin{definition}
If $\rho=(x,\xi)\in S^*X$, we say that $\rho$ escapes directly in the forward
direction, denoted $\rho\in\mathcal{D}\mathcal{E}_+$, if $b(x) <
\epsilon_0$ and
$\dot{b}(x,\xi)\leq 0$.

If $\rho=(x,\xi)\in S^*X$, we say that $\rho$ escapes directly in the backward
direction, denoted $\rho\in \mathcal{D}\mathcal{E}_-$, if $b(x)<
\epsilon_0$ and
$\dot{b}(x,\xi)\geq 0$.
\end{definition}
Note that we have
\begin{equation*}
\{\rho\in S^*X; b(\rho) < \epsilon_0\}=\mathcal{DE}_-\cup \mathcal{DE}_+.
\end{equation*}

\subsubsection{Hyperbolicity} \label{averse}
Let us now describe the hyperbolicity assumption we make.

For $\rho\in S^*X$, we will say that $\rho\in \Gamma^\pm$ if $\{\Phi^t(\rho), \pm t\leq 0\}$
is a bounded subset of $S^*X$; that is to say, $\rho$ does not “go to
infinity”, respectively in the past or
in the future. The sets $\Gamma^\pm$ are called respectively the
\textit{outgoing} and \textit{incoming} tails.

The \textit{trapped set} is defined as
\begin{equation*}K:=\Gamma^+\cap \Gamma^-.
\end{equation*}
It is a flow invariant set, and it is compact by the geodesic convexity assumption.

\begin{hyp}[Hyperbolicity of the trapped set] \label{sieste}
We assume that $K$ is non-empty, and is a hyperbolic set for the flow
$\Phi^t$. That is to say,
there exists an adapted metric $g_{ad}$ on a neighbourhood of $K$ included in
$S^*X$, and $\lambda>0$, such that the following holds. For each
$\rho\in K$, there is a decomposition \begin{equation*}T_\rho(S^*X)=\mathbb{R}\frac{\partial \big{(}\Phi^t(\rho)\big{)}}{\partial t} \oplus E_\rho^+\oplus
E_\rho^-
\end{equation*} such that
\begin{equation*}\|d\Phi_\rho^t(v)\|_{g_{ad}}\leq  e^{-\lambda|t|}\|v\|_{g_{ad}}
\text{   for all } v\in E_\rho^\mp, \pm t\geq 0.
\end{equation*}
\end{hyp}

The spaces $E_\rho^\pm$ are respectively called the \emph{unstable} and \emph{stable} spaces at $\rho$.

\begin{figure}
    \center
   \includegraphics[scale=0.4]{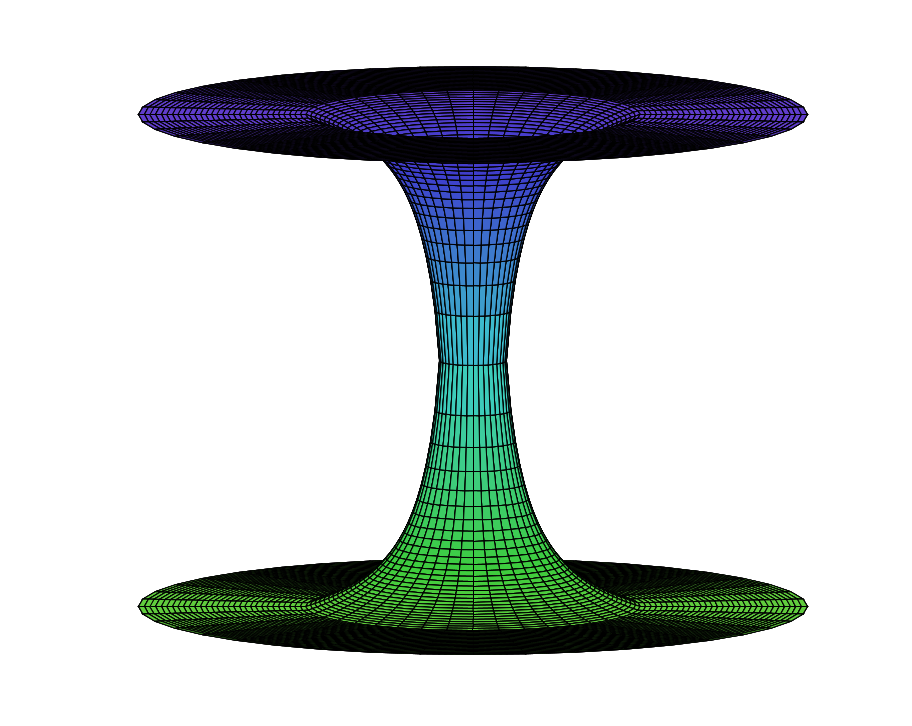}
    \caption{A surface which has negative curvature close to the trapped
set of the geodesic flow, and which is isometric to two copies of
$\mathbb{R}^2\backslash B(0,R_0)$ outside of a compact set. It satisfies Hypothesis \ref{sieste} near the trapped set (which consists of a single geodesic) and Hypothesis \ref{Guepard} near infinity.} \label{example}
\end{figure}

We may extend $g_{ad}$ to a metric
on $S^*X$, so that outside of the interaction region, it coincides with the restriction of the metric on $T^*X$ induced from the Riemannian metric on $X$. From now on, we will denote by
\begin{equation*}d_{ad} \text{ the Riemannian
distance associated to the metric } g_{ad} \text{ on } S^*X.
\end{equation*}

Any $\rho\in K$  admits local strongly (un)stable manifolds $
W_{\epsilon}^\pm (\rho)$, defined for $\epsilon\bel 0$ small enough by
\begin{equation*}W_{\epsilon}^\pm (\rho)=\{\rho'\in \mathcal{E};
d(\Phi^t(\rho),\Phi^t(\rho'))<\epsilon \text{ for all } \pm t\leq 0 \text{
and} \lim\limits_{t\rightarrow \mp
\infty} d(\Phi^t(\rho'),\Phi^t(\rho))=0\}.
\end{equation*}
Note that $E_\rho^\pm$ is the ($(d-1)$-dimensional) tangent space of $W^\pm_\epsilon$ at $\rho$. 
We also define the weakly unstable manifolds by
\begin{equation*}W_{\epsilon}^{\pm 0} (\rho)=\{\rho'\in \mathcal{E};
d(\Phi^t(\rho),\Phi^t(\rho'))<\epsilon \text{ for all } \pm t\leq 0\}.
\end{equation*}

We call
\begin{equation*}E_\rho^{+0}:= E_\rho^+ \oplus \mathbb{R} H_p(\rho),~~~~ E_\rho^{-0}:=
E_\rho^- \oplus \mathbb{R} H_p(\rho),
\end{equation*}
the \textit{weak unstable} and \textit{weak stable} subspaces at the point
$\rho$ respectively, which are respectively the tangent spaces to $W^{\pm 0}_\epsilon$ at $\rho$.
\subsubsection*{Adapted coordinates}
To state our result concerning the propagation of Lagrangian manifolds, we
need \textit{adapted coordinates} close to the trapped set.  These
coordinates, constructed in \cite[Lemma 2]{Ing}, satisfy the following
properties:

For each $\rho\in K$, we build an adapted
system of symplectic coordinates $(y^\rho,\eta^\rho)$ on a neighbourhood
of $\rho$ in
$T^*X$, such that the following holds.
\[
\begin{aligned}
&(i)~~\rho\equiv (0,0)\\
&(ii)~~E_{\rho}^+ = span \{\frac{\partial}{\partial y^\rho_i}(\rho), ~~~
i=2,...,d\} ,\\
&(iii)~~E_{\rho}^- = span \{\frac{\partial}{\partial \eta_i^\rho}(\rho), ~~~
i=2...d\} ,\\
&(iv)~~ \eta_1^\rho= p-1  \text{ is the energy coordinate},\\
&(v)~~\big{\langle}\frac{\partial}{\partial
y^\rho_i}(\rho),\frac{\partial}{\partial
y^\rho_j}(\rho)\big{\rangle}_{g_{ad}(\rho)}=\delta_{i,j},~~~i,j=2,...,d.
\end{aligned}
\]

For $j=1,...,d-1$, write 
\begin{equation}\label{pavillon}
u_j^\rho := y_{j+1}^\rho \text{, and } s_j^\rho :=
\eta_{j+1}^\rho.
\end{equation}

Let us now introduce unstable Lagrangian manifolds, that is to say, Lagrangian manifolds whose tangent spaces form small angles with the unstable space at $\rho$.

\begin{definition}
Let $\Lambda\subset S^*X$ be an isoenergetic Lagrangian manifold
(not
necessarily connected) included in a small neighbourhood $W$ of a point
$\rho\in K$,
and let $\gamma>0$. We will say that $\Lambda$ is a
$\gamma$\textit{-unstable Lagrangian manifold} (or
that $\Lambda$ is in the $\gamma$-unstable cone) in the coordinates
$(y^\rho,\eta^\rho)$ if it
can be written in the form
\begin{equation*}\Lambda = \{(y_1^\rho,u^\rho,0,F(y_1^\rho,u^\rho)) ; (y^\rho_1,u^\rho)\in D \},
\end{equation*}
where $D\subset \mathbb{R}^d$, is an open subset with finitely many
connected components, and with piecewise smooth boundary, and $F :
\mathbb{R}^{d} \longrightarrow
\mathbb{R}^{d}$ is a smooth function with $\|dF\|_{C^0}\leq \gamma$.
\end{definition}
Let us note that, since $\Lambda$ is isoenergetic and is Lagrangian, an immediate computation shows that $F$ does not depend on $y_1^\rho$, so that $\Lambda$ can actually be put in the form
\begin{equation*}\Lambda = \{(y_1^\rho,u^\rho,0,f(u^\rho)) ; (y^\rho_1,u^\rho)\in D \},
\end{equation*}
where $f :
\mathbb{R}^{d-1} \longrightarrow
\mathbb{R}^{d-1}$ is a smooth function with $\|df\|_{C^0}\leq \gamma$.

Note that, since $f$ is defined on $\mathbb{R}^{d-1}$, a $\gamma$-unstable
manifold may always be seen as a submanifold of a \textit{connected
$\gamma$-unstable Lagrangian manifold}.

\subsubsection{Topological pressure} \label{press}
We shall now give a definition of topological pressure, so as to formulate Hypothesis \ref{Husserl}.
Recall that the distance $d_{ad}$ was defined in section \ref{averse}, and that it was associated to the adapted metric. We say
that a set $\mathcal{S}\subset K$ is $(\epsilon,t)$-separated if for
$\rho_1, \rho_2\in \mathcal{S}$, $\rho_1 \neq \rho_2$, we have
$d_{ad}(\Phi^{t'}(\rho_1),\Phi^{t'}(\rho_2))>\epsilon$ for some $0\leq t \leq
t'$. (Such a set is necessarily finite.)

The metric $g_{ad}$ induces a volume form $\Omega$ on any $d$-dimensional
subspace of $T(T^*\mathbb{R}^d)$. Using this volume form, we will define
the unstable Jacobian on $K$. For any $\rho\in K$, the determinant map
\begin{equation*}\Lambda^n d\Phi^t(\rho)|_{E_\rho^{+0}} : \Lambda^n E_\rho^{+0}
\longrightarrow \Lambda^n E_{\Phi^t(\rho)}^{+0}
\end{equation*}
can be identified with the real number
\begin{equation*}\det\big{(} d\Phi^t(\rho)|_{E_\rho^{+0}}\big{)} :=
\frac{\Omega_{\Phi^t(\rho)}\big{(}d\Phi^tv_1 \wedge
d\Phi^tv_2\wedge...\wedge d\Phi^tv_n\big{)}}{\Omega_\rho(v_1\wedge v_2
\wedge... \wedge v_n)},
\end{equation*}
where $(v_1,...,v_n)$ can be any basis of $E_\rho^{+0}$. This number
defines the unstable Jacobian:
\begin{equation}\label{defJaco}
\exp \lambda^+_t(\rho) := \det\big{(}
d\Phi^t(\rho)|_{E_\rho^{+0}}\big{)}.
\end{equation}
From there, we take
\begin{equation*}Z_t(\epsilon,s):= \sup \limits_{\mathcal{S}} \sum_{\rho \in \mathcal{S}}
\exp(-s\lambda_t^+(\rho)),
\end{equation*}
where the supremum is taken over all $(\epsilon,t)$-separated sets. The
pressure is then defined as
\begin{equation*}\mathcal{P}(s):= \lim \limits_{\epsilon \rightarrow 0} \limsup \limits_{t
\rightarrow \infty} \frac{1}{t} \log  Z_t(\epsilon,s) .
\end{equation*}
This quantity is actually independent of the volume form $\Omega$ and of the metric chosen: after
taking logarithms, a change in $\Omega$ or in the metric will produce a term $O(1)/t$,
which is not relevant in the $t\rightarrow \infty$ limit.
\begin{hyp} \label{Husserl}
We assume the following inequality on the topological pressure associated with $\Phi^t$ on $S^*X$:
\begin{equation} \label{Laurence}
\mathcal{P}(1/2)<0.
\end{equation}
\end{hyp}

\subsubsection{Additional assumptions on the manifold}\label{NewManifold}
In order to obtain stronger results than in \cite{Ing}, we will need the following additional assumptions on $(X,g)$.

\begin{hyp} \label{Strong}
From now on, we will suppose that

(i) $(X,g)$ has nonpositive sectional curvature.

(ii) The sectional curvatures are all bounded from below by some constant $-b_0$, with $b_0\in (0,\infty)$.

(iii) The injectivity radius goes to infinity at infinity, in the following sense : for all sequences of points $(x_n)\subset X$ such that $b(x_n)$ goes to $0$, we have $r_i(x_n)\longrightarrow \infty$.
\end{hyp}

Recall that if $x\in X$, the injectivity radius of $x$, denoted by $r_i(x)$, is the largest number $r\bel 0$ such that the exponential map at $x$ is injective on the open ball $B(0,r)$. On a manifold of nonpositive curvature, saying that $r_i(x)< \infty$ means that there exists $y\in X$ such that $d_X(x,y)= r_i(x)$ and there exist two different unite-speed minimizing geodesics from $x$ to $y$.

Part (iii) of Hypothesis \ref{Strong} implies\footnote{Actually, one can easily show that (\ref{mininj}) is implied by Hypothesis \ref{Guepard} and part (i) of Hypothesis \ref{Strong}.} that
\begin{equation}\label{mininj}
r_i:=\inf\limits_{x\in X} r_i(x) \bel 0.
\end{equation}

\begin{ex}
The manifold represented on Figure \ref{example} is a simple example of a manifold which fulfils Hypothesis
\ref{Strong}.
\end{ex}
\begin{ex}
Any quotient of the hyperbolic space by a convex co-compact group of isometries satisfies Hypothesis \ref{Strong}. If we perturb slightly the metric on a compact set of such a manifold, it will still satisfy Hypothesis \ref{Strong}. Manifolds which are hyperbolic near infinity will be considered in more detail in Appendix \ref{portugal}.
\end{ex}

\subsection{Hypotheses on the distorted plane waves}\label{hypdistorted}

\subsubsection{Hypotheses on the incoming Lagrangian
manifold}\label{scotiabank}
Let us consider an isoenergetic Lagrangian manifold $\Lag\subset
S^*X$ of the form 
$$\Lag:=\{(x,\Dir(x)),x\in
X_1\},$$
where $X_1$ is a closed subset of $X$ with finitely
many connected components and piecewise smooth boundary, and $\Dir :
X_2\ni x\longrightarrow \Dir(x)\in T_x^*X$ is a smooth
co-vector field defined on some neighbourhood $X_2$ of $X_1$.

We make the following additional hypothesis on $\Lag$:

\begin{hyp}[Invariance hypothesis]\label{chaise}
We suppose that $\Lag$ satisfies the following invariance properties.
\begin{equation}\label{invfut}
\forall t \geq 0, \Phi^t(\Lag)\cap \mathcal{DE}_- = \Lag\cap\mathcal{DE}_-.
\end{equation}
\begin{equation}\label{invpass}
\forall t \geq 0, \Phi^{-t}(\Lag)\cap \mathcal{DE}_+ = \Lag\cap\mathcal{DE}_+.
\end{equation}
\end{hyp}

\begin{ex} \label{sega}
Suppose that $(X\backslash X_0,g) \cong (\mathbb{R}^d\backslash B(0,R),g_{Eucl}) $ for
some $R>0$.
Given a $\xi\in \mathbb{R}^d$ with $|\xi|^2=1$, the Lagrangian manifold $$S^*X\supset \Lambda_\xi:=\{(x,\xi); x\notin X_0\}$$ fulfils Hypothesis \ref{chaise}.
\end{ex}

\begin{ex} \label{playstation}
Suppose that $(X\backslash X_0,g) \cong (\mathbb{R}^d\backslash B(0,R),g_{Eucl}) $ for
some $R>0$. Then the incoming spherical Lagrangian, defined by
\begin{equation*}\Lambda_{sph}:= \{ (x,-\frac{x}{|x|}); |x|>R\},
\end{equation*}
fulfills Hypothesis \ref{chaise}.
\end{ex}

We also make the following transversality assumption on the Lagrangian
manifold $\Lag$. It roughly says that $\Lag$ intersects the stable manifold transversally.

\begin{hyp}[Transversality hypothesis]\label{Happy}
We suppose that $\Lag$ is such that, for any $\rho\in K$, for
any $\rho'\in \Lag$, for any $t\geq 0$, we have for $\epsilon\bel 0$ small enough,
\begin{equation*}
\Phi^t(\rho')\in W_{\epsilon}^-(\rho)\Longrightarrow W_{\epsilon}^-(\rho) \text{ and } \Phi^t(\Lag) \text{ intersect transversally at } \Phi^t(\rho'),
\end{equation*} 
that is to say
\begin{equation}\label{telus}
T_{\Phi^t(\rho')}\Lag \oplus T_{\Phi^t(\rho')} W_{\epsilon}^-(\rho) = T_{\Phi^t(\rho')} S^*X.
\end{equation}
\end{hyp}
Note that (\ref{telus}) is equivalent to $T_{\Phi^t(\rho')}\Lag \cap T_{\Phi^t(\rho')} W_{\epsilon}^-(\rho) = \{0\}$.

In general, this assumption is not easy to check. However, we will show in Proposition \ref{musaraigne} that it is always satisfied if the hypotheses of sections \ref{NewManifold} and \ref{NewDistorted} are satisfied.

\subsubsection{Assumptions on the generalized eigenfunctions}
We consider a family of smooth functions $E_h\in C^\infty(X)$ indexed by $h\in(0,1]$ which satisfy
\begin{equation*}
(P_h -1) E_h=0,
\end{equation*} 
where
\begin{equation*}
P_h=-h^2\Delta-c_0 h^2.
\end{equation*}
Here, $c_0\bel 0$ is a constant which is equal to 0 in the case of Euclidean near infinity manifolds, and to $(d-1)^2/4$ on manifolds that are hyperbolic near infinity.

We will furthermore assume that these generalized eigenfunctions may be decomposed as follows. For the definitions of Lagrangian states, tempered distributions and wave-front sets, we refer the reader to Appendix \ref{pivoine}

\begin{hyp}\label{pied}
We suppose that $E_h$ can be put in the form
\begin{equation}\label{Poitiers}
E_h= E_h^0+E_h^1,
\end{equation}
where $E_h^0$ is a Lagrangian state associated to a Lagrangian manifold $\Lag$ which satisfies Hypothesis \ref{chaise} of invariance, as well as Hypothesis \ref{Happy} of transversality, and where $E_h^1$ is a tempered distribution such that
for each $\rho\in WF_h(E_h^1)$, we have $\rho\in S^*X$.

Furthermore, we suppose that $E_h^1$ is \emph{outgoing} in the sense that there exists $\epsilon_2\bel 0$ such that for all $(x,\xi)\in T^*X$ such that $b(x)<\epsilon_2$, we have
\begin{equation}\label{outgoing4}
\rho\in WF_h (E_h^1) \Rightarrow \rho\in\mathcal{DE}_+.
\end{equation}
\end{hyp}

\begin{remarque}\label{remarqueoutgoing}
Note that (\ref{outgoing4}) implies that for any $\chi\in C_c^\infty(X)$, we may find $\hat{\chi}\in C_c^\infty(X)$ such that $\hat{\chi}\equiv 1$ on $\mathrm{supp}(\chi)\cup \{x\in X; b(x)\geq \epsilon_2\}$, and such that the support of $\hat{\chi}(1-\hat{\chi})$ is small enough so that for any $t\geq 1$, we have
\begin{equation}\label{outgoing}
\Phi^t\Big{(}WF_h\big{(}(1-\hat{\chi}) E_h^1\big{)}\Big{)} \cap T^*\mathrm{supp} (\hat{\chi}) = \emptyset.
\end{equation}
We will often use this consequence of (\ref{outgoing4}).
\end{remarque}

\begin{ex}
In \cite{Ing}, it is explained how distorted plane waves enter in this framework on manifold that are Euclidean near infinity. In Appendix \ref{portugal}, we will show that distorted plane waves on manifolds that are hyperbolic near infinity do also satisfy this assumption.
\end{ex}

\subsubsection{Additional assumptions on the Lagrangian manifold $\Lag$}\label{NewDistorted}
From now on, we will denote by $d_X$ the Riemannian distance on the base manifold. It should not be confused with the distance $d_{ad}$ on the energy layer which was introduced in section \ref{averse}, and which we will sometimes use too. If $\rho,\rho'\in T^*X$, we will write $d_X(\rho,\rho')$ for $d(\pi_X(\rho),\pi_X(\rho'))$, where $\pi_X$ denotes the projection on the base manifold.

We will need the following assumptions on the incoming Lagrangian manifold $\Lag$. First of all, we require that $\Lag$ does not expand when propagated in the past.

\begin{hyp} \label{puissanceplus}
We suppose that $\Lag$ is such that $\forall \rho_1, \rho_2\in \Lag, \forall t\leq 0, d_X
(\Phi^t(\rho_1),\Phi^t(\rho_2))\leq
d_X (\rho_1,\rho_2)$.
\end{hyp}

\begin{ex}
Without Hypothesis \ref{puissanceplus}, Theorem \ref{linfini} might not be satisfied. For example, on $X=\mathbb{R}^d$
the Lagrangian manifold $\Lambda_{circ}$ from Example \ref{playstation} does
not satisfy part of Hypothesis \ref{puissanceplus}.

Consider $E_h(x):=\frac{1}{h^{(d-1)/2}}\int_{\mathbb{S}^{d-1}} e^{i\theta\cdot x/h}\mathrm{d}\theta$. One can show, by stationary phase (see \cite[\S 2]{Mel}) that 
\begin{equation}\label{orenci}
E_h(x)= \frac{C}{|x|^{(d-1)/2}} \big{(}e^{ir/h}+ e^{-ir/h} \big{)}+ R_h,
\end{equation} where $R_h(x)$ goes to zero when $h/|x|$ goes to zero.

If $\chi\in C^\infty_c(\mathbb{R}^d)$ is such that $\chi(x)=1$ for $|x|\leq R$ for some $R\bel 0$, then $E_h^0(x):= (1-\chi(x))e^{i|x|/r}$ is a Lagrangian state associated to $\Lambda_{circ}$. $E_h^1:= E_h-E_h^0$ is a tempered distribution, and Hypothesis \ref{pied} is satisfied. However, $E_h(0)=  \frac{C}{h^{(d-1)/2}}$, so that $E_h$ is not bounded in $L^\infty$ independently of $h$.

Therefore, Hypothesis \ref{puissanceplus} is essential for Theorem \ref{linfini} to hold.
\end{ex}

For Theorem \ref{smallscale7}, we also require a sort of completeness assumption for $\Lag$, which is as follows. Note that Hypothesis \ref{puissanceplus2} is not required for Theorem \ref{linfini} to hold.

\begin{hyp}\label{puissanceplus2}
We suppose that $\Lag$ is such that for all $\rho\in \mathcal{DE}_-$, we have
\begin{equation*} \Big{[}\exists \rho'\in \Lag, \forall t\leq 0, d_X
(\Phi^t(\rho),\Phi^t(\rho'))\leq
d_X (\rho,\rho')\Big{]}\Longrightarrow \rho\in \Lag.
\end{equation*}
\end{hyp}

Beware that if $\Lag$ satisfies Hypothesis \ref{puissanceplus2}, a subset of $\Lag$ may not satisfy Hypothesis \ref{puissanceplus2}, even if it satisfies the invariance property of Hypothesis \ref{chaise}.

\section{Main results}
In this section, we state our main results concerning distorted plane waves on manifolds of nonpositive curvature. Before doing so, we recall the main results of \cite{Ing}, so as to introduce some useful notations, and since we will need them in the proofs in sections \ref{preuveclassique} and \ref{preuvequantique}.
\subsection{Recall of the main results from \cite{Ing}}
Let us recall the main result from \cite{Ing}. The definitions of pseudo-differential operators and of Fourier integral operators are recalled in appendix \ref{pivoine}.

\begin{theoreme}\label{ibrahim}
 Suppose that the manifold $(X,g)$ satisfies Hypothesis \ref{Guepard} near infinity, and that the geodesic flow $(\Phi^t)$ satisfies
Hypothesis \ref{sieste} on hyperbolicity and Hypothesis \ref{Husserl} concerning the topological pressure. Let $E_h$ be a generalized eigenfunction of the form described in Hypothesis \ref{pied}, where $E_h^0$ is associated to a Lagrangian manifold $\Lag$ which satisfies
the invariance Hypothesis \ref{chaise} as well as the transversality
Hypothesis \ref{Happy}.
 
Then there exists a finite set of points $(\rho_b)_{b\in B_1}\subset K$ and a family $(\Pi_b)_{b\in B_1}$ of operators in $\Psi^{comp}_h(X)$ microsupported in a small neighbourhood of $\rho_b$ such that $\sum_{b\in B_1} \Pi_b = I$ microlocally on a neighbourhood of $K$ in $T^*X$ such that the following holds.

Let $\mathcal{U}_b: L^2(X)\longrightarrow L^2(\mathbb{R}^d)$ be a Fourier integral operator quantizing the symplectic
change of local coordinates $\kappa_b:(x,\xi) \mapsto
(y^{\rho_b},\eta^{\rho_b})$, and which is microlocally unitary on the microsupport of $\Pi_b$.

For any $r>0$ and $\ell\in \mathbb{N}$, there exists $M_{r,\ell}>0$ such that
 we have as $h\rightarrow 0$:
\begin{equation}\label{greenday}\mathcal{U}_b \Pi_b E_h(y^{\rho_b}) = \sum_{n=0}^{\lfloor
M_{r,\ell} |\log h|\rfloor}
\sum_{\beta\in \mathcal{B}_n} e^{i \phi_{\beta,b}(y^{\rho_b})/h}
a_{\beta,b}(y^{\rho_b};h) + R_{r,\ell},
\end{equation}
where the $a_{\beta,b}\in S^{comp}(\mathbb{R}^d)$ are classical symbols in the sense of Definition \ref{defsymbclassique}, and each $\phi_{\beta,b}$ is a smooth function independent of $h$, and defined in a neighbourhood of the support of
$a_{\beta,b}$. Here, $\mathcal{B}_n$ is a set of words with length close to $n$; hence  its cardinal behaves like some exponential of $n$.

We have the following estimate on the remainder
 \begin{equation*}\|R_{r,\ell}\|_{C^\ell}=O(h^r).
\end{equation*}

 For any $\ell\in \mathbb{N}$, $\epsilon>0$, there exists $C_{\ell,\epsilon}$
such that for all $n\geq 0$, for all $h\in (0,h_0]$, we have
\begin{equation}\label{sheriff}
\sum_{\beta\in \mathcal{B}_n} \|a_{\beta,b}\|_{C^\ell} \leq C_{\ell,\epsilon}
e^{n(\mathcal{P}(1/2)+\epsilon)}.
\end{equation}
\end{theoreme}

Let us recall in a very sketchy way the idea behind the proof of Theorem \ref{ibrahim}. Since $(-h^2\Delta-1)E_h=0$, we have formally that $e^{-itP_h/h} E_h= e^{-it/h}E_h$, where $e^{-itP_h/h}$ is the Schrödinger propagator. Of course, this statement can only be formal, since $E_h$ is not in $L^2$, but by working with cut-off functions, we can make it rigorous up to a $O(h^\infty)$ remainder. 

By using some resolvent estimates and hyperbolic dispersion estimates, one can show that if we propagate $E_h$ by the Schrödinger flow during a long enough time (of the order of some logarithm of $h$), the term involving $E_h^1$ becomes smaller than any power of $h$. Hence we only have to study the propagation during long times of $E_h^0$. Since, by assumption, $E_h^0$ is a Lagrangian state, we can use the WKB method to study its propagation. The main part in the WKB analysis is to understand the Lagrangian manifold $\Phi^t(\Lag)$, especially for large values of $t$. This is the content of Theorem \ref{Cyril} below. Before stating it, we recall a few notations.

Let us fix 
\begin{equation}\label{defgamma}
\ins>0
\end{equation}
 small enough. In \cite{Ing}, we built $(V_b)_{b\in B}$ a finite open cover of $S^*X$ in $T^*X$ (depending on $\ins$) such that Theorem \ref{Cyril} below holds. This open cover had the following properties.
\begin{itemize}
\item We have $B= B_1\sqcup B_2\sqcup \{0\}$, where $V_0$ is as in (\ref{frite}).

\item For each $b\in B_1$, $(V_b)_{b\in B_1}$ is an open cover of $K$ in $T^*X$, such that for every $b\in B_1$, there exists a point $\rho^b\in V_b\cap K$, and
such that the adapted coordinates $(y^b,\eta^b)$ centred on $\rho^b$ are well defined on $V_b$
for every $b\in B_1$. 

\item The sets $V_b$ are all bounded for $b\in B_1\sqcup B_2$.
\end{itemize}

\paragraph{Truncated Lagrangians}
Let
$N\in\mathbb{N}$, and
let $\beta=\beta_0,\beta_1...\beta_{N-1}\in B^{N}$.
Let $\Lambda$ be a Lagrangian manifold
in $T^*X$. We define the sequence of (possibly empty) Lagrangian manifolds
$(\Phi_\beta^{k}(\Lambda))_{0\leq k \leq N-1}$ by recurrence by:
\begin{equation}\label{deftronque}
\Phi_\beta^{0}(\Lambda)= \Lambda \cap V_{\beta_0},~~~~
\Phi_\beta^{k+1}(\Lambda) =
V_{\beta_{k+1}} \cap  \Phi^1 (\Phi_{\beta}^{k}
(\Lambda)).
\end{equation}

If $\beta\in B^N$, we will often write 
\begin{equation*}
\Phi_\beta (\Lambda) := \Phi_\beta^{N-1}(\Lambda). 
\end{equation*}

For any $\beta\in B^{N}$ such that $\beta_{N-1}\neq 0$, we will define
\begin{equation}\label{rollingstone}
\tau(\beta):= \max \{1\leq i \leq N-1; \beta_i=0\}
\end{equation}
 if there exists 
$1\leq i \leq N-1$ with $\beta_i=0$, and $\tau(\beta)=0$ otherwise.

The sets $\Phi^N_\beta(\Lag)$ are related to the result of Theorem \ref{ibrahim} as follows. For all $N\in \mathbb{N}$, $\beta\in B^N$ and $b\in B_1$, we have, using the notations of Theorem \ref{ibrahim}:
\begin{equation}\label{calamar}
\kappa_b\big{(} \Phi^N_\beta(\Lag)\big{)}\supset \{ (y^{\rho_b},\partial \phi_{\beta,b}(y^{\rho_b})); y^{\rho_b}\in \Omega_{\beta,b}\},
\end{equation}
where $\Omega_{\beta,b}\subset \mathbb{R}^d$ is an open set containing the support of $a_{\beta,b}$.
The properties of the sets $\Phi^N_\beta(\Lag)$ are described in the following theorem.

\begin{theoreme} \label{Cyril}
Suppose that, the manifold $(X,g)$ satisfies Hypothesis
\ref{Guepard} at infinity, that the Hamiltonian flow $(\Phi^t)$ satisfies
Hypothesis \ref{sieste}, and that the Lagrangian manifold $\Lag$ satisfies
the invariance Hypothesis \ref{chaise} as well as the transversality
Hypothesis \ref{Happy}.

There exists $\tins\in \mathbb{N}$ such that for all
$N\in
\mathbb{N}$, for all $\beta\in B^{N}$ and all $b\in B_1$, then
$V_b\cap\Phi_\beta(\Lag)$ is either empty, or can be written in the coordinates $(y^{b},\eta^{b})$ as

\begin{equation*}V_b\cap\Phi_\beta(\Lag) = \{(y_1^b,u^b,0,f_{b,\beta}(u^b)) ; (y^b_1,u^b)\in D_{b,\beta} \},
\end{equation*}
where $f_{b,\beta} :
\mathbb{R}^{d-1} \longrightarrow
\mathbb{R}^{d-1}$ is a smooth function. and $D_{b,\beta}\subset \mathbb{R}^d$ is a bounded open set.

For each $\alpha\in \mathbb{N}^{d-1}$, there exists a constant $C_\alpha$ such that for all $N\in
\mathbb{N}$, for all $\beta\in B^{N}$ and all $b\in B_1$, we have
\begin{equation*}
\|\partial^\alpha f_{b,\beta}\|_{C^0}\leq C_\alpha.
\end{equation*}

Furthermore, if $N-\tau(\beta)\geq \tins$, then $\|d f_{b,\beta}\|_{C^0}\leq \ins$.
\end{theoreme}

\subsection{New results in nonpositive curvature}\label{newyork}
The results of \cite{Ing} can be improved in the case of geometric scattering in
nonpositive sectional curvature, for Lagrangian manifolds that are “non-expanding in the past”, as we shall now describe.

\subsubsection{Results on the propagation of $\Lag$}
The first consequence of Hypotheses \ref{Strong} and \ref{puissanceplus} is the following lemma, which guarantees that Hypothesis \ref{Happy} concerning transversality is always satisfied.
\begin{proposition}\label{musaraigne}
Suppose that $(X,g)$ satisfies Hypothesis \ref{Guepard} near infinity, Hypothesis \ref{Strong}, as well as Hypothesis \ref{sieste} on hyperbolicity, and that $\Lag$ is a Lagrangian manifold which satisfies Hypothesis \ref{chaise} of invariance, as well as Hypothesis \ref{puissanceplus}. Then $\Lag$ satisfies Hypothesis \ref{Happy} on transversality.
\end{proposition}

To state our main result concerning the propagation of $\Lag$, which is an improvement on Theorem \ref{Cyril}, we need the following definition.
\begin{definition}\label{smoothproj}
If $X'$ is a $d$-dimensional submanifold of $T^*X$, we shall say that $X'$ \emph{projects smoothly} on $X$ if it is contained in a smooth section of $T^*X$. That is to say, $X'$ can be written in the form
\begin{equation}\label{lisse}
X'= \{(x, f(x)), x\in \Omega\},
\end{equation}
where $\Omega$ is an open subset of $X$, and $f$ is smooth.
\end{definition}
\begin{theoreme}\label{getup}
Suppose that $(X,g)$ satisfies Hypothesis \ref{Guepard} near infinity, Hypothesis \ref{Strong}, as well as Hypothesis \ref{sieste} on hyperbolicity, and that $\Lag$ is a Lagrangian manifold which satisfies Hypothesis \ref{chaise} of invariance, as well as Hypothesis \ref{puissanceplus}. Then there exists a $\tilde{\gamma}\bel 0$ such that, if we take $\ins\leq \tilde{\gamma}$ in (\ref{defgamma}), the following holds.

Let $\mathcal{O}\subset X$ be an open set which is small enough so that we may define local coordinates on it.
Then for any $N\in \mathbb{N}$ and any $\beta\in B^N$, $\Phi_\beta(\Lag)\cap (S^*\mathcal{O})$ is a Lagrangian manifold which may be projected smoothly on $X$. In particular, in local coordinates, the manifold $\Phi^N_\beta(\Lag)\cap T^*\mathcal{O}$ may be written in the form
\begin{equation*}
\Phi^N_\beta(\Lag)\cap T^*\mathcal{O}\equiv \{(x,\partial_x \varphi_{\beta, \mathcal{O}} (x)); x\in \mathcal{O}^\beta\},
\end{equation*}
where $\mathcal{O}^\beta$ is an open subset of $\mathcal{O}$.

Furthermore, for any $\ell\in \mathbb{N}$, there exists a $C_{\ell,\mathcal{O}}\bel 0$ such that for any $N\in \mathbb{N}$, $\beta\in B^N$, we have
\begin{equation}\label{bogota}
\|\partial_x \varphi_{\beta,\mathcal{O}}\|_{C^\ell}\leq C_{\ell,\mathcal{O}}.
\end{equation}
\end{theoreme}

\subsubsection{Quantum results}
\begin{theoreme}\label{ibrahim2}
Let $X$ be a manifold which is Euclidean or hyperbolic near infinty, and which satisfies Hypothesis \ref{Strong}. Suppose that the geodesic flow $(\Phi^t)$ satisfies
Hypothesis \ref{sieste} on hyperbolicity, Hypothesis \ref{Husserl} concerning the topological pressure. Let $E_h$ be a generalized eigenfunction of the form described in Hypothesis \ref{pied}, where $E_h^0$ is associated to a Lagrangian manifold $\Lag$ which satisfies
the invariance Hypothesis \ref{chaise} as well as
Hypothesis \ref{puissanceplus}.

Let $\mathcal{K}\subset X$ be compact. There exists $\varepsilon_\mathcal{K}\bel 0$ such that for any  $\chi\in C_c^\infty(X)$ with a support in $\mathcal{K}$ of diameter smaller than $\varepsilon_{\mathcal{K}}$, the following holds.
There exists a set $\tilde{\mathcal{B}}^\chi$ and a function $\tilde{n}: \tilde{\mathcal{B}}^\chi\rightarrow \mathbb{N}$ such that the number of elements in $\{\tilde{\beta}\in\tilde{\mathcal{B}}^\chi; \tilde{n}(\tilde{\beta})\leq N\}$ grows at most exponentially with $N$.

For any $r>0$, $\ell>0$, there exists $\tilde{M}_{r,\ell}>0$ such that
 we have as $h\rightarrow 0$:
\begin{equation}\label{qad2} \chi E_h(x) = \sum_{\substack{\tilde{\beta}\in \tilde{\mathcal{B}}^\chi\\
                \tilde{n}(\tilde{\beta})\leq \tilde{M}_{r,\ell}|\log h|}} e^{i \varphi_{\tilde{\beta}}(x)/h}
a_{\tilde{\beta}}(x;h) + R_{r,\ell},
\end{equation}
where $a_{\tilde{\beta}}\in S^{comp}(X)$ is a classical symbol in the sense of Definition \ref{defsymbclassique}, and each $\varphi_{\tilde{\beta}}$ is a smooth function defined in a neighbourhood of the support of
$a_{\tilde{\beta}}$. 
We have
 \begin{equation*}\|R_{r,\ell}\|_{C^\ell}=O(h^r).
\end{equation*}

 For any $\ell\in \mathbb{N}$, $\epsilon>0$, there exists $C_{\ell,\epsilon}$
such that
\begin{equation}\label{sheriff3}
\sum_{\substack{\tilde{\beta}\in \tilde{\mathcal{B}}^\chi\\ \tilde{n}(\tilde{\beta})=n}} \|a_{\tilde{\beta}}\|_{C^\ell} \leq C_{\ell,\epsilon}
e^{n(\mathcal{P}(1/2)+\epsilon)}.
\end{equation}

Furthermore, there exists a constant $C_1$ such that for all $\tilde{\beta}, \tilde{\beta}'\in \tilde{\mathcal{B}}^\chi$, we have
\begin{equation}\label{hurry2}
|\partial \varphi_{\tilde{\beta}} (x)- \partial \varphi_{\tilde{\beta'}}(x)|\geq C_1 e^{-\sqrt{b_0} \max(\tilde{n}(\tilde{\beta}),\tilde{n}(\tilde{\beta'}))}, 
\end{equation}
where $b_0$ is as in Hypothesis \ref{Strong}.
\end{theoreme}

The link between this theorem and the previous one is as follows: let $\chi\in C_c^\infty(X)$ be as in the theorem and $\mathcal{O}$ be a small open set such that
 $\spt(\chi)\cap \mathcal{O}\neq \emptyset$. As defined in section \ref{regrouping}, the set $\tilde{\mathcal{B}}^\chi$ is a set of equivalence classes of a subset of $\cup_{n\in\mathbb{N}}B^{n}$. 
 Let $\tilde{\beta}\in \tilde{\mathcal{B}}^\chi$ and let $\beta\in B^n$ be a representative of $\tilde{\beta}$. 
 We may consider $\Phi_\beta^{\tilde{n}}(\Lag)$. We have:

\begin{equation*}
\forall x\in \spt(\varphi_{\beta,\mathcal{O}})\cap \spt(\varphi_{\tilde{\beta}}),~~ \partial \varphi_{\beta,\mathcal{O}}(x)= \partial \varphi_{\tilde{\beta}}(x).
\end{equation*}

Therefore, locally, the gradient of the phases described in Theorem \ref{getup} and \ref{ibrahim2} are the same.

\begin{remarque}
Note that in Theorem \ref{ibrahim2}, the assumption that $\chi$ has a small support is important only to obtain (\ref{hurry2}). If $\chi$ is any function in $C^\infty_c(X)$, we may use Theorem \ref{ibrahim2} combined with a partition of unity argument to write $\chi E_h$ as a decomposition similar to (\ref{qad2}), with an estimate as in (\ref{sheriff3}). Actually, this will be done in a more direct way in the proof of Theorem \ref{ibrahim2} (see (\ref{faust}) and the discussion which follows).
\end{remarque}

As a corollary of Theorem \ref{ibrahim2}, we may deduce the following generalisation of Theorem \ref{linfini}.
\begin{corolaire}\label{clocks}
We make the same hypotheses as in Theorem \ref{ibrahim2}. Let $\ell\in \mathbb{N}$ and $\chi\in C_c^\infty(X)$. Then there exists $C_{\ell,\chi}\bel 0$ such that, for any $h\bel 0$, we have
\begin{equation*}
\|\chi E_h\|_{C^\ell}\leq \frac{C_{\ell,\chi}}{h^\ell}.
\end{equation*}
\end{corolaire}
In particular, the sequence $(E_h)_h$ is uniformly bounded with respect to $h$ in $L^\infty_{loc}$.
\begin{proof}
The corollary follows from the decomposition (\ref{qad2}) along with the estimates (\ref{sheriff3}) and  (\ref{bogota}).
\end{proof}

\begin{corolaire} \label{blacksabbath3}
We make the same hypotheses as in Theorem \ref{ibrahim2}.
Let $\chi\in C_c^\infty(X)$ and let $\epsilon\bel 0$.
Then there exists a finite measure $\mu_\chi$ on $S^*X$ such that we have for any $\psi\in S^{comp}(S^*X)$
\begin{equation*}\langle Op_h(\psi) \chi E_h, \chi E_h\rangle = \int_{T^*X} \psi(x,\xi)
\mathrm{d}\mu_{\chi}(x,\xi) +
O\Big{(}h^{\min\big{(}1, \frac{|\mathcal{P}(1/2)|}{2|\sqrt{b_0}}-\epsilon \big{)}}\Big{)},
\end{equation*}
where $-b_0$ is the minimal value taken by the sectional curvature on $X$.

If $\mathcal{K}$ is a compact set and if the support of $\chi$ is in $\mathcal{K}$ and of diameter smaller than $\varepsilon_\mathcal{K}$, we have
\begin{equation*}\mathrm{d}\mu_{\chi}(x,\xi) = \sum_{\tilde{\beta}\in
\tilde{\mathcal{B}}^\chi} | (a_{\tilde{\beta}}^0)|^2(x) \delta_{\{\xi=\partial
\varphi_{\tilde{\beta}}(x)\}} \mathrm{d} x,
\end{equation*}
where $a_{\tilde{\beta}}$ is as in (\ref{qad2}), and $a_{\tilde{\beta}}^0$ is its principal symbol as defined in Definition \ref{defsymbclassique}.

Furthermore, if $\Lag$ satisfies Hypothesis \ref{puissanceplus2}, then for every $N\in \mathbb{N}$, there exists $c_N\bel 0$ such that for any $x\in X$ such that $\chi(x)=1$, we have
\begin{equation}\label{beatles}
\sum_{\substack{\tilde{\beta}\in
\tilde{\mathcal{B}}^\chi\\
\tilde{n}(\tilde{\beta})\geq N}}  |a_{\tilde{\beta}}^0|^2(x) \geq c_N.
\end{equation}

\end{corolaire}
We will prove this corollary in section \ref{warpigs}.

Let us finally state a generalization of Theorem \ref{smallscale7}. We will prove it in section \ref{preuveequidpetite}.
\begin{theoreme}\label{nuque}
Let $X$ be a manifold which satisfies Hypothesis \ref{Guepard} near infinity, and which satisfies Hypothesis \ref{Strong}. Suppose that the geodesic flow $(\Phi^t)$ satisfies
Hypothesis \ref{sieste} on hyperbolicity, Hypothesis \ref{Husserl} concerning the topological pressure. Let $E_h$ be a generalized eigenfunction of the form described in Hypothesis \ref{pied}, where $E_h^0$ is associated to a Lagrangian manifold $\Lag$ which satisfies
the invariance Hypothesis \ref{chaise}, part (iii) of
Hypothesis \ref{Strong} and Hypotheses \ref{puissanceplus} and \ref{puissanceplus2}.

Let $\chi\in C_c^\infty(X)$. Then there exist constants $C, C_1, C_2\bel 0$ such that the following result holds. For all $x_0\in X$ such that $\chi(x_0)=1$, for any sequence $r_h$ such that $1\bel \bel r_h \bel C h$, we have for $h$ small enough:
\begin{equation*}
C_1 r_h^d\leq \int_{B(x_0,r_h)} |\Re E_h|^2(x) \mathrm{d}x\leq C_2 r_h^d.
\end{equation*}

In particular, for any bounded open set $U\subset X$, there exists $c(U)\bel 0$ and $h_U\bel 0$ such that for all $0<h<h_U$, we have
\begin{equation*}
\int_U |\Re E_h|^2\geq c(U).
\end{equation*}
\end{theoreme}

\section{Proofs of the results concerning the propagation of $\Lag$.}\label{preuveclassique}
\subsection{General facts concerning manifolds of nonpositive curvature}
\subsubsection{Growth of the distance between points on manifolds of nonpositive curvature.}
In this paragraph, we will recall a few facts about the way the distances $d_X$ and $d_{ad}$ between $\Phi^t(\rho_1)$ and $\Phi^t(\rho_2)$ depend on time. Remember that the distance $d_{ad}$ was introduced in section \ref{averse}, while $d_X$ was defined in section \ref{NewDistorted}.

The easiest bound simply comes from the compactness of $S^*X_0$, and the geodesic convexity of $X_0$: we may find a constant $\mu\bel 0$ such that for any $\rho,\rho'\in S^*X_0$ and for any $t\geq 0$ such that $\Phi^t(\rho),\Phi^t(\rho')\in S^*X_0$, we have
\begin{equation}\label{flowers}
d_{ad}(\Phi^t(\rho),\Phi^t(\rho'))\leq e^{\mu t} d_{ad}(\rho,\rho').
\end{equation}

In all the sequel, we will shrink the sets $(V_b)_{b\in B_2}$ appearing in Theorems \ref{ibrahim} and \ref{Cyril} so that the following holds: the sets $V_b$, $b\in B_1\cup B_2$ have a
diameter smaller than some constant $\epsilon_{max}$ such that
\begin{equation} \label{conditionepsilon}
\forall x,y\in X_0, ~~ d_{ad}(x,y)< \epsilon_{max} ~~ \implies  d_X(\Phi^{1}(x), \Phi^{1}(y)) < e^{-\mu} r_i,
\end{equation}
where $\mu$ is as in (\ref{flowers}).

\begin{remarque}\label{lenovo}
Actually, when we work close to the trapped set, using the bounds on the growth of Jacobi fields which may be found in \cite[III.B]{Eber}, one can show that there exists $C\bel 0$ such that if $\rho,\rho'\in S^*X$ and $T\geq 0$ are such that for all $t\in [0,T]$, we have $\Phi^t(\rho)\in \bigcup_{b\in B_1} V_b$ and $\Phi^t(\rho')\in \bigcup_{b\in B_1} V_b$, then we have
\begin{equation*}
d_{ad}(\Phi^t(\rho),\Phi^t(\rho'))\leq C e^{\sqrt{b_0} t} d_{X}(\rho,\rho'),
\end{equation*}
where $b_0$ is the lowest value taken by the sectional curvature on $X$ as in Hypothesis \ref{Strong}.
\end{remarque}

On the other hand, if two points are on the same local stable manifold, they will approach each other exponentially fast in the future. This is the point of the following classical lemma, whose proof can be found in \cite[Theorem 17.4.3 (3)]{KH}.

\begin{lemme}\label{fix}
There exists $C',\lambda\bel 0$ such that for any $\rho\in K$ and $\rho_1,\rho_2\in W_{\epsilon}^-(\rho)$ for some $\epsilon\bel 0$ small enough, we have
\begin{equation*}
d_{ad}(\Phi^t(\rho_1),\Phi^t(\rho_2))\leq C' e^{-\lambda t} d_{ad}(\rho_1,\rho_2).
\end{equation*}
\end{lemme}

On a manifold of nonpositive curvature, the square of the distance $d_X$ between two points will be convex with respect to time, as long as these two points remain close enough to each other. This is the content of the following lemma, whose proof may be found in \cite[\S 4.8]{jost}.

\begin{lemme}\label{dompteur}
Let $\gamma_1,\gamma_2$ be two geodesics on a manifold of nonpositive sectional curvature $X$ which is simply connected. Then $t\mapsto d^2_X(\gamma(t),\gamma'(t))$ is a convex function.

Furthermore, if $\gamma_1$ and $\gamma_2$ are different geodesics, there exist $-\infty\leq t_1<t_2 \leq +\infty$ such that for all $t\in (t_1,t_2)$, $\gamma_1(t)$ and $\gamma_2(t)$ belong to a region of $X$ where sectional curvature is strictly negative, then on $(t_1,t_2)$, $t\mapsto d^2_X(\gamma_1(t),\gamma_2(t))$ is strictly convex.
\end{lemme}

From now on, we will denote by $\tilde{X}$ the universal covering of $X$.

\begin{corolaire}\label{dompteur2}
Suppose that $\rho_1, \rho_2\in X$ and $-\infty\leq t_1<t_2\leq + \infty$ are such that for all $t\in (t_1,t_2)$, we have $d_X(\Phi^t(\rho_1),\Phi^t(\rho_2))< r_i$. Then $(t_1,t_2)\ni t\mapsto d^2_X(\Phi^t(\rho_1),\Phi^t(\rho_2))$ is a convex function.

Furthermore, if $\rho_1$ and $\rho_2$ do not belong to the same geodesic and if for all $t\in (t_1,t_2)$, $\Phi^t(\rho_1)$ and $\Phi^t(\rho_2)$ belong to a region of $X$ where sectional curvature is strictly negative, then on $(t_1,t_2)$, $t\mapsto d^2_X(\Phi^t(\rho_1),\Phi^t(\rho_2))$ is strictly convex.

\end{corolaire}
\begin{proof}
Using the fact that the exponential map is a covering map on balls of radius smaller than the injectivity radius, we can push back the curves $\Phi^t(\rho_i)$, $i=1,2$ to geodesics $\tilde{\Phi}^t(\tilde{\rho}_i)$ on the universal cover $\tilde{X}$ of $X$ such that $d_{\tilde{X}}(\tilde{\Phi}^t(\tilde{\rho}_1), \tilde{\Phi}^t(\tilde{\rho}_2)) = d_X(\Phi^t(\rho_1),\Phi^t(\rho_2))$. We may then conclude by Lemma \ref{dompteur}.
\end{proof}

\subsubsection{Covering of $X$ by $\bigcup_{t\bel 0}\Phi^t(\Lag)$}
The following lemma can be found in \cite[\S IV.A]{Eber}.
\begin{lemme}\label{piscine}
Let $X$ satisfy Hypothesis \ref{Strong}, and let $\rho\in T^*\tilde{X}$ and $x\in \tilde{X}$. Then there exists a unique $\xi\in T^*_x\tilde{X}$ such that $d_{\tilde{X}}(\Phi^t(\rho),\Phi^t(x,\xi))$ remains bounded as $t\rightarrow - \infty$. In particular, by Lemma \ref{dompteur}, $t\mapsto d_{\tilde{X}} (\Phi^t(\rho),\Phi^t(x,\xi))$ is non-decreasing.
\end{lemme}
Consider the set
\begin{equation*}
\Phi^\infty(\Lag):= \bigcup_{t\bel 0} \Phi^t(\Lag).
\end{equation*}

The following corollary of Lemma \ref{piscine} says that if Hypothesis \ref{puissanceplus2} is satisfied, then $X$ is covered infinitely many times by $\Phi^\infty(\Lag)$. Note that this is the only place where we actually need Hypothesis \ref{puissanceplus2} to hold.
\begin{corolaire}\label{musso}
Let $X$ be a manifold satisfying Hypothesis \ref{Guepard} near infinity and Hypothesis \ref{Strong} of nonpositive curvature, and such that the trapped set $K$ is non-empty and satisfies Hypothesis \ref{sieste} of hyperbolicity.

Let $\Lag$ be a Lagrangian manifold which satisfies Hypotheses \ref{chaise}, \ref{puissanceplus} and \ref{puissanceplus2}.

Let $x\in X$. Then there exist infinitely many $\xi \in T^*_x X$ such that $(x,\xi)\in \Phi^\infty(\Lag)$.
\end{corolaire}
\begin{proof}
Fix some $(x_0,\xi_0)=\rho_0\in \Lag\cap \mathcal{DE}_-$ and some $x\in X$.

Since $K$ is non-empty and hyperbolic, it contains at least a closed orbit. Therefore, $\pi_1(X)$ is not trivial, hence infinite. (Indeed, by \cite[III.G]{Eber}, any non trivial element in the fundamental group of a manifold of nonpositive curvature has infinite order.) Therefore,
$\rho_0$ has infinitely many pre-images by the projection $\widetilde{S^*X} \cong S^*\tilde{X}\rightarrow S^*X$.

 Let us denote them by $(\tilde{\rho}_0^i)_{i\in I}=(\tilde{x}_0^i,\tilde{\xi}_0^i)_{i\in I}$, and let us fix a point $\tilde{x}$ such that $\pi_{\tilde{X}\rightarrow X} (\tilde{x})=x$.
 
Thanks to Lemma \ref{piscine}, for each $i\in I$, there is a $\tilde{\xi}_i \in S^*_{\tilde{x}}\tilde{X}$ such that $d_{\tilde{X}}(\Phi^t(\tilde{\rho}^0_i),\Phi^t(\tilde{x},\tilde{\xi}_i))$ remains bounded as $t\rightarrow - \infty$. Let us write 
$(x,\xi_i) = \pi_{T^*\tilde{X}\rightarrow T^*X}(\tilde{x},\tilde{\xi}_i).$ (The $\xi_i$ and $\tilde{\xi}_i$ may of course be identified.)

For each $t\leq 0$, we have
\begin{equation*}
d_{X}(\Phi^t(\rho_0),\Phi^t(x,\xi_i)) \leq d_{\tilde{X}}(\Phi^t(\tilde{\rho}^i_0),\Phi^t(\tilde{x_i},\tilde{\xi}_i)),
\end{equation*} 
which is bounded as $t\rightarrow -\infty$ by assumption. Therefore, in the topology of the compactification of $X$ given by Hypothesis \ref{Guepard}, $\Phi^t(\rho_0)$ and $\Phi^t(x,\xi_i)$ are approaching each other as $t\rightarrow -\infty$. Since we took $\rho_0\in\mathcal{DE}_-$, we also have $\Phi^t(x,\xi_i)\in \mathcal{DE}_-$ if $-t$ is large enough. Therefore, by Hypothesis \ref{puissanceplus2}, $\Phi^{t}(x,\xi_i)\in \Lag$ if $-t$ is large enough. 

To prove the corollary, we have to check that the $\xi_i$ are all distinct, and hence that the $\tilde{\xi}_i$ are all distinct. It suffices to show that for $i\neq i'$, we have
\begin{equation}\label{luminaire}
d_{\tilde{X}}(\Phi^t(\rho_0^i),\Phi^t(\rho_0^{i'}))~~ \text{ is unbounded as } t\rightarrow - \infty.
\end{equation}
Suppose for contradiction that (\ref{luminaire}) were not true for some $i \neq i'$.

\begin{figure}
    \center
   \includegraphics[scale=0.6]{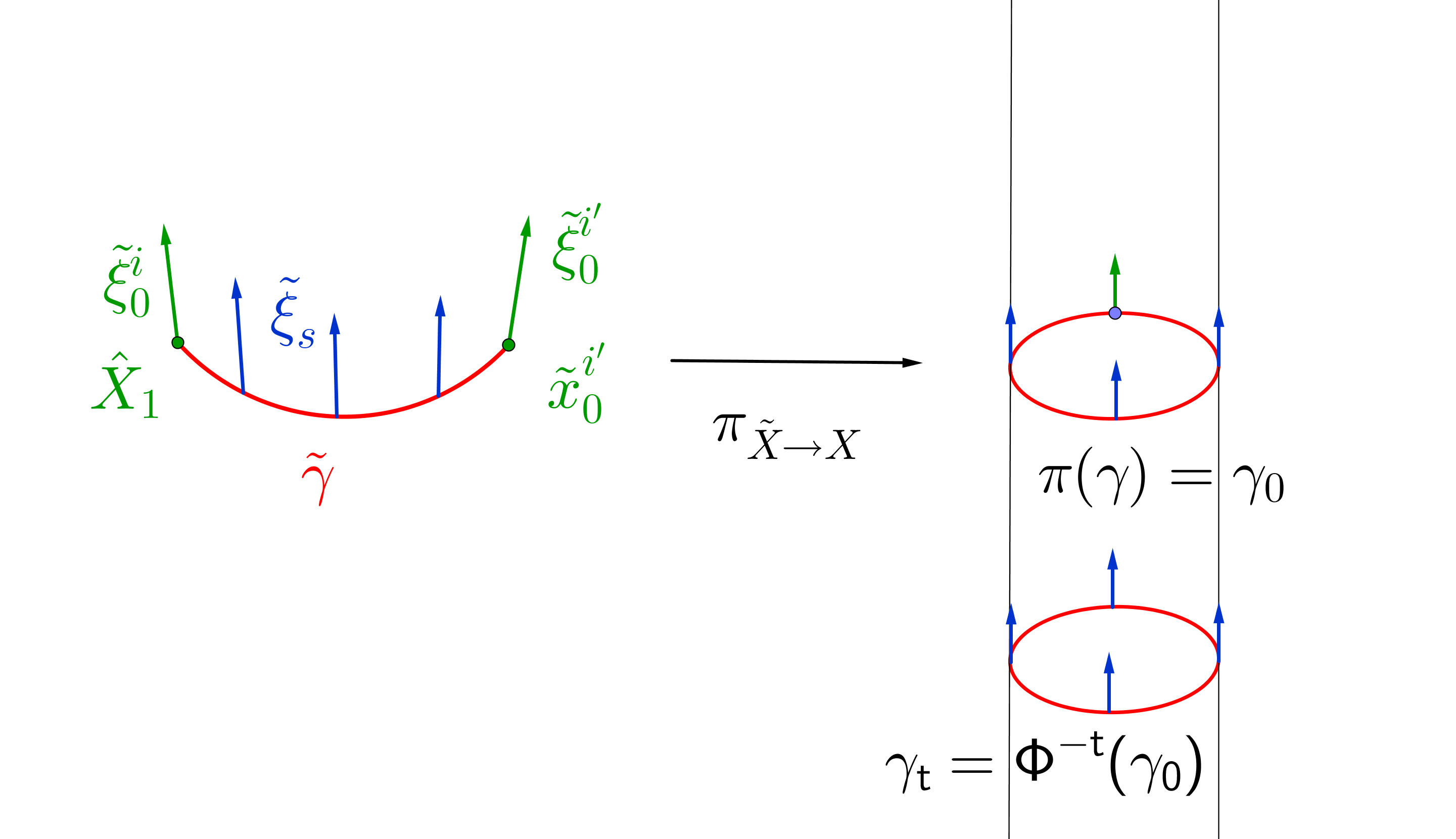}
   \caption{What happens if (\ref{luminaire}) is not true}\label{projectionrevetement}
\end{figure}

 Consider, as in Figure \ref{projectionrevetement}, the geodesic $\tilde{\gamma}$ in $\tilde{X}$ such that $\tilde{\gamma}(0)=\tilde{x}_0^i$ and $\tilde{\gamma}(1)=\tilde{x}_0^{i'}$.
  According to Lemma \ref{piscine}, for each $s\in [0,1]$, there exists a unique direction $\tilde{\xi}_s$ such that $d_{\tilde{X}}(\Phi^t(\rho_0^i), \Phi^t(\tilde{\gamma}(s),\tilde{\xi}_s))$ is bounded as $t\rightarrow -\infty$.
By the uniqueness part in Lemma \ref{piscine}, and because we assumed that (\ref{luminaire}) is not satisfied, we have that $\tilde{\xi}_1=\tilde{\xi}_0^{i'}$.

Let us consider $(\gamma(s),\xi_s):= \pi_{S^*\tilde{X}\rightarrow S^*X} (\tilde{\gamma}(s),\tilde{\xi}_s)$. We have $(\gamma(0),\xi_0)= \rho_0$ by assumption, and $\gamma(0)=\gamma(1)$. Furthermore,
\begin{equation*}
d_{X} (\Phi^t(\rho_0), \Phi^t(\gamma(s),\xi_s))\leq d_{\tilde{X}} (\Phi^t(\tilde{\gamma(0)},{\xi}_0), \Phi^t(\tilde{\gamma}(s),\tilde{\xi}_s)),
\end{equation*}
which is bounded as $t\rightarrow -\infty$ by construction of $\tilde{\xi}_s$.

Let us write $\gamma_t(s) = \Phi^t(\gamma(s),\xi_s)$. This curve has a length bounded independently of $t$ by what precedes, and its points are going to infinity as $t\rightarrow -\infty$.

Furthermore, for each $t\geq 0$, $\gamma_t$ is a closed curve which is not contractible, since it joins two different points in $\tilde{X}$. 

Therefore, for each $t\geq 0$, there must be a $s\in [0,1]$ such that $d_X(\gamma_t(0),\gamma_t(s))\geq r_i(\gamma_t(0))$. Indeed, if this were not true, $\gamma_t$ would be contained in a chart where the exponential map is injective, and it would be contractible.

But since $\gamma_t$ has a length bounded independently of $t$, and since $\gamma_t(0)$ goes to infinity with $t$, we obtain a contradiction with part (iii) of Hypothesis \ref{Strong}.
\end{proof}

\begin{remarque}\label{plaque}
The end of the proof actually shows that, if we denote by $\big{(}\tilde{\Lag}^j\big{)}_{j\in J}$ the different preimages of $\Lag$ by the projection $\pi_{S^*\tilde{X}\rightarrow S^*X}$, then we have for each $j\neq j'\in J$, $\tilde{\Lag}^j\cap\tilde{\Lag}^{j'}=\emptyset$.
\end{remarque}

\subsection{Smooth projection and Transversality}
\subsubsection{General criteria}
Recall that \emph{smoothly projecting} manifolds were introduced in Definition \ref{smoothproj}. We shall now rephrase the notion of projecting smoothly on $X$ in terms of transversality.

Let $L_1$, $L_2$ be two submanifolds of $S^*X$ and let $\rho\in L_1\cap L_2$. Recall that we say that $L_1$ and $L_2$ \emph{intersect transversally at $\rho$} if $T_\rho(S^*X)=T_\rho L_1\oplus T_\rho L_2$.

\begin{lemme}\label{smoothproj2}
Let $X'$ be a submanifold of $S^*X$
which can be written in the form
\begin{equation*}
X'= \{(x, f(x)), x\in \Omega\},
\end{equation*}
where $\Omega$ is an open subset of $X$, and where $f$ is $C^0$.

Suppose furthermore that for every $x\in \Omega$, the manifolds $X'$ and $S_{x}^*X$ intersect transversally at $(x,f(x))$. Then $X'$ projects smoothly on $X$.
\end{lemme}

\begin{proof}
Let us write $\kappa :X\ni x \mapsto (x,f(x))\in X'$ and $\pi : X'\ni (x,f(x))\mapsto x\in X$. We have of course $\pi \circ \kappa = Id$. The transversality assumption ensures us that $d\pi$ is invertible wherever it is defined. By the inverse function theorem, $\kappa$ is $C^\infty$. Therefore, since $f$ is the second component of $\kappa$, it is also smooth. 
\end{proof}
The following lemma gives us a criterion to check that two manifolds intersect transversally, which we shall use several times.

\begin{lemme}\label{hibou}
Let $L_1$, $L_2$ be two submanifolds of $S^*X$ with dimension $d$ and $d-1$ respectively, and let $\rho\in L_1\cap L_2$. We assume that (i) is satisfied, as well as (ii) or (ii'), where the assumptions (i), (ii) and (ii') are as follows.

(i) For all $\rho_1\in L_1$, the function $t\mapsto d_X (\Phi^{-t}(\rho_1),\Phi^{-t}(\rho))$ is non-increasing for $t\geq 0$.

(ii) For all $\rho_2\in L_2$, we have $d_X(\rho_2,\rho)=0$ (that is to say, $L_2\subset S_x^*X$ for some $x\in X$).

(ii') Set $x=\pi_{S^*X\rightarrow X}(\rho)$. The manifolds $L_1$ and $S^*_x X$ are transverse at $\rho$. Furthermore, there exists $\nu \bel 0$, $C\bel 0$ and $\epsilon\bel 0$ such that for all $\rho_2\in L_2$ and for all $t\leq 0$ such that $d_{ad}(\Phi^t(\rho),\Phi^t(\rho_2))\leq \epsilon$, we have:
\begin{equation*}
d_X(\Phi^{t}(\rho_2),\Phi^{t}(\rho))\geq C e^{-\nu t} d_X(\rho_2,\rho).
\end{equation*}

Then $L_1$ and $L_2$ intersect transversally at $\rho$.
\end{lemme}

\begin{proof}

Suppose for contradiction that $L_1$ and $L_2$ do not intersect transversally at $\rho$. This means that there exists $v\in T_\rho(L_1)\cap T_\rho(L_2)$, $v\neq 0$.

Therefore, there exist $\rho_1^n\in L_1$, $\rho_2^n\in L_2$ such that $d_{ad}(\rho_1^n,\rho)=1/n$ and $d_{ad}(\rho_2^n,\rho)=1/n$ but $d_{ad}(\rho_1^n,\rho_2^n)=o(1/n)$, as in as in Figure \ref{nontrans}.

\begin{figure}
    \center
   \includegraphics[scale=0.2]{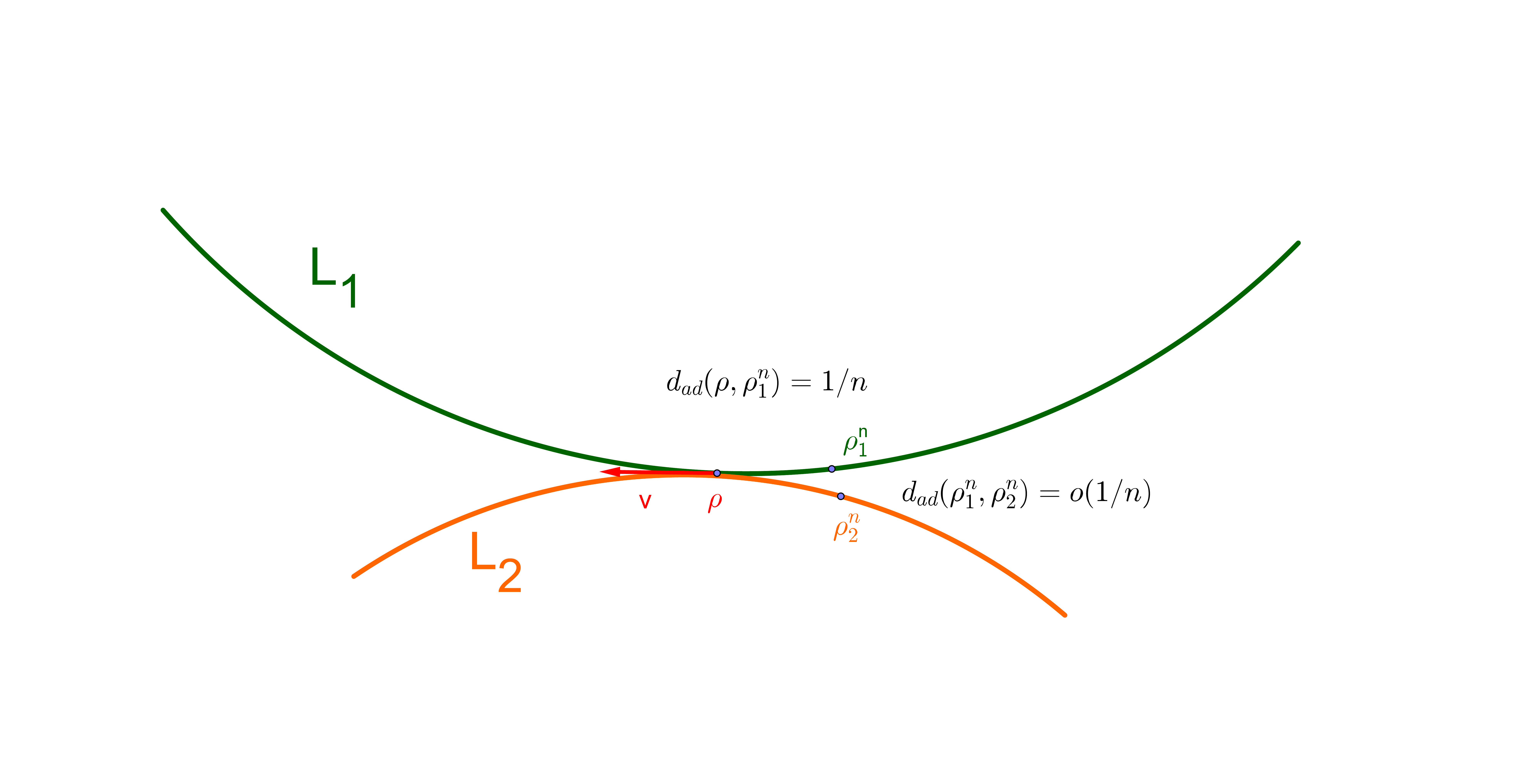}
   \caption{Two non-transverse manifolds}\label{nontrans}
\end{figure}
We will find a contradiction by finding a time $t\bel 0$ such that 
\begin{equation}\label{nespresso}
d_X(\Phi^{-t}(\rho),\Phi^{-t}(\rho_1^n))\bel d_X(\rho,\rho_1^n),
\end{equation} 
which will contradict assumption (i).

We have 
\begin{equation}\label{heureux}
d_X(\Phi^{-t}(\rho_1^n),\Phi^{-t}(\rho_2^n))\leq d_{ad}(\Phi^{-t}(\rho_1^n),\Phi^{-t}(\rho_2^n)) \leq e^{\mu t} o(1/n).
\end{equation}

Suppose first that assumption (ii) is satisfied, and write $\rho=(x,\xi)$ and $\rho_2^n=(x,\xi_2^n)$. Then, since $d_{ad}(\rho_2^n,\rho)= 1/n$ and since the distances induced by all metrics on $S^*X$ are equivalent in a neighbourhood of $\rho$, there exists a constant $c\bel 0$ such that $\|\xi-\xi_2^n\|\geq c/n$ for all $n$ large enough.

Therefore, $ \frac{d}{dt}\Big{|}_{t=0} d_X(\Phi^{-t}(\rho),\Phi^{-t}(\rho_2^n))\geq c/n$. Since the second derivative with respect to time of the distance between geodesics only involve local terms and since the trajectories of $\Phi^t(\rho)$ and $\Phi^{t}(\rho_1^n)$ are close to each other for $t$ small enough and $n$ large enough,
we may find a time $t_0\bel 0$ independent of $n$ and a $c'\bel 0$ such that for all $n$ large enough and all $t\in [0,t_0]$, we have $d_X(\Phi^{-t}(\rho),\Phi^{-t}(\rho_2^n))\bel c't/n$.  

We hence have
\begin{equation*}
\begin{aligned}
d_X(\Phi^{-t_0}(\rho),\Phi^{-t_0}(\rho_1^n))&\geq \Big{|}d_X(\Phi^{-t_0}(\rho),\Phi^{-t_0}(\rho_2^n))- d_X(\Phi^{-t_0}(\rho_2^n),\Phi^{-t_0}(\rho_1^n)) \Big{|}\\
&\geq c't_0/n + o(1/n).
\end{aligned}
\end{equation*}

On the other hand, we have
\begin{equation*}
d_{X}(\rho_1^n,\rho)\leq d_X (\rho_2^n,\rho)+ d_X(\rho_1^n,\rho_2^n)= d_X(\rho_1^n,\rho_2^n)\leq C d_{ad}(\rho_1^n,\rho_2^n) = o(1/n).
\end{equation*}

Therefore, taking $t_n=t_0$ gives (\ref{nespresso}).

Suppose now that (ii') is satisfied. 
 Since we have assumed that the manifolds $L_1$ and $S^*_x X$ are transverse at $\rho$, there exists a constant $c\bel 0$ such that for all $n$ large enough, we have
\begin{equation*}
c d_{ad}(\rho,\rho_1^n)\leq  d_X(\rho, \rho_1^n).
\end{equation*}

Therefore, $d_X(\rho_1^n,\rho)\geq c/n$. Let us consider a $t\geq 0$ such that $Ce^{\nu t}\bel 2$. If $n$ is large enough, we have $d_{ad}(\Phi^{-t}(\rho),\Phi^{-t}(\rho_2^n))\leq \epsilon$. We hence have
\begin{equation*}
\begin{aligned}
d_X(\Phi^{-t}(\rho),\Phi^{-t}(\rho_2^n)) \geq Cc \frac{e^{\nu t}}{n} \geq \frac{2c}{n}.
\end{aligned}
\end{equation*}

On the other hand, 
$$d_X(\Phi^{-t}(\rho_2^n),\Phi^{-t}(\rho_1^n))\leq e^{\mu t}d_{ad}(\rho_2^n,\rho_1^n) = o\Big{(}\frac{1}{n}\Big{)}.$$

We deduce from this that

\begin{equation*}
\begin{aligned}
d_X(\Phi^{-t}(\rho),\Phi^{-t}(\rho_1^n))&\geq \frac{2c}{n} - d_X(\Phi^{-t}(\rho_2^n),\Phi^{-t}(\rho_1^n)) \\
&> \frac{c}{n}.
\end{aligned}
\end{equation*}

This gives us (\ref{heureux}), and concludes the proof of the lemma.
\end{proof}

\subsubsection{Three applications of Lemma \ref{hibou}}
As a first application of Lemma \ref{hibou}, let us now state a useful lemma. It seems rather classical, but since we could not find a proof of it in the literature, we recall it. It is likely that this result (and most of the results of this section) still holds if we suppose that $(X,g)$ has \emph{no conjugate points}, instead of supposing that it has non-positive curvature. We refer the reader to \cite{ruggiero2007dynamics} for an overview of the properties of such manifolds.
\begin{lemme}\label{technicolor}
Let $X$ be a manifold of non-positive sectional curvature, such that $K$ is a hyperbolic set. Let $\rho\in K$. Then there exists $\epsilon\bel 0$ such that $W_{\epsilon}^{\pm0}(\rho)$ projects smoothly on $X$.
\end{lemme}
\begin{proof}
We shall only prove the lemma for $W^{+0}_\epsilon$. The result for $W^{-0}_\epsilon$ will follow, since stable and unstable manifolds are exchanged by changing the sign of the speeds $\xi$.

We want to apply Lemma \ref{smoothproj2}. 
Let us first check that $W_\epsilon^{+0}$ can indeed be written as a graph. Suppose that $\rho_i=(x,\xi_i)\in S^*X\cap W_\epsilon^{+0}$ for $i=1,2$. Then if $g(t)= d_X(\Phi^t(\rho_1),\Phi^t(\rho_2))$, we have $g(0)=0$ and $\lim \limits_{t\rightarrow \infty} g(t) <\infty$ by definition of the stable manifold. But, if $\epsilon$ is small enough, we have $g(t)\leq r_i$ for all $t\geq 0$, so that by Corollary \ref{dompteur2}, we have $\xi_1=\xi_2$. Therefore, $W_\epsilon^{+0}$ can be written for $\epsilon$ small enough as
\begin{equation*}
W_\epsilon^{+0} = \{(x,f(x)); x\in \Omega\}
\end{equation*}
for some open set $\Omega\subset X$. The fact that $W_\epsilon^{+0}$ is a connected manifold implies that $f$ is continuous.

We now have to prove the transversality condition of Lemma \ref{smoothproj2}, by applying Lemma \ref{hibou} for $L_1= W_\epsilon^-$ and $L_2=S_x^*X$. (ii) is trivially satisfied. Let us check that (i) is satisfied. Take $\rho_1,\rho_2\in W_\epsilon^{+0}(\rho)$, and write $g(t) = d_X(\Phi^t(\rho_1),\Phi^t(\rho_2))$. We have $\lim\limits_{t\rightarrow +\infty} g(t)<\infty$, and if $\epsilon$ is taken small enough, we may assume that $g(t)\leq r_i$ for all $t\geq 0$. Therefore, by Corollary \ref{dompteur2}, $g$ is non-increasing for $t\geq 0$, and (i) is satisfied. Hence we may apply Lemma \ref{hibou} and Lemma \ref{smoothproj2} to conclude the proof.
\end{proof}

\begin{remarque}\label{gabriele2}
Actually, with the same proof, we can prove the following statement, saying that if we propagate in the past (resp. future) a small piece of the stable (resp. unstable) manifold, it projects smoothly on $X$:

Let $\rho\in K$. Then there exists $\epsilon\bel 0$ such that for all $\pm t \geq 0$ and $\rho'\in W_{\epsilon}^{\pm 0}(\rho)$, there exists $\epsilon'\bel 0$ such that $\Phi^t(\{\rho''\in W^{\pm 0}_\epsilon(\rho); d_{ad}(\rho'',\rho')< \epsilon'\}$ projects smoothly on $X$. 
\end{remarque}

We now prove a lemma which is the first step in the proof of Theorem \ref{getup}.

\begin{lemme}\label{res}
Let $\tau\geq 0$ and let $\rho\in\Lag$. For $\epsilon\bel 0$, the manifold $\Phi^\tau( \{\rho'\in \Lag; d_{ad}(\rho,\rho')<\epsilon\}$ projects smoothly on $X$.
\end{lemme}

\begin{proof}
We have to check that the assumptions of Lemma \ref{smoothproj2} are satisfied. If $\epsilon$ is chosen small enough, then for all $t\leq \tau$ and for all $\rho'\in \Lag$ such that $d_{ad}(\rho,\rho')<\epsilon$, we have that $d_{ad}\big{(}\Phi^t(\rho'),\Phi^t(\rho)\big{)}\leq r_i$ thanks to Hypothesis \ref{puissanceplus}.
This implies that $\Phi^\tau( \{\rho'\in \Lag; d_{ad}(\rho,\rho')<\epsilon\}$ can be written as a graph above $X$. Indeed, suppose that this manifolds contains two points $\rho_1=(x,\xi_1)$ and $\rho_2=(x,\xi_2)$. Then, by Lemma \ref{dompteur2}, $t\mapsto d^2_X(\Phi^t(\rho_1),\Phi^t(\rho_2)\}$ is a convex function on $(-\infty;0)$. By Hypothesis \ref{puissanceplus}, this function goes to a constant as $t$ goes to $-\infty$. Hence it must be constant equal to zero, and $\rho_1=\rho_2$.

We now have to check that this graph is transversal to the vertical fibres. To do so, we want to apply Lemma \ref{hibou} for $L_1=\Phi^\tau( \{\rho'\in \Lag; d_{ad}(\rho,\rho')<\epsilon\}$ and $L_2= S_x^*X$.
Hypothesis (ii) of Lemma \ref{hibou} is then trivially satisfied. As for hypothesis (i), it is satisfied by Hypothesis \ref{puissanceplus} combined with Lemma \ref{dompteur2}. Therefore, we may apply Lemma \ref{hibou} and Lemma \ref{smoothproj2} to conclude the proof of the Lemma.
\end{proof}

As a last application of Lemma \ref{hibou}, we shall prove Proposition \ref{musaraigne}.
Recall that this proposition says that:
\begin{prop}
Suppose that $X$ satisfies Hypothesis \ref{Strong}, as well as Hypothesis \ref{sieste} on hyperbolicity, and that $\Lag$ satisfies Hypothesis \ref{chaise}. Then the Lagrangian manifold $\Lag$ automatically satisfies Hypothesis \ref{Happy} on transversality.
\end{prop}
\begin{proof}
Let $\rho\in K$, let $\tau \geq 0$ and let $\rho'\in \Phi^\tau(\Lag)\cap W^-_\epsilon(\rho)$.
We want to apply Lemma \ref{hibou} to $\Phi^\tau(\Lag)$ and $W_\epsilon^-(\rho)$, or at least to the restriction of these manifolds to a small neighbourhood of $\rho'$. If we take a small enough neighbourhood $V$ of $\rho'$, then by Lemma \ref{dompteur2} and Hypothesis \ref{puissanceplus}, we have that $L_1=\Phi^\tau(\Lag)\cap V$ satisfies assumption (i) of Lemma \ref{hibou}. 

Let us check that assumption (ii') is satisfied. The fact that $L_1$ projects smoothly on $X$ comes from Lemma \ref{res}. As for the second condition, we know by Lemma \ref{fix} that 
there exists $C,\lambda\bel 0$ such that for any $\rho\in K$ and $\rho_1,\rho_2\in W_{\epsilon}^-(\rho)$, we have for all $t\geq 0$:
\begin{equation*}
d_{ad}(\Phi^t(\rho_1),\Phi^t(\rho_2))\leq C e^{-\lambda t} d_{ad}(\rho_1,\rho_2).
\end{equation*}

We therefore have, for all $t\geq 0$ such that $\Phi^{-t}(\rho_i)\in W_{\epsilon}^-(\rho)$ for $i=1,2$ :
\begin{equation*}
d_{ad}(\rho_1,\rho_2)\leq C e^{-\lambda t} d_{ad}(\Phi^{-t}(\rho_1),\Phi^{-t}(\rho_2)).
\end{equation*}

But, since $W^{-0}_\epsilon(\rho)$ projects smoothly on $X$ for $\epsilon$ small enough by Lemma \ref{technicolor}, there exists a constant $C'$ such that for all $\rho_1,\rho_2\in W^-_\epsilon(\rho)$, we have
\begin{equation*}
\frac{1}{C'} d_X(\rho_1,\rho_2)\leq d_{ad}(\rho_1,\rho_2)\leq C' d_X(\rho_1,\rho_2).
\end{equation*}
Therefore, hypothesis (ii') of Lemma \ref{hibou} is satisfied, and we can use Lemma \ref{hibou} and Lemma \ref{smoothproj2} to conclude the proof.
\end{proof}

\subsection{Proof of Theorem \ref{getup}}
\begin{proof}
First of all, let us make a few remarks to show that we don't need to consider all sequences $\beta$.
We have $\Lag= (\Lag\cap \mathcal{DE}_-)\cup(\Lag\cap \mathcal{DE}_+)$. We therefore have $$\Phi_\beta^N(\Lag)\cap T^*\mathcal{O}= \Big{(}\Phi_\beta^N(\Lag\cap \mathcal{DE}_+)\cap T^*\mathcal{O}\Big{)}\cup \Big{(}\Phi_\beta^N(\Lag\cap \mathcal{DE}_-)\cap T^*\mathcal{O} \Big{)}.$$
By hypothesis (\ref{invpass}), we have $\Phi_\beta^N(\Lag\cap \mathcal{DE}_+)\subset \Lag  $. 
Consequently, we only have to focus on the propagation of $\Lag\cap \mathcal{DE}_-$.

Let us note that by (\ref{invpass}), we have that for all $k\geq 1$, $\Phi^{k}_{0,...,0}(\Lag\cap\mathcal{DE}_-)\subset \Phi^1_0(\Lag\cap\mathcal{DE}_-)$. We may hence, without loss of generality, restrict ourselves in what follows to sequences $\beta$ such that $\beta_2\neq 0$.
 
By the geodesic convexity assumption, for all bounded open set $\mathcal{O}$, we may find a $N_\mathcal{O}$ such that for all $N\geq N_{\mathcal{O}}$ and for all $\beta\in B^N$ such that $\beta_2\neq 0$ and $\Phi_\beta^N(\Lag)\cap \mathcal{O}\neq \emptyset$, we have $\forall k=2... N-N_\mathcal{O}$, $\beta_k\neq 0$.

Now, since the sets $V_b$, $b\in B_2$ are at positive distance from the trapped set, we may find a $N_1$ such that for all $N\geq N_\mathcal{O}+ 2 N_1$ and for all $\beta\in B^N$ such that $\beta_2\neq 0$ and $\Phi_\beta^N(\Lag)\cap \mathcal{O}\neq \emptyset$, we have  $\forall k = N_1... N-N_1-N_\mathcal{O}$, $\beta_k\in B_1$.

 Let us now check that a single Lagrangian manifold $\Phi_\beta^N(\Lag)$ projects smoothly on the base manifold.

\begin{lemme}\label{res2}
Let $\mathcal{O}\subset X$ be an open bounded set. Then for any $N\in \mathbb{N}$, $\beta\in B^N$, $\Phi_\beta^N(\Lag)\cap (S^*\mathcal{O})$ projects smoothly on $X$.
\end{lemme}

\begin{proof}
The proof follows from Lemma \ref{res}, if we can check that  $\Phi_\beta^N(\Lag)\cap (S^*\mathcal{O})$  can be written as a graph on $X$.
Let $\rho_i=(x,\xi_i)\in\Phi_\beta^N(\Lag)$, $i=1,2$. Thanks to condition (\ref{conditionepsilon}), we see that $\Phi^t(\rho_1)$ and $\Phi^t(\rho_2)$ remain at a distance smaller than
$r_i$ from each other for all $t\leq 0$. Therefore, by Lemma \ref{dompteur}, we have $\Phi^t(\rho_1) \equiv \Phi^t(\rho_2)$, and
$\xi_1=\xi_2$, which concludes the proof.
\end{proof}

Let us fix a small open set $\mathcal{O}$ and a local chart on it. Thanks to Lemma \ref{res2}, in these local coordinates, the manifold $\Phi^N_\beta(\Lag)\cap T^*\mathcal{O}$, if nonempty, may be written in the form
\begin{equation}\label{aaron}
\Phi^N_\beta(\Lag)\cap T^*\mathcal{O}\equiv \{(x,\partial_x \varphi_{\beta, \mathcal{O}} (x)); x\in \mathcal{O}^\beta\},
\end{equation}
where $\mathcal{O}^\beta$ is an open subset of $\mathcal{O}$, and $\partial_x \varphi_{\beta, \mathcal{O}}$ is a smooth function.

We will now show that the functions $\varphi_{\beta,\mathcal{O}}$ are bounded independently of $\beta$ and $N$ in any $C^\ell$ norm.

Let us start by working close to the trapped set, that is to say, by considering the case where $\mathcal{O}=\mathcal{O}_b= \pi_X(V_b)$ for some $b\in B_1$.

Let us denote by $\kappa_b$ the symplectomorphism sending $(x,\xi)$ to $(y^b,\eta^b)$ in a neighbourhood of $\rho^b$. In the notations of Theorem \ref{Cyril},
we have
\begin{equation*}
\{(x,\partial_x \varphi_{\beta, \mathcal{O}_b} (x)); x\in \mathcal{O}_b^\beta\}= \kappa_b^{-1} \big{(} \{(y_1^b,u^b,0,f_{b,\beta}(u^b)) ; (y^b_1,u^b)\in D_{b,\beta,N} \}\big{)}.
\end{equation*}
We will write $F_{b,\beta}:= (0,f_{b,\beta})$, and we will denote by $\hat{x}_b$ and $\hat{\xi}_b$ the components of $\kappa_b^{-1}$. 

We shall consider the map $\tilde{\kappa}_{\beta,b}$ sending $x\in \mathcal{O}^\beta_b$ to $y^b\in \mathbb{R}^d$ such that
\begin{equation*}
\kappa_b(x,\partial\varphi_{\beta,\mathcal{O}_b}(x)) = (y^b,F_{b,\beta}(y^b)).
\end{equation*}

\begin{lemme}\label{rennes}
If we take $\ins$ small enough in (\ref{defgamma}), we may find a constant $c\bel 0 $ such that for any $b\in B_1$, for any $N\in \mathbb{N}$ and for any $\beta\in B^N$, we have
\begin{equation}\label{nantes}
\Big{|}\frac{\partial \tilde{\kappa}_{\beta,b}^{-1}}{\partial y^b}\Big{|} \bel c.
\end{equation}
\end{lemme}
\begin{proof}
We have
\begin{equation}\label{gabriele}
\tilde{\kappa}_{\beta,b}^{-1} (y^b) = \hat{x}_b(y^b, F_{b,\beta}(y^b)),
\end{equation}
hence
\begin{equation*}
\begin{aligned}
\frac{\partial \tilde{\kappa}_{\beta,b}^{-1}}{\partial y^b} = \frac{\partial \hat{x}_b}{\partial y} (y^b, F_{b,\beta}(y^b)) + \frac{\partial  F_{b,\beta}(y^b)}{\partial y^b} \frac{\partial \hat{x}_b}{\partial \eta} (y^b, F_{b,\beta}(y^b)).
\end{aligned}
\end{equation*}
Now, $\frac{\partial \hat{x}_b}{\partial y}(y^b,0)$ is invertible at every $y^b$, because the unstable manifold projects smoothly on the base manifold, hence we may find a $c\bel 0$ such that $\Big{|}\frac{\partial x_b}{\partial y}(y^b,0)\Big{|}\bel 2c$ for any $b\in B_1$ and any $y^b$. Using the fact that $\| F_{b,\beta}\|_{C^1}\leq \ins$ for $\tau(\beta)$ large enough, the lemma follows for $\tau(\beta)$ large enough, that is to say, for $N$ large enough since we have assumed that $\beta_2\neq 0$. For finite values of $N$, inequality (\ref{nantes}) is true thanks to Lemma \ref{res}. The statement follows.
\end{proof}

Let us come back to the proof of Theorem \ref{getup}.
Equation (\ref{gabriele}) ensures us that $\tilde{\kappa}_{\beta,b}^{-1}$ is bounded independently of $\beta$ in any $C^\ell$ norm. Indeed, when we differentiate it, it depends on $\beta$ only through $F_{b,\beta}$, which is bounded independently of $\beta$ in any $C^\ell$ norm.

Therefore, we deduce from Lemma \ref{rennes} and the chain rule that we may bound the $C^\ell$ norm of the functions $\tilde{\kappa}_{\beta,b}$ independently of $\beta$.

Next, we check that $\partial \varphi_{\beta,\mathcal{O}_b}$ is bounded independently of $\beta$ in any $C^\ell$ norm. Indeed, we have
\begin{equation*}
\partial \varphi_{\beta,\mathcal{O}_b}(x)= \hat{\xi}_b \big{(}\tilde{\kappa}_{\beta,b}(x), F_{b,\beta} (\tilde{\kappa}_{\beta,b}(x) )\big{)}.
\end{equation*}

When we differentiate this expression, we see that the only terms which depend on $\beta$ involve some derivatives of $F_{b,\beta}$ and some derivatives of $\kappa_{\beta,b}$, both of which are bounded independently of $\beta$.

This proves the theorem in the case where $\mathcal{O}\subset \pi(V_b)$ for some $b\in B_1$.

 Let us now consider an arbitrary $\mathcal{O}$. We may suppose that $N\geq N_\mathcal{O}+2N_1$, since for finite values of $N$, the result just follows from Lemma \ref{res2}. Let $\beta\in B^N$ be such that $\beta_2\neq 0$ and $\Phi_\beta^N(\Lag)\cap \mathcal{O}\neq \emptyset$. By the preliminary remarks of the proof, there exists $k\leq N$ such that $\beta_k\in B_1$. We may suppose that $N-k\leq N_\mathcal{O}$ for a $N_\mathcal{O}$ which does not depend on $N$ or $k$. We then have
\begin{equation*}
\Phi^N_\beta(\Lag)\cap T^*\mathcal{O}\equiv \{(x,\partial_x \varphi_{\beta, \mathcal{O}} (x))\} = \Phi^{N-k}_{\beta''}\big{(} \{x', \partial\varphi_{\beta',\mathcal{O}_{\beta_k}}(x')\}\big{)}\cap T^*\mathcal{O},
\end{equation*}
where $\beta'=\beta_0...\beta_{k-1}$, $\beta''=\beta_k...\beta_{N-1}$ and where the $x'$ belong to a subset of $\mathcal{O}_b$.

Let us denote by $\tilde{\Phi}_{\beta,k}$ the map sending $x$ to the $x'$ such that $$(x,\partial_x \varphi_{\beta, \mathcal{O}} (x)) =  \Phi^{N-k}\big{(} x', \partial\varphi_{\beta',\mathcal{O}_{\beta_k}}(x')\big{)}.$$

We have $\partial_x \varphi_{\beta,\mathcal{O}}(x)= \pi_{\xi} \big{(} \Phi^{N-k} \big{(} \tilde{\Phi}_{\beta,k}(x), \partial\varphi_{\beta',\mathcal{O}_{\beta_k}}(\tilde{\Phi}_{\beta,k}(x)\big{)}\big{)}\big{)}$, where $\pi_\xi$ denotes the projection on the $\xi$ variable. Hence, since $\varphi_{\beta',\mathcal{O}_{\beta_k}}$ is bounded in any $C^\ell$ norm independently of $\beta$, we only have to prove that $\tilde{\Phi}_{\beta,k}$ is bounded in any $C^\ell$ norm independently of $\beta$.

On the one hand, we have $\tilde{\Phi}^{-1}_{\beta,k}(x')= \pi_x \big{(} \Phi^{N-k}(x',\partial\varphi_{\beta',\mathcal{O}_{\beta,k}}(x'))\big{)}$. Therefore, since $\Phi^{N-k}$ is a diffeomorphism, the map $\tilde{\Phi}_{\beta,k}$ is bounded in any $C^\ell$ norm independently of $\beta$ just like $\partial\varphi_{\beta',\mathcal{O}_{\beta,k}}$ is.

On the other hand, if $(x_1,...,x_d)$ are local coordinates on $X$, then for all $j=1...d$, and all $x\in \mathcal{O}$, we claim that 
\begin{equation}\label{inverseborne}
\Big{|}\partial_j \tilde{\Phi}^{-1}_{\beta,k}(x)\Big{|}\geq 1.
\end{equation} 

Suppose for contradiction that we can find $x$ such that $\Big{|}\partial_j \tilde{\Phi}^{-1}_{\beta,k}(x)\Big{|}< 1$. We may then find a sequence $x_n$ such that $d_X(x,x_n)=1/n$, and $d_X\big{(}\tilde{\Phi}^{-1}_{\beta,k}(x),\tilde{\Phi}^{-1}_{\beta,k}(x_n)\big{)}= <1/n$ for $n$ large enough. Write $\rho= \tilde{\Phi}^{-1}_{\beta,k}(x)$ and $\rho_n= \tilde{\Phi}^{-1}_{\beta,k}(x_n)\big{)}$. The function $d_X(\Phi^t(\rho),\Phi^t(\rho^n))$ is bounded for all $t\leq 0$, thanks to Hypothesis \ref{puissanceplus}. Furthermore, we have $d_X\big{(}\Phi^k(\rho),\Phi^k(\rho_n)\big{)}< d_X(\rho,\rho_n)$. But, if we take $n$ large enough, we can ensure that $d_X\big{(}\Phi^t(\rho), \Phi^t(\rho_n)\big{)}<r_i$ for all $t\in (-\infty,k)$, thus contradicting Corollary \ref{dompteur2}. This proves (\ref{inverseborne}).

The theorem then follows from the chain rule, since the derivatives of the inverse map are bounded as long as we apply them to vectors away from zero.
\end{proof}
\subsection{Distance between Lagrangian manifolds}
Let us now state a lemma concerning the distance between the Lagrangian manifolds described in Theorem \ref{getup}, which will play a key role in the proof of Corollary \ref{blacksabbath3}. It is very similar to Proposition 4 in \cite{Ing}.
\begin{lemme}\label{gidelim}
Let $\mathcal{O}$ be a bounded open set in $X$. There exists a constant $C_\mathcal{O}\bel 0$ such that for any $n,n'\in \mathbb{N}$, for any $\beta\in B^n, \beta'\in B^{n'}$, and for any $x\in\mathcal{O}$, such that $x\in \pi_X\big{(}\Phi_{\beta,\mathcal{O}}(\Lag)\big{)}\cap  \pi_X\big{(}\Phi_{\beta',\mathcal{O}}(\Lag)\big{)}$,
we have either $\partial \varphi_{\beta,\mathcal{O}}(x) = \partial \varphi_{\beta',\mathcal{O}}(x) $ or 
\begin{equation}\label{hurry}
|\partial \varphi_{\beta,\mathcal{O}} (x)- \partial \varphi_{\beta',\mathcal{O}}(x)|\geq C_1 e^{-\sqrt{b_0} \max(n-\tau(\beta),n'-\tau(\beta'))}, 
\end{equation}
with $\tau(\beta)$ defined as in (\ref{rollingstone}), and where $b_0$ is as in Hypothesis \ref{Strong}.
\end{lemme}

\begin{proof}
By the remarks at the beginning of the proof of Theorem \ref{ibrahim2}, we may restrict ourselves to sequences $\beta,\beta'$ such that $\beta_2,\beta'_2\neq 0$, that is to say, to sequences such that $\tau(\beta),\tau(\beta')\leq 1$, and we may find constants $N_\mathcal{O}$ and $N_1$ such that for all $N\geq N_\mathcal{O}+ 2 N_1$ and for all $\beta\in B^N$ such that $\beta_2\neq 0$ and $\Phi_\beta^N(\Lag)\cap \mathcal{O}\neq \emptyset$, we have  $\forall k = N_1... N-N_1-N_\mathcal{O}$, $\beta_k\in B_1$.

Let $n,n'\in \mathbb{N}$, $n,n'\geq N_\mathcal{O}+2N_1$, $\beta\in B^n, \beta'\in B^{n'}$, such that $\beta_2,\beta'_2\neq 0$, and let $x\in\mathcal{O}$ be such that $x\in \pi_X\big{(}\Phi_{\beta,\mathcal{O}}(\Lag)\big{)}\cap  \pi_X\big{(}\Phi_{\beta',\mathcal{O}}(\Lag)\big{)}$, and $\partial \varphi_{\beta,\mathcal{O}}(x) \neq \partial \varphi_{\beta',\mathcal{O}}(x)$. Let us write $\rho=(x,\partial\varphi_{\beta,\mathcal{O}})$ and $\rho'=(x,\partial\varphi_{\beta',\mathcal{O}})$.

  We claim that there exists $0\leq k\leq max(n,n')$ such that for each $b'\in B$, if $\Phi^{-k}(\rho)\in V_{b'}$, then $\Phi^{-k}(\rho')\notin V_{b'}$.
Indeed, if no such $k$ existed, then for each $k$, there would exist $b_k\in B$ such that $\Phi^{-k}(\rho)\in V_{b_k}$ and $\Phi^{-k}(\rho')\in V_{b_k}$ for each $0\leq k\leq \max(n,n')$. We would then have $\rho\in \Phi_{\beta''}^{\max(n,n')}(\Lag)$, and $\rho'\in \Phi_{\beta''}^{\max(n,n')}(\Lag)$ for some sequence $\beta''$ built by possibly adding some $0$'s in front of the sequences $\beta$ and $\beta'$. This would contradict the statement of Theorem \ref{getup}.

Recall that from the (elementary) Lemma 22 in \cite{Ing}, we have the existence of a constant $c_0\bel 0$ such that for any $\rho,\rho'\in S^*X_0$ such that $d_{ad}(\rho,\rho')< c_0$, there exists $b\in B$ such that $\rho,\rho'\in V_b$.

Thanks to this, we deduce that there exists $0\leq k \leq \max(n,n')$ such that $d_{ad}(\Phi^{-k}(\rho),\Phi^{-k} (\rho'))\geq c_0$.

Combining this fact with equation (\ref{flowers}) and Remark \ref{lenovo}, we get
\begin{equation*}
d_{ad}(\rho,\rho')\geq C e^{-\sqrt{b_0} \max(n,n')}
\end{equation*}
Using the fact that all metrics are equivalent on a compact set, we may compare $d_{ad}(\rho,\rho')$ with $|\partial \varphi_{\beta,\mathcal{O}} (x)- \partial \varphi_{\beta',\mathcal{O}}(x)|$ and we deduce from this the lemma.
\end{proof}

\section{Proof of the results concerning the distorted plane waves}\label{preuvequantique}
\subsection{Proof of Theorem \ref{ibrahim2}}
\begin{proof}
We want to write $E_h$ as a sum of Lagrangian states associated to Lagrangian manifolds which do all project smoothly on the base manifold. Most of the work was done in \cite{Ing}, and we shall now recall what we need from this paper.

Let us write $P_h:= -h^2\Delta-c_0 h^2 -1$ and $U(t):= e^{-it P_h/h}$.
Let us fix $\chi\in C_c^\infty(X)$, and a function $\hat{\chi}$ as in Remark \ref{remarqueoutgoing}. Recall that the pseudo-differential operators $(\Pi_b)_{b\in B}$ were introduced in Theorem \ref{ibrahim}.
 For each $b\in B$, we set $U_b:= e^{i/h}\Pi_b U(1)$. For each $n\in \mathbb{N}$, for each $\beta=\beta_0,...,\beta_{n-1}\in B^n$, we set $U_\beta= U_{\beta_{n-1}}\circ...\circ U_{\beta_0}$.
 
Recall that $\mathcal{U}_b: L^2(X)\longrightarrow L^2(\mathbb{R}^d)$ is a Fourier integral operator quantizing the symplectic
change of local coordinates $\kappa_b:(x,\xi) \mapsto
(y^{\rho_b},\eta^{\rho_b})$, and which is microlocally unitary on the microsupport of $\Pi_b$. Therefore, $\mathcal{U}_b^*: L^2(\mathbb{R}^d)\rightarrow L^2(X)$ quantizes the change of coordinates $(y^{\rho_b},\eta^{\rho_b})\mapsto (x,\xi)$.

Recall that the time $\tins$ was introduced in Theorem \ref{Cyril}. Equation (81) in \cite{Ing} tells us that there exist $\hat{\chi}_1\in C^\infty_c(X)$, and a constant\footnote{The time $N_{\hat{\chi}}\in \mathbb{N}$ is actually chosen so that for any $N\in \mathbb{N}$, if $\rho\in \mathrm{supp} (\hat{\chi}_{1})$ and $\Phi^N(\rho)\in \mathrm{supp} (\hat{\chi})$, then for all $k$ with $N_{\hat{\chi}}\leq k \leq N-N_{\hat{\chi}}$, we have $\Phi^{k}(\rho)\in V_b$ for some $b\in B_1\cup B_2$.} $N_{\hat{\chi}}\bel 0$ such that, writing $n':=n+2\tins+1$, we have

\begin{equation*}
\begin{aligned}
\chi E_h=\chi \hat{\chi}E_h &= \chi\sum_{n=1}^{M_r |\log h|} (\hat{\chi} U(1))^{N_{\hat{\chi}+1}} \Big{(}
\sum_{\beta\in \mathcal{B}_{n'+\tins+1}}U_{\beta_{n'}...\beta_{n'+\tins+1}}
\mathcal{U}^*_{\beta_{n'}} \theta_{n',\beta}\Big{)}\\
&+\chi\sum_{l=0}^{N_\chi+ 3\tins+3}\big{(}\hat{\chi} U(1)\big{)}^l(1-\hat{\chi})\hat{\chi}_{1} E_h^0 +O_{L^2}(h^r).
\end{aligned}
\end{equation*}

Here the sets $\mathcal{B}_N\subset B^N$ are defined, for $N\geq 2\tins+2$ by
\begin{equation*}
\mathcal{B}_N= (B_1\cup B_2)^{\tins+1} \times B_1^{N-2\tins-2} \times  (B_1\cup B_2)^{\tins+1} ,
\end{equation*}

and the $\theta_{n',\beta}$ are Lagrangian states of the form
\begin{equation*}
\theta_{n',\beta}(y)=a_{n',\beta}(y;h)e^{i\phi_{n',\beta}(y)},
\end{equation*}
with $y$ in some bounded open subset of $\mathbb{R}^d$.
Furthermore, we have estimates analogous to (\ref{sheriff}), that is: for any $\ell\in \mathbb{N}$, $\epsilon>0$, there exists $C_{\ell,\epsilon}$
such that for all $n\geq 0$, for all $h\in (0,h_0]$, we have
\begin{equation*}
\sum_{\beta\in \mathcal{B}_{n'}} \|a_{n',\beta}\|_{C^\ell} \leq C_{\ell,\epsilon}
e^{n(\mathcal{P}(1/2)+\epsilon)}.
\end{equation*}

Just as in (\ref{calamar}), the Lagrangian states $\theta_{n',\beta}$ are associated to the Lagrangian manifolds described in Theorem \ref{Cyril}. Namely, we have

\begin{equation*}
\{(y^{{\beta_{n'}}}, \partial \phi_{n',\beta} (y^{{\beta_{n'}}}))\}\subset \kappa_{\beta_{n'}}\big{(}\Phi_{\beta_{n'}}^{n'}(\Lag)\big{)}.
\end{equation*}

Now, $(\hat{\chi} U(1))^{N_{\hat{\chi}}+1} U_{\beta_{n'},...,\beta_{n'+\tins+1}} \mathcal{U}_{\beta_{n'}}^*$ 
and $(\hat{\chi} U(1))^l(1-\hat{\chi})\hat{\chi}_1$ are Fourier integral operators.
We want to apply Lemma \ref{pacal} to describe their action on the states $e_{n'}$ and $E_h^0$ respectively. 
Theorem \ref{getup} ensures us that (\ref{beignet}) is satisfied. We must check that (\ref{minhlien}) is satisfied. For the operators $(\chi U(1))^{N_\chi+1} U_{\beta_{n'},...,\beta_{n'+\tins+1}} \mathcal{U}_{\beta_{n'}}^*$, it follows from the fact that the propagation in the future of small pieces of local unstable manifolds project smoothly on $X$ as explained in Remark \ref{gabriele2}.

To apply $(\hat{\chi} U(1))^l(1-\hat{\chi})\hat{\chi}_1$ to $E_h^0$, take $\rho\in \Lag\cap S^*\big{(}\mathrm{supp}((1-\hat{\chi})\hat{\chi}_1))\big{)}$. We may take coordinates on $X$ near $\pi(x)$ which induces symplectic coordinates $(\hat{x},\hat{\xi})$ on a neighbourhood of $\rho$ in $T^*X$ such that $\rho=(0,0)$ in these coordinates and such that the manifold $\Lag$ is tangent to $\{\hat{\xi}=0\}$. In this system of coordinates near $\rho$, and with any system of coordinates on $X$ near $\pi(\Phi^t(\rho))$, we are in the framework of Lemma \ref{pacal}. Indeed, the condition (\ref{minhlien}) is then satisfied thanks to Lemma \ref{res}.

By writing 
\begin{equation*}
\begin{aligned}
{\mathcal{B}}_n^{\hat{\chi}}&:= \{\beta\beta'~~;~~ \beta\in \mathcal{B}_{n'},  ~~~\beta'\in B^{N_{\hat{\chi}} +1}\},~~~~ \text{ if } n\geq 2,\\
{\mathcal{B}}_1^{\hat{\chi}}&:= \{\beta\in B^{N_{\hat{\chi}}+ 3N_{uns} +1}\},
\end{aligned}
\end{equation*}
we obtain that

\begin{equation}\label{faust}
\chi E_h(x) = \sum_{n=0}^{\lfloor
M_{r,\ell} |\log h|\rfloor}
\sum_{\beta\in \mathcal{B}^{\hat{\chi}}_n} e^{i \varphi_{\beta}(x)/h}
a_{\beta,{\hat{\chi}}}(x;h) + R_{r,{\hat{\chi}}},
\end{equation}
where $a_{\beta,\hat{\chi}}\in S^{comp}(X)$, and each $\varphi_{\beta}$ is a smooth function defined in a neighbourhood of the support of
$a_{\beta,\hat{\chi}}$, and where
we have
 \begin{equation*}\|R_{r,\hat{\chi}}\|_{C^\ell}=O(h^r).
\end{equation*}

 For any $\ell\in \mathbb{N}$, $\epsilon>0$, there exists $C_{\ell,\epsilon}$
such that
\begin{equation}\label{sheriff4}
\sum_{\beta\in \mathcal{B}^{\hat{\chi}}_n} \|a_{\beta,{\hat{\chi}}}\|_{C^\ell} \leq C_{\ell,\epsilon}
e^{n(\mathcal{P}(1/2)+\epsilon)}.
\end{equation}

Note that the expression (\ref{faust}) has interest of its own, because it applies to functions $\chi\in C_c^\infty$ with arbitrary support. We shall now show that if the support of $\chi$ is small enough, we may regroup the terms of (\ref{faust}) in a clever way so that (\ref{hurry2}) is satisfied. This is the content of the next subsection.

\subsubsection{Regrouping the Lagrangian states}\label{regrouping}
From now on, we fix a compact set $\mathcal{K}\subset X$, and we suppose that $\hat{\chi}=1$ on $\mathcal{K}$.
\begin{lemme}\label{dedikodu}
There exists $\varepsilon_\mathcal{K}$ such that, for any open set $\mathcal{O}\subset \mathcal{K}$ of diameter smaller than $\varepsilon_\mathcal{K}$, the following holds: $\forall n\in \mathbb{N}$, $\forall \beta \in \mathcal{B}_n^{\hat{\chi}}$, $\forall t \geq 0$ and for all $\rho,\rho'\in \Phi_\beta^n(\Lag)\cap S^*\mathcal{O}$, we have $d_X(\Phi^{-t}(\rho),\Phi^{-t}(\rho')) < r_i$.
\end{lemme}
\begin{proof}
First of all, note that thanks to Hypothesis \ref{puissanceplus}, we only have to show the result for $t\leq |\beta|$, since $\Phi^{-|\beta|}(\rho)$ and $\Phi^{-|\beta|}(\rho')$ both belong to $\Lag$, and hence can only approach each other in the past.

Take $\epsilon$ small enough so that for all $\rho_1,\rho_2\in S^*\mathcal{K}$ with $d_{ad}(\rho_1,\rho_2)< \epsilon$, we have $d\big{(}\Phi^t(\rho_1),\Phi^t(\rho_2)\big{)} < r_i$ for all $t\in [0;N_\chi+3\tins+1]$.

Thanks to Theorem \ref{getup}, the Lagrangian manifolds $\Phi_\beta^n(\Lag)\cap S^*\mathcal{K}$ project smoothly on $X$, so that we may find a constant $C_\mathcal{K}$ such that, $\forall n\in \mathbb{N}$, $\forall \beta \in \mathcal{B}_n^{\hat{\chi}}$ and for all $\rho,\rho'\in \Phi_\beta^n(\Lag)\cap S^*\mathcal{O}$, we have
$$d_{ad}(\rho,\rho')\leq C_{\mathcal{K}} d_X(\rho,\rho').$$

By taking $\varepsilon_{\mathcal{K}}= \epsilon/C_{\mathcal{K}}$, this gives us the result when $n=1$. 

When $n\geq 2$, the result follows from the definition of $\mathcal{B}_n^{\hat{\chi}}$, from the fact that the sets $(V_b)_{b\in B_1\cup B_2}$ all have diameter smaller than $r_i$, and from equation (\ref{conditionepsilon}).
\end{proof}

Without loss of generality, we will always take $\varepsilon_{\mathcal{K}}\leq 1$. Let $\mathcal{O}\subset \mathcal{K}$ be an open set of diameter smaller that $\varepsilon_\mathcal{K}$.

Let us fix $\tilde{\Lag}$ to be a pre-image of $\Lag$ by the projection $S^*\tilde{X}\rightarrow S^*X$.
Let us denote by $(\mathcal{O}_j)_{j\in \mathcal{J}}$ the pre-images of $\mathcal{O}$ by the projection $\tilde{X}\rightarrow X$. If we suppose that the diameter of $\mathcal{O}$ is smaller than $r_i$, then the $\mathcal{O}_j$ are all disjoint.

For each $b\in B$, we set $\tilde{V}_b$ to be the pre-image of $V_b$ by the projection $S^*\tilde{X}\rightarrow S^*X$. The truncated propagator $\Phi^k_\beta$ may then be defined on $S^*X$ just as in (\ref{deftronque}), but with $V_b$ replaced by $\tilde{V}_b$.

For every $n\in\mathbb{N}$ and every $\beta\in \mathcal{B}^{\hat{\chi}}_n$ such that $\Phi_\beta(\Lag)\cap S^*\mathcal{O}\neq \emptyset$, we claim that there exists a unique $j_\beta\in \mathcal{J}$ such that $\Phi_\beta(\tilde{\Lag})\subset S^*\mathcal{O}_{j_\beta}$. Indeed, it is by definition, we have $\Phi_\beta(\tilde{\Lag})\subset \bigcup_{j\in \mathcal{J}} S^*\mathcal{O}_{j}$. But, by Lemma \ref{dedikodu}, we know that for all $t\leq 0$, $\Phi^{-t} \big{(} \Phi_\beta(\tilde{\Lag}) \cap S^*\mathcal{O} \big{)}$ has a diameter smaller than $r_i$, so that it is contained in a single coordinate chart.

If $\beta\in \mathcal{B}^{\hat{\chi}}_n$ and $\beta'\in \mathcal{B}^{\hat{\chi}}_{n'}$
 are such that $\Phi_\beta(\Lag)\cap S^*\mathcal{O}\neq \emptyset$, $\Phi_{\beta'}(\Lag)\cap S^*\mathcal{O}\neq \emptyset$, 
 we shall say that $\beta\sim_{\mathcal{O}} \beta'$ if $j_\beta= j_{\beta'}$. The relation $\sim_{\mathcal{O}}$ is clearly an equivalence relation on $\bigcup_{n\in \mathbb{N}}\{\beta\in\mathcal{B}^{\hat{\chi}}_n; \Phi_\beta(\Lag)\cap S^*\mathcal{O}\neq \emptyset\}$, so we may define
\begin{equation}
\tilde{\mathcal{B}}^\mathcal{O}:= \Big{(}\bigcup_{n\in \mathbb{N}}\{\beta\in\mathcal{B}^{\hat{\chi}}_n; \Phi_\beta(\Lag)\cap S^*\mathcal{O}\neq \emptyset\}\Big{)} \backslash \sim_\mathcal{O}.
\end{equation}

We then define for each $\tilde{\beta}\in \tilde{\mathcal{B}}^\mathcal{O}$:
\begin{equation}\label{rhodia}
\tilde{n}(\tilde{\beta}):= \min \{n\in \mathbb{N};~~ \exists \beta \in \mathcal{B}^{\hat{\chi}}_n \text{ such that } \beta\in \beta'\}.
\end{equation}

\begin{lemme}\label{classefinie}
There exists $N_\mathcal{K}$ such that for all $n\in \mathbb{N}$, $\beta\in \mathcal{B}^{\hat{\chi}}_n$, we have $\beta\in \tilde{\beta}\Longrightarrow n\leq \tilde{n}(\tilde{\beta})+N_\mathcal{K}$.
\end{lemme}

Note in particular that this lemma implies that there are only finitely many elements in the equivalence class $\tilde{\beta}$.
\begin{proof}
First of all, by compactness of $X_0$, we may find $N_0$ such that for any $\rho\in \Lag$, we have $\Phi^{-N_0}(\rho)\notin \bigcup_{b\in B_2}V_b$. Note that we then have for all $t\geq 0$:
\begin{equation}\label{seloigne}
d_X\Big{(}\Phi^{-N_0-t}(\rho), \bigcup_{b\in B_2}V_b\Big{)} \geq t.
\end{equation}
 
Let $\beta,\beta'\in \tilde{\beta}$, with $\beta\in \mathcal{B}^{\hat{\chi}}_n$ and $\beta'\in \mathcal{B}^{\hat{\chi}}_{n'}$. Suppose for contradiction that $n\geq n'+N_0 +3$. Let $\tilde{\rho}\in \Phi_\beta(\tilde{\Lag})\cap S^*\mathcal{O}_{j_\beta}$ and $\tilde{\rho'}\in \Phi_\beta(\tilde{\Lag})\cap S^*\mathcal{O}_{j_\beta}$. By assumption, we have $d_X(\rho,\rho')\leq \varepsilon_\mathcal{K}$. Furthermore, by Lemma \ref{dompteur}, $d_X(\Phi^{-t}(\rho),\Phi^{-t}(\rho'))$ is non-increasing. Hence 
\begin{equation}\label{derh}
d_X(\Phi^{-n}(\rho),\Phi^{-n}(\rho'))\leq \varepsilon_\mathcal{K}\leq 1.
\end{equation}

On the other hand, since $n\geq n'+N_0 +3$, we have by (\ref{seloigne}) that $d_X\big{(}\Phi^{-n}(\rho'),\bigcup_{b\in B_2} V_b\big{)} \geq 3$, while $d_X\big{(}\Phi^{-n}(\rho'),\bigcup_{b\in B_2} V_b\big{)} \leq 1$. Therefore,  we have $d_X(\Phi^{-n}(\rho),\Phi^{-n}(\rho'))\geq 2$, which contradicts (\ref{derh}).

This concludes the proof by taking for $\beta'$ a sequence which realises the minimum in (\ref{rhodia}).
\end{proof}

Let us now give a more useful description of the equivalence relation $\sim_\mathcal{O}$

\begin{lemme}\label{yandaki}
Let $n,n'\in \mathbb{N}$ and let $\beta\in \mathcal{B}^{\hat{\chi}}_n$, $\beta'\in \mathcal{B}^{\hat{\chi}}_{n'}$ be such that $\spt(\varphi_{\beta,\mathcal{O}})\cap \spt(\varphi_{\beta',\mathcal{O}})\neq \emptyset$. Then we have 
$$(\beta\sim_\mathcal{O} \beta')\Leftrightarrow \Big{(}\forall x\in \spt(\varphi_{\beta,\mathcal{O}})\cap \spt(\varphi_{\beta',\mathcal{O}}), \text{ we have } \partial\varphi_{\beta,\mathcal{O}}(x)=\partial\varphi_{\beta',\mathcal{O}}(x)\Big{)}.$$
\end{lemme}

\begin{proof}
Suppose first that $\beta\sim_\mathcal{O}\beta'$, and suppose first for contradiction that there exists $x\in \spt(\varphi_{\beta,\mathcal{O}})\cap \spt(\varphi_{\beta',\mathcal{O}})$ such that $ \partial\varphi_{\beta,\mathcal{O}}(x)\neq \partial\varphi_{\beta',\mathcal{O}}(x)$. Let us denote by $\tilde{x}$ the unique pre-image of $x$ by $\tilde{X}\rightarrow X$ such that $\tilde{x}\in \mathcal{O}_{j_\beta}$.  In $\tilde{X}$, the geodesics starting at $\tilde{x}$ with respective speeds $\partial \varphi_{\beta,\mathcal{O}}(x)$ and $\partial \varphi_{\beta',\mathcal{O}}(x)$ approach each other in the past, since they belong to $\bigcup_{t\geq 0} \Phi^t(\tilde{\Lag})$, which contradicts Lemma \ref{dompteur}. 

Reciprocally, suppose there exists $x\in \spt(\varphi_{\beta,\mathcal{O}})\cap \spt(\varphi_{\beta',\mathcal{O}})$, such that we have  $\partial\varphi_{\beta,\mathcal{O}}(x)=\partial\varphi_{\beta',\mathcal{O}}(x)$. Consider $\rho= \Phi^{-n}((x,\partial \varphi_{\beta,\mathcal{O}}(x))\in \Lag$. Let $\tilde{\rho}$ be the pre-image of $\rho$ by $S^*\tilde{X}\rightarrow S^*X$ which is in $\tilde{\Lag}$. By definition, we have $\Phi^n(\tilde{\rho})\in S^*\mathcal{O}_{j_\beta}$, but also $\Phi^{n'}(\tilde{\rho})\in S^*\mathcal{O}_{j_{\beta'}}$. Therefore, $j_\beta=j_{\beta'}$, so that $\beta\sim_\mathcal{O} \beta'$.
\end{proof}

Thanks to this lemma, it is possible for each $\tilde{\beta}\in \tilde{\mathcal{B}}^\mathcal{O}$
 to build a phase function $\varphi_{\tilde{\beta}} : \mathcal{O}\rightarrow \mathbb{R}$ such that 
for every $\beta\in \tilde{\beta}$, 
we have $\partial \varphi_{\tilde{\beta}}(x)=\partial \varphi_{\beta,\mathcal{O}}(x)$ for every $x\in \spt(\varphi_{\beta,\mathcal{O}})$.

Let $\chi\in C_c^\infty(X)$ be such that $\spt \chi\subset \mathcal{K}$ has a diameter smaller than $\varepsilon_\mathcal{K}$. Let us write $\tilde{\mathcal{B}}^{\chi}:=\tilde{\mathcal{B}}^{\spt(\chi)}$.

For every $\tilde{\beta}\in \tilde{\mathcal{B}}^\mathcal{\chi}$, we set 
\begin{equation*}
a_{\tilde{\beta},\chi}:= \sum_{\beta\in \tilde{\beta}} a_{\beta,\chi} e^{i(\varphi_{\tilde{\beta}}-\varphi_{\beta,\spt(\chi)})/h}.
\end{equation*}
Defined this way, $a_{\tilde{\beta},\chi}\in S^{comp}(X)$. Indeed, by Lemma \ref{classefinie}, the number of terms in the sum is bounded by a constant independent of $\tilde{\beta}$, and by Lemma \ref{yandaki}, the exponentials which appear in the sum are only constants. Therefore, by (\ref{sheriff4}), for each $\ell\in \mathbb{N}$ there exists $C'_{\ell,\epsilon}\bel 0$ such that for every $n\in \mathbb{N}$

\begin{equation}\label{sheriff5}
\sum_{\substack{\tilde{\beta}\in \tilde{\mathcal{B}}^\mathcal{\chi}\\ \tilde{n}(\tilde{\beta})= n}} \|a_{\tilde{\beta},\chi}\|_{C^\ell} \leq C_{\ell,\epsilon}
e^{n(\mathcal{P}(1/2)+\epsilon)}.
\end{equation}

From this construction and from (\ref{faust}), we obtain that
for any $r>0$, $\ell>0$, there exists $M_{r,\ell}>0$ such that
 we have as $h\rightarrow 0$:
\begin{equation*} \chi E_h(x) = \sum_{\substack{\tilde{\beta}\in \tilde{\mathcal{B}}^\chi\\
                \tilde{n}(\tilde{\beta})\leq \tilde{M}_{r,\ell}|\log h|}} e^{i \varphi_{\tilde{\beta}}(x)/h}
a_{\tilde{\beta}}(x;h) + R_{r,\ell},
\end{equation*}
with
 \begin{equation*}\|R_{r,\ell}\|_{C^\ell}=O(h^r).
\end{equation*}

From Lemma \ref{gidelim} and from the fact that if $\tilde{\beta}\neq \tilde{\beta}'$, we have $\partial \varphi_{\tilde{\beta}}(x)\neq \partial \varphi_{\tilde{\beta}'}(x)$ thanks to Lemma \ref{yandaki}, we deduce that there exists a constant $C_1$ such that for all $\tilde{\beta}, \tilde{\beta}'\in \tilde{\mathcal{B}}^\chi$, we have
\begin{equation}\label{hurry7}
|\partial \varphi_{\tilde{\beta}} (x)- \partial \varphi_{\tilde{\beta'}}(x)|\geq C_1 e^{-\sqrt{b_0} \max(\tilde{n}(\tilde{\beta}),\tilde{n}(\tilde{\beta'}))}. 
\end{equation}

This concludes the proof of Theorem \ref{ibrahim2}.

\end{proof}

\subsection{Proof of Corollary \ref{blacksabbath3}}\label{warpigs}
The main ingredient in the proof of Corollary \ref{blacksabbath3} is  non-stationary phase. Let us recall the estimate we will use, and which can be proven by integrating by parts.

Let $a,\phi\in S^{comp}(X)$, with $\spt(a)\subset \spt(\phi)$.
We consider the oscillatory integral:
$$I_h(a,\phi):= \int_{X} a(x) e^{\frac{i\phi(x,h)}{h}} \mathrm{d}x.$$

The following result is classical, and its proof similar to that of \cite[Lemma 3.12]{Zworski_2012}.
\begin{proposition}[Non stationary phase]\label{nonstat}
Let $\epsilon>0$. Suppose that there exists $C>0$ such that, $\forall x\in \spt(a), \forall  0<h<h_0$, $|\partial \phi(x,h)|\geq C h^{1/2-\epsilon}$. Then
$$I_h(a,\phi)=O(h^\infty).$$
\end{proposition}

\begin{proof}[Proof of Corollary \ref{blacksabbath3}]
Let $a\in C^\infty_c(S^*X)$, and let us write $\chi_a$ for a $C_c^\infty$ function which is equal to 1 on $\pi_X(\spt(a))$. We clearly have $\big{\langle} Op_h(a) \chi E_h,\chi E_h \big{\rangle} = \big{\langle} Op_h(a) \chi_a E_h,\chi_a E_h \big{\rangle} + O(h^\infty)$.

Since the statement of Corollary \ref{blacksabbath3} is linear in $a\in C^\infty_c(S^*X)$, it is sufficient to prove it only for $a$ supported in a small open set. In particular, we may suppose that $\chi_a$ is supported in a small enough set so that Theorem \ref{ibrahim2} applies.

We then have (with $r=1, \ell = 0$):

\begin{equation*}
\begin{aligned}
\big{\langle} Op_h(a) \chi_a E_h,\chi_a E_h \big{\rangle}  &=\Big{\langle} Op_h(a) \sum_{\substack{\tilde{\beta}\in \tilde{\mathcal{B}}^{\chi_a}\\
                \tilde{n}(\tilde{\beta})\leq \tilde{M}|\log h|}} e^{i \varphi_{\tilde{\beta}}(x)/h}
a_{\tilde{\beta}}(x;h), \sum_{\substack{\tilde{\beta}\in \tilde{\mathcal{B}}^{\chi_a}\\
                \tilde{n}(\tilde{\beta})\leq \tilde{M}|\log h|}} e^{i \varphi_{\tilde{\beta}}(x)/h}
a_{\tilde{\beta}}(x;h) \Big{\rangle} + O(h)\\
&= \sum_{\substack{\tilde{\beta}\in \tilde{\mathcal{B}}^{\chi_a}\\
                \tilde{n}(\tilde{\beta})\leq \tilde{M}|\log h|}} \big{\langle} Op_h(a) e^{i \varphi_{\tilde{\beta}}(x)/h}
a_{\tilde{\beta}}(x;h), e^{i \varphi_{\tilde{\beta}}(x)/h}
a_{\tilde{\beta}}(x;h) \big{\rangle} \\
&+ \sum_{\substack{\tilde{\beta}, \tilde{\beta}'\in \tilde{\mathcal{B}}^{\chi_a}\\
               \tilde{n}(\tilde{\beta}'),\tilde{n}(\tilde{\beta})\leq \tilde{M}|\log h|\\
                \tilde{\beta}\neq \tilde{\beta}'}} \big{\langle} Op_h(a) e^{i \varphi_{\tilde{\beta}}(x)/h}
a_{\tilde{\beta}}(x;h), e^{i \varphi_{\tilde{\beta}'}(x)/h}
a_{\tilde{\beta}'}(x;h) \big{\rangle} + O(h).
\end{aligned}
\end{equation*}

We want to use Proposition \ref{nonstat} to say that the second term corresponding to non-diagonal terms is a $O(h^\infty)$.

Thanks to (\ref{hurry2}), we know that if $\beta,\beta'\in \mathcal{B}^{\chi_a}$ are such that $\tilde{n}(\tilde{\beta}),\tilde{n}(\tilde{\beta}')\leq \big{(}\frac{1}{2 \sqrt{b_0}} - \epsilon\big{)} |\log h|$, then we have $|\partial \varphi_{\tilde{\beta}}(x)-\partial \varphi_{\tilde{\beta}'}(x)|\geq C h^{1/2-\epsilon}$, so that by Proposition \ref{nonstat}, we have that 
\begin{equation*}
{\langle} Op_h(a) e^{i \varphi_{\tilde{\beta}}(x)/h}
a_{\tilde{\beta}}(x;h), e^{i \varphi_{\tilde{\beta}'}(x)/h}
a_{\tilde{\beta}'}(x;h) \big{\rangle}=O(h^\infty).
\end{equation*}

Therefore, we have
\begin{equation*}
\begin{aligned}
\big{\langle} Op_h(a) \chi_a E_h,\chi_a E_h \big{\rangle}
&= \sum_{\substack{\tilde{\beta}\in \tilde{\mathcal{B}}^{\chi_a}}} \big{\langle} Op_h(a) e^{i \varphi_{\tilde{\beta}}(x)/h}
a_{\tilde{\beta}}(x;h), e^{i \varphi_{\tilde{\beta}}(x)/h}
a_{\tilde{\beta}}(x;h) \big{\rangle} \\
+ &\sum_{\substack{\tilde{\beta}, \tilde{\beta}'\in \tilde{\mathcal{B}}^{\chi_a}\\
              \tilde{n}(\tilde{\beta})\leq \tilde{M}|\log h|\\
             \tilde{n}(\tilde{\beta}) \text{ or }  \tilde{n}(\tilde{\beta}') \geq \big{(}\frac{1}{2 \sqrt{b_0}} - \epsilon\big{)} |\log h|  \\
                \tilde{\beta}\neq \tilde{\beta}'}} \big{\langle} Op_h(a) e^{i \varphi_{\tilde{\beta}}(x)/h}
a_{\tilde{\beta}}(x;h), e^{i \varphi_{\tilde{\beta}'}(x)/h}
a_{\tilde{\beta}'}(x;h) \big{\rangle}\\
&-\sum_{\substack{\tilde{\beta}\in \tilde{\mathcal{B}}^{\chi_a}\\
                \tilde{n}(\tilde{\beta})\geq \tilde{M}|\log h|}} \big{\langle} Op_h(a) e^{i \varphi_{\tilde{\beta}}(x)/h}
a_{\tilde{\beta}}(x;h), e^{i \varphi_{\tilde{\beta}}(x)/h}
a_{\tilde{\beta}}(x;h) \big{\rangle}+ O(h). 
\end{aligned}
\end{equation*}

We may estimate the second and third terms thanks to (\ref{sheriff3}). We obtain a remainder which is a $O\Big{(}h^{\min\big{(}1, \frac{|\mathcal{P}(1/2)|}{2|\sqrt{b_0}}-\epsilon \big{)}}\Big{)}$, which gives us

\begin{equation*}
\begin{aligned}
\big{\langle} Op_h(a) \chi_a E_h,\chi_a E_h \big{\rangle}
&= \sum_{\substack{\tilde{\beta}\in \tilde{\mathcal{B}}^{\chi_a}}} \big{\langle} Op_h(a) e^{i \varphi_{\tilde{\beta}}(x)/h}
a_{\tilde{\beta}}(x;h), e^{i \varphi_{\tilde{\beta}}(x)/h}
a_{\tilde{\beta}}(x;h) \big{\rangle}+O\Big{(}h^{\min\big{(}1, \frac{|\mathcal{P}(1/2)|}{2|\sqrt{b_0}}-\epsilon \big{)}}\Big{)}.
\end{aligned}
\end{equation*}

Each term in the first sum is then the Wigner distribution associated to a Lagrangian state, and the associated semiclassical measure may be computed by stationary phase, just as in  \cite[\S 5.1]{Zworski_2012} and this gives us the first part of the statement of Corollary \ref{blacksabbath3}. Let us now prove (\ref{beatles}).

Recall that the symbols $a_{\beta,\chi}(x;h)$ which appear in (\ref{faust}) are built by applying formula (\ref{shotgun}) several times to $E_h^0$ (see \cite[\S 5.2]{Ing} for details). 
In (\ref{shotgun}), the phase $\theta$ which appears depends only on the trajectory of the point $(x_1,\phi_1'(x))$, so that these phases are the same for $\beta$ and $\beta'$ if $\beta\sim_{\spt(\chi)} \beta'$. 
In particular, for all $\tilde{\beta}\in \tilde{\mathcal{B}}^\chi$, for all representative $\beta\in \mathcal{B}_n^\chi$ and for all $x\in X$, we have that
\begin{equation*}
|a_{\tilde{\beta}}|(x)\geq |a_{\beta,\chi}|(x).
\end{equation*}

 Therefore, the result will be proven if for every $N\in \mathbb{N}$, we can find a constant $c_N\bel 0$ such that for any $x\in X$ with $\chi(x)=1$, we have
\begin{equation}\label{xiangqi}
\sum_{n=N}^{\infty}\sum_{\beta\in \mathcal{B}^\chi_n} \sigma_h\big{(}|a_{\beta,\chi}|^2\big{)}(x) \geq c_N. 
\end{equation}

Furthermore, still from (\ref{shotgun}), we see that $\sigma_h(|a_{\beta,\chi}|(x;h))(x)\bel 0$ as long as there exists $\xi$ such that $(x,\xi)\in \Phi_{\beta}^{n}(\Lag)$.

But for each $x\in X$ such that $\chi(x)=1$, Corollary \ref{musso} gives us infinitely many $(\xi_i)_{i\in I}\subset T^*_x X$ and $t_i\bel 0$ such that $(x,\xi)\in \Phi^{t_i}(\Lag)$. In particular, for each of them, there is a $n_i\in\mathbb{N}$ and a $\beta_i\in \mathcal{B}_{n_i}^\chi$ such that $(x,\xi)\in \Phi^{n_i}_{\beta_i}(\Lag)$. 

Now, using Corollary \ref{dompteur2}, we see that if $i\neq i'$, we have $\beta_i\neq \beta_{i'}$ : otherwise, we could build two distinct geodesics staying at a distance less than $r_i$ for all negative times, and whose distance is $0$ at time zero, thus contradicting convexity. Hence, since each $\mathcal{B}_n^\chi$ is finite, there exists a $\beta\in \mathcal{B}_n^\chi$ for some $n\geq N$ and a $i\in I$ such that $\beta = \beta_i$

Therefore, $\sigma_h(|a_{\beta_i,\chi}|(x;h))(x)\bel 0$. By continuity, this is true uniformly in a neighbourhood of $x$. By compactness of $\spt \chi$, we obtain (\ref{xiangqi}).
\end{proof}

\section{Small-scale equidistribution}
Thanks to Corollary \ref{blacksabbath3} along with Corollary \ref{clocks}, we know that there exist constants $C_1,C_2\bel 0$ such that for any $x_0\in X$ such that $\chi(x_0)=1$, for any $r\bel 0$ small enough, we have
\begin{equation}\label{smallscale}
C_1 r^d\leq \int_{B(x_0,r)} |E_h|^2(x) \mathrm{d}x\leq C_2 r^d,
\end{equation}
where $B(x_0,r)$. The goal of this section is to show that (\ref{smallscale}) still holds if we replace $E_h$ by $\Re E_h$, and if $r$ depends on $h$, as long as $r \bel C h$ for some $C>0$ large enough. 

Let us start by showing small-scale equidistribution for $E_h$.

\subsection{Small-scale equidistribution for $E_h$}
The aim of this section is to show the following proposition.
\begin{proposition}\label{smallscale5}
Let $X$ be a manifold which satisfies Hypothesis \ref{Guepard} near infinity, and which satisfies Hypothesis \ref{Strong}. Suppose that the geodesic flow $(\Phi^t)$ satisfies
Hypothesis \ref{sieste} on hyperbolicity, Hypothesis \ref{Husserl} concerning the topological pressure. Let $E_h$ be a generalized eigenfunction of the form described in Hypothesis \ref{pied}, where $E_h^0$ is associated to a Lagrangian manifold $\Lag$ which satisfies
the invariance Hypothesis \ref{chaise}, part (iii) of
Hypothesis \ref{Strong} and Hypotheses \ref{puissanceplus} and \ref{puissanceplus2}.

 Let $\chi\in C_c^\infty(X)$. Then there exist constants $C,C_1, C_2\bel 0$ such that the following holds. For any $x_0\in X$ such that $\chi(x_0)=1$, for any sequence $r_h$ such that $1\bel \bel r_h \bel C h$, we have for $h$ small enough:
\begin{equation}\label{smallscale2}
C_1 r_h^d\leq \int_{B(x_0,r_h)} |E_h|^2(x) \mathrm{d}x\leq C_2 r_h^d.
\end{equation}
\end{proposition}

\begin{proof} 
The upper bound is a direct consequence of Corollary \ref{clocks}. Let us explain why the lower bound holds.

As in the proof of Corollary \ref{blacksabbath3},
$|E_h|^2(x)$ can be written, up to a remainder of order $O(\epsilon)$ as a sum of terms of the form $a_{\beta,\beta'}(x;h) e^{i\phi_{\beta,\beta'}(x)/h}$, such that $\sum_{\beta,\beta'} |a_{\beta,\beta'}(x)|\leq C_0$ for some $C_0>0$. 

Furthermore, there exists $c(\varepsilon)>0$ such that for all $\beta,\beta'$ satisfying $\partial \phi_{\beta,\beta'}(x_0)\neq 0$, we have $|\partial \phi_{\beta,\beta'}(x_0)|>c(\varepsilon)$. 

Consider a term where $\partial \phi_{\beta,\beta'}(x_0)\neq 0$.
 By a change of variables and by Stokes theorem, we have
 
\begin{equation*}
\begin{aligned}
\int_{B(x_0,r_h)} a_{\beta,\beta'}(x;h) e^{i\phi_{\beta,\beta'}(x)/h} \mathrm{d}x &= r_h^d\int_{B(0,1)} a_{\beta,\beta'} (r_h x+x_0;h) e^{i \phi_{\beta,\beta'}(r_h x+x_0)/h} \mathrm{d}x\\
&=r_h^d\int_{B(0,1)} \Big{[}\frac{h}{r_h} \frac{a_{\beta,\beta'} (r_h x+x_0;h)}{|\partial \phi_{\beta,\beta'}(r_hx+x_0)|^2}\\
&\times \big{(}\partial \phi_{\beta,\beta'}(r_hx+x_0)\big{)}\cdot\partial\big{(}e^{i \phi_{\beta,\beta'}(r_h x+x_0)/h}\big{)} \Big{]} \mathrm{d}x \\
&= r_h^d \int_{\partial B(0,1)} \Big{[}\frac{h}{r_h} \frac{a_{\beta,\beta'} (r_h x+x_0;h)}{|\partial \phi_{\beta,\beta'}(r_hx+x_0)|^2} \\
&\times e^{i \phi_{\beta,\beta'}(r_h x+x_0)/h} \big{(}\partial \phi_{\beta,\beta'}(r_hx+x_0)\big{)}\Big{]}\cdot \mathrm{d} \vec{n}\\
& - r_h^d\int_{B(0,1)} \Big{[}\frac{h}{r_h} \partial\cdot\Big{(}\frac{a_{\beta,\beta'} (r_h x+x_0;h)}{|\partial \phi_{\beta,\beta'}(r_hx+x_0)|^2} \big{(}\partial \phi_{\beta,\beta'}(r_hx+x_0)\big{)}\Big{)}\\
&\times e^{i \phi_{\beta,\beta'}(r_h x+x_0)/h} \Big{]} \mathrm{d}x\\
&\leq c'(\varepsilon) h r_h^{d-1} \|a_{\beta,\beta'}(\cdot;h)\|_{C^1},
\end{aligned}
\end{equation*}
where $c'(\varepsilon)$ is a constant depending on $\varepsilon$, but not on $h$. By summing over $\beta$ and $\beta'$, we obtain that the sum of non-diagonal terms is bounded by $c'(\varepsilon) C_0 h r_h^{d-1} + r_h^d R_\varepsilon$, where $R_\varepsilon= O(\varepsilon)$. 

 On the other hand, thanks to (\ref{beatles}), the diagonal terms will give a contribution larger than $c_0 r_h^d$ for some $c_0\bel 0$ independent of $h$ and of $x$. By taking $\varepsilon$ small enough so that $R_\varepsilon< c_0/3$ and $C$ sufficiently small so that $c'(\varepsilon) C_0/C < \varepsilon/3$, we obtain the result.
\end{proof}
\begin{remarque}
Recall that theorem \ref{ibrahim2} tells us that 
\begin{equation*} \chi E_h(x) = \sum_{\substack{\tilde{\beta}\in \tilde{\mathcal{B}}^\chi\\
                \tilde{n}(\tilde{\beta})\leq \tilde{M}_{r,\ell}|\log h|}} e^{i \varphi_{\tilde{\beta}}(x)/h}
a_{\tilde{\beta},\chi}(x;h) + R_{r,\ell},
\end{equation*}

Since equation (\ref{beatles}) is still true if we limit ourselves to $\tilde{\beta}$ such that $\tilde{n}(\tilde{\beta})>N$ for any $N\geq 0$, the proof above can be easily adapted to give the following result.

For any $N>0$, there exist $C'_0, C'_1, C'_2>0$ such that for any $x_0\in X$ with $\chi(x_0)=1$, for any sequence $r_h$ such that $1\bel \bel r_h \bel C'_0 h$, we have for $h$ sufficiently small:
\begin{equation}\label{smallscale9}
C'_1 r_h^d\leq \int_{B(x_0,r_h)} \Big{|}\sum_{\substack{\tilde{\beta}\in \tilde{\mathcal{B}}^\chi\\
                N\leq \tilde{n}(\tilde{\beta})\leq \tilde{M}_{r,\ell}|\log h|}} e^{i \varphi_{\tilde{\beta}}(x)/h}
a_{\tilde{\beta},\chi}(x;h) + R_{r,\ell}(x)\Big{|}^2 \mathrm{d}x\leq C'_2 r_h^d.
\end{equation}
\end{remarque}

\subsection{Small-scale equidistribution for $\Re E_h$}\label{preuveequidpetite}
The aim of this section is to prove Theorem \ref{nuque}, which we now recall.

\begin{theo}
Let $X$ be a manifold which satisfies Hypothesis \ref{Guepard} near infinity, and which satisfies Hypothesis \ref{Strong}. Suppose that the geodesic flow $(\Phi^t)$ satisfies
Hypothesis \ref{sieste} on hyperbolicity, Hypothesis \ref{Husserl} concerning the topological pressure. Let $E_h$ be a generalized eigenfunction of the form described in Hypothesis \ref{pied}, where $E_h^0$ is associated to a Lagrangian manifold $\Lag$ which satisfies
the invariance Hypothesis \ref{chaise}, part (iii) of
Hypothesis \ref{Strong} and Hypotheses \ref{puissanceplus} and \ref{puissanceplus2}.

Let $\chi\in C_c^\infty(X)$. Then there exist constants $C, C_1, C_2\bel 0$ such that the following result holds. For all $x_0\in X$ such that $\chi(x_0)=1$, for any sequence $r_h$ such that $1\bel \bel r_h \bel C h$, we have for $h$ small enough:
\begin{equation}\label{smallscale20}
C_1 r_h^d\leq \int_{B(x_0,r_h)} |\Re E_h|^2(x) \mathrm{d}x\leq C_2 r_h^d.
\end{equation}

In particular, for any bounded open set $U\subset X$, there exists $c(U)\bel 0$ and $h_U\bel 0$ such that for all $0<h<h_U$, we have
\begin{equation}\label{bientotlafin}
\int_U |\Re E_h|^2\geq c(U).
\end{equation}
\end{theo}

To prove Theorem \ref{nuque}, we first need to prove the following lemma, which says that few trajectories “go backwards”.

\begin{lemme}\label{youhou}
There exists $N\in\mathbb{N}$ such that for any $n\geq N$, for any $n'\in \mathbb{N}$ and for any $\beta\in \mathcal{B}^{\tilde{\chi}}_n$, $\beta'\in \mathcal{B}^{\tilde{\chi}}_{n'}$, the following holds.

Suppose $x\in \spt(\chi)$ is such that $\partial\varphi_{\beta}(x)=-\partial \varphi_{\beta'}(x)$. Then there exists a small neighbourhood $V$ of $x$ such that for all $y\in V$, we have
\begin{equation*}
\big{(}\partial\varphi_{\beta}(y) = - \partial \varphi_{\beta'}(y)\big{)} \Longrightarrow \big{(}y,\partial \phi_{\beta}(y)\big{)}=\Phi^t\big{(}x,\partial \phi_{\beta}(x)\big{)} ~~\text{ for some } t \text{ small enough.}
\end{equation*}
\end{lemme}
In the proof, we shall use the following notation. 
If $\rho=(x,\xi)\in \mathcal{E}$, we shall write $$\rho^{\mathrm{t}}=(x,-\xi).$$
\begin{proof}
Let us consider integers $n,n'$ such that the statement above is false, and show that there can be only finitely many values for $n$.
For such $n,n'$, there exists a sequence $y_k$ converging to some $x\in \spt(\chi)$ such that $\partial\varphi_{\beta}(y_k) = - \partial \varphi_{\beta'}(y_k)$, and such that $\rho_k$ and $\rho$ do not belong to the same trajectory, where we write $\rho_k:= (y_k, \partial\varphi_{\beta'}(y_k))$ and $\rho= (x, \partial\varphi_{\beta'}(x))$.

Write $N'=1+\max(n,n')$. Thanks to (\ref{flowers}), by taking $k$ large enough, we may assume that $d_{X}\big{(} \Phi^{-t}(\rho_k),\Phi^{-t}(\rho) \big{)}< r_i$ for $|t|\leq N'$. On the other hand, since $\Phi^{-N'}(\rho_k),\Phi^{-N'}(\rho)\in \Lag$, the point (iii) in Hypothesis \ref{Strong} ensures us that for $t\leq -N'$, we also have $d_{X}\big{(} \Phi^{-t}(\rho_k),\Phi^{t}(\rho) \big{)}< r_i$.

We also have $\big{(}\Phi^{N'}(\rho_k)\big{)}^{\mathrm{t}},\big{(}\Phi^{N'}(\rho)\big{)}^{\mathrm{t}}\in \Lag$, so that  $d_{X}\big{(} \Phi^{t}(\rho_k),\Phi^{-t}(\rho) \big{)}< r_i$ for all $t\geq N'$.

Therefore, by Corollary \ref{dompteur2}, we see that $d_{X}\big{(} \Phi^{t}(\rho_k),\Phi^{t}(\rho) \big{)}$ must be constant with respect to time.

If $n$ is large enough, then there will be a $t\bel 0$ such that $\Phi^{-t}(\rho)$ and $\Phi^{-t}(\rho_k)$ are in a small neighbourhood of the trapped set $K$. But then, because of hyperbolicity, the two trajectories $\Phi^{t}(\rho_k)$ and $\Phi^{t}(\rho)$ cannot remain at a constant distance from each other, unless they belong to the same trajectory. This proves the lemma.
\end{proof}

\begin{proof}[Proof of Theorem \ref{nuque}]
Let us apply Theorem \ref{ibrahim2} (for $r=1$, $\ell=0$) with $\chi\equiv 1$ on $U$. We have, for $x\in U$:
\begin{equation} E_h(x) = \sum_{n=0}^{\lfloor
M |\log h|\rfloor}
\sum_{\tilde{\beta}\in \tilde{\mathcal{B}}^{\hat{\chi}}_n} e^{i \varphi_{\tilde{\beta}}(x)/h}
a_{\tilde{\beta},U}(x;h) + O_{L^\infty}(h) =: S_1 + S_2 + O_{L^\infty}(h),
\end{equation}
where $S_1$ denotes the sum for $\tilde{n}(\tilde{\beta})\leq N$, and $S_2$ the sum for $\tilde{n}(\tilde{\beta})\bel N$, where $N$ is as in Lemma \ref{youhou}.

We may write 
\begin{equation}\label{claudebernard}
\begin{aligned}
|\Re E_h|^2&= (\Re (S_1) + (S_2+ \overline{S_2})/2)^2 +O_{L^\infty}(h)\\
&= (\Re S_1)^2 + \frac{S_2^2}{4} + \frac{\overline{S_2^2}}{4} + \frac{|S_2|^2}{2} + 2\Re (S_1)\Re (S_2)+O_{L^\infty}(h).
\end{aligned}
\end{equation}

When we compute $\int_{B(x_0,r_h)} |\Re E_h|^2$, the term $\int_{B(x_0,r_h)}   (\Re S_1)^2 + \frac{S_2^2}{4} + \frac{\overline{S_2^2}}{4}$ is of course non-negative. As for the term $\int_{B(x_0,r_h)} |S_2|^2$, we may use (\ref{smallscale9}) to find a constant $c\bel 0$ independent of $h$ such that
\begin{equation*}
\int_{B(x_0,r_h)} |S_2|^2 \geq c r_h^d + o_{h\rightarrow 0}(r_h^d).
\end{equation*}

Let us show that 
\begin{equation}\label{petitreste}
 \Big{|}\int_{B(x_0,r_h)} 2\Re(S_2) \Re(S_1) \Big{|} \leq \frac{c r_h^d}{2}.
 \end{equation}

The function $2\Re(S_2) \Re(S_1)$ can be written, up to a remainder of size $\varepsilon$, like a finite sum of oscillating terms with phases $\pm\varphi_{\tilde{\beta}}(x) \pm \varphi_{\tilde{\beta}'}(x)$ with $N< \tilde{n}(\tilde{\beta})\leq N'(\varepsilon)$ and $\tilde{n}(\tilde{\beta'})\leq N'(\varepsilon)$.
Thanks to Lemma \ref{youhou}, we know that for each $\beta'$, the phase can only cancel on a finite number of curves.

Just as in the proof of proposition \ref{smallscale5}, we rewrite
\begin{equation*}
\begin{aligned}
&\frac{1}{r_h^d}\int_{B(x_0,r_h)}e^{i \varphi_{\tilde{\beta}}(x)/h}
a_{\tilde{\beta},\chi}(x; h)  e^{i \pm \varphi_{\tilde{\beta'}}(x)/h}
a_{\tilde{\beta'},\chi}(x;h) \mathrm{d}x\\
&= \int_{B(0,1)} e^{i \varphi_{\tilde{\beta}}(r_h x+x_0)/h}
a_{\tilde{\beta},\chi}(r_h x+x_0; h)  e^{i \pm \varphi_{\tilde{\beta'}}(r_h x+x_0)/h}
a_{\tilde{\beta'},\chi}(r_h x+x_0; h) \mathrm{d}x
\end{aligned}
\end{equation*}

By the method of stationary phase, we obtain that each of these integrals is bounded by $c_\varepsilon o_{h\rightarrow 0}\Big{(}\Big{(}\frac{h}{r_h}\Big{)}^{(d-1)/2}\Big{)}$.

By taking $\varepsilon= c/4$, and then $C$ large enough, we obtain  (\ref{petitreste}). The lemma follows.

\end{proof}
\subsection{Lower bound on the nodal volume}\label{lowernod}
The aim of this section is to prove the following corollary of Theorem \ref{nuque}.
\begin{corolaire} \label{nodality}
Let $X$ be a manifold which satisfies Hypothesis \ref{Guepard} near infinity, and which satisfies Hypothesis \ref{Strong}. Suppose that the geodesic flow $(\Phi^t)$ satisfies
Hypothesis \ref{sieste} on hyperbolicity, Hypothesis \ref{Husserl} concerning the topological pressure. Let $E_h$ be a generalized eigenfunction of the form described in Hypothesis \ref{pied}, where $E_h^0$ is associated to a Lagrangian manifold $\Lag$ which satisfies
the invariance Hypothesis \ref{chaise}, part (iii) of
Hypothesis \ref{Strong} and Hypotheses \ref{puissanceplus} and \ref{puissanceplus2}.

Fix $\mathcal{K}$ a compact subset of $X$. There exists $C_{\mathcal{K}}$ such that

$$Haus_{d-1} (\{x\in \mathcal{K}; \Re(E_h)(x)=0\}) \geq \frac{C_\mathcal{K}}{h}.$$
\end{corolaire}
\begin{proof}
Let us fix $U\subset X$ a bounded open set which contains $\mathcal{K}$.
The proof relies on the so-called Dong-Sogge-Zelditch formula, which we recall.
Let $f\in C_0^\infty(X)$, and let $g_h\in C^\infty(X)$ be a family of smooth functions such that $(-h^2\Delta+1)g_h=0$. Let us write $Z_h:=\{x\in X; g_h(x)=0\}$. Then, as proven in \cite{Dong}, \cite{SoZe}\footnote{This formula was proven on compact manifolds, but the proof works exactly the same on noncompact manifolds, as long as the function $f$ is compactly supported.} we have
\begin{equation}\label{DSZ}
\int_X \big{(}(-h^2\Delta+1)f\big{)}|g_h| \mathrm{d}x= h^2\int_{Z_h} f |\nabla g_h|\mathrm{d}S,
\end{equation}
where $\mathrm{d}S$ denotes the $(d-1)$-dimensional Hausdorff measure.

Let $f\in C_c^\infty(X;[0,1])$ be such that $\spt (f)\subset \mathcal{K}$, and $f$ is not identically zero.
Applying (\ref{DSZ}) with $g_h = \Re(E_h)$, we have
\begin{equation}\label{kremlin}
\begin{aligned}
\int_X \big{(}(-h^2\Delta+1)f\big{)}|\Re(E_h)|\mathrm{d}x&= h^2\int_{Z_h} f |\nabla \Re(E_h)|\mathrm{d}S\\
&\leq C h \int_{Z_h} f \mathrm{d}S\\
&\leq C h Haus_{d-1} (\{x\in \mathcal{K}; \Re(E_h)(x)=0\}),
\end{aligned}
\end{equation}
where we used Theorem \ref{linfini} to go from the first to the second line.
Therefore, the theorem will be proven if we can find a lower bound on the left-hand side in (\ref{kremlin}) uniform in $h$.

Since $f$ is not identically zero, we may find a constant $c\bel 0$ and an open set $U\subset X$ such that $f\geq c$ on $U$. We have
\begin{equation*}
\begin{aligned}
\int_X \big{(}(-h^2\Delta+1)f\big{)}|\Re(E_h)|&= \int_X f |\Re(E_h)| + O(h^2) \\
&\geq c\int_U |\Re (E_h)| + O(h^2)\\
&\geq \frac{c}{\|E_h\|_{C^0(U)}} \int_U |\Re(E_h)|^2 \\
&\geq C \int_U |\Re(E_h)|^2
\end{aligned}
\end{equation*}
thanks to Theorem \ref{linfini}. Therefore, Theorem \ref{nodality} follows from (\ref{bientotlafin}).
\end{proof}

\begin{appendices}
  \renewcommand\thetable{\thesection\arabic{table}}
  \renewcommand\thefigure{\thesection\arabic{figure}}
  \section{Reminder of semiclassical analysis} \label{pivoine}
\subsection{Pseudodifferential calculus} \label{greve}

Let $Y$ be a Riemannian manifold.
We will say that a function $a(x,\xi;h)\in C^{\infty}(T^*Y\times (0,1])$ is in the class $S^{comp}(T^*Y)$ if it can be written as
\begin{equation*}
a(x,\xi;h)= \tilde{a}_h(x,\xi) + O\Big{(}\Big{(}\frac{h}{\langle \xi \rangle}\Big{)}^\infty \Big{)},
\end{equation*}
where the functions $\tilde{a}_h\in C_c^\infty(T^*Y)$ have all their semi-norms and supports bounded independently of $h$. If $U$ is an open set of $T^*Y$, we will sometimes write $S^{comp}(U)$ for the set of funtions  $a$ in $S^{comp}(T^*Y)$ such that for any $h\in ]0,1]$, $\tilde{a}_h$ has its support in $U$. 

\begin{definition}\label{defsymbclassique}
Let $a\in S^{comp}(T^*Y)$. We will say that $a$ is a \emph{classical symbol} if there exists a sequence of symbols $a_k\in S^{comp}(T^*Y)$ such that for any $n\in \mathbb{N}$, 
$$a-\sum_{k=0}^n h^k a_k \in h^{n+1} S^{comp}(T^*Y).$$
We will then write $a^0(x,\xi):= \lim\limits_{h\rightarrow 0} a(x,\xi;h)$ for the \emph{principal symbol} of $a$.
\end{definition}

We will sometimes write that $a\in S^{comp}(Y)$ if it can be written as 
\begin{equation*}
a(x;h)= \tilde{a}_h(x) + O\big{(}h^\infty \big{)},
\end{equation*}
where the functions $\tilde{a}_h\in C_c^\infty(Y)$ have all their semi-norms and supports bounded independently of $h$.

We associate to $S^{comp}(T^*X)$ the class of pseudodifferential operators
$\Psi_h^{comp}(X)$, through a surjective quantization map
\begin{equation*}Op_h:S^{comp}(T^*X)\longrightarrow \Psi^{comp}_h(X).
\end{equation*} This quantization
map is defined using coordinate charts, and the standard Weyl quantization
on $\mathbb{R}^d$. It is therefore not intrinsic. However, the principal
symbol map
\begin{equation*}\sigma_h : \Psi^{comp}_h (X)\longrightarrow S^{comp}(T^*X)/
h S^{comp}(T^*X)
\end{equation*} is intrinsic, and we have
\begin{equation*}\sigma_h(A\circ B) = \sigma_h (A) \sigma_h(B)
\end{equation*}
and
\begin{equation*}\sigma_h\circ Op: S^{comp}(T^* X) \longrightarrow S^{comp} (T^*X) /h
S^{comp}(T^*X)
\end{equation*}
is the natural projection map.

For more details on all these maps and their construction, we refer the reader
to \cite[Chapter 14]{Zworski_2012}.

For $a\in S^{comp}(T^*X)$, we say that its essential support is equal to a given
compact $K\Subset T^*X$,
\begin{equation*} \text{ ess supp}_h a = K \Subset T^*X,
\end{equation*}
if and only if, for all $\chi \in S(T^*X)$,
\begin{equation*}\spt \chi \subset (T^*X\backslash K) \Rightarrow \chi a \in h^\infty S(T^*
X).
\end{equation*}
For $A\in \Psi^{comp}_h(X), A=Op_h(a)$, we define the wave front set of $A$ as:
\begin{equation*}WF_h(A)= \text{ ess supp}_h a,
\end{equation*}
noting that this definition does not depend on the choice of the
quantization. When $K$ is a compact subset of $T^*X$ and $WF_h(A)\subset K$, we will sometimes say that $A$ is \emph{microsupported} inside $K$.

Let us now state a lemma which is a consequence of Egorov theorem \cite[Theorem 11.1]{Zworski_2012}. Recall that $U(t)$ is the Schrödinger propagator $U(t)= e^{it P_h/h}$.
\begin{lemme}\label{theclash}
Let $A,B\in \Psi_h^{comp}(X)$, and suppose that $\Phi^t( WF_h(A))\cap WF_h(B)=\emptyset$. Then we have
\begin{equation*}
A U(t) B= O_{L^2\rightarrow L^2}(h^\infty).
\end{equation*}
\end{lemme}

If $U,V$ are bounded open subsets of $T^*X$, and if $T,T' : L^2(X)\rightarrow L^2(X)$ are bounded operators, we shall say that $T\equiv T'$ \emph{microlocally} near $U\times V$ if there exist bounded open sets $\tilde{U}\supset \overline{U}$ and $\tilde{V} \supset \overline{V}$ such that for any $A,B\in \Psi_h^{comp}(X)$ with $WF(A)\subset \tilde{U}$ and $WF(B)\subset \tilde{V}$, we have
\begin{equation*}
A(T-T')B = O_{L^2\rightarrow L^2} (h^\infty)
\end{equation*}

\paragraph{Tempered distributions}
Let $u=(u(h))$ be an $h$-dependent family of distributions in $\mathcal{D}'(X)$. We say it is \emph{$h$-tempered} if for any bounded open set $U\subset X$, there exists $C\bel 0$ and $N\in \mathbb{N}$ such that
\begin{equation*}
\|u(h)\|_{H_h^{-N}(U)}\leq C h^{-N},
\end{equation*}
where $\|\cdot\|_{H_h^{-N}(U)}$ is the semiclassical Sobolev norm.

For a tempered distribution $u=(u(h))$, we say that a point $\rho\in T^*X$ does not lie in the wave front set $WF(u)$ if there exists a neighbourhood $V$ of $\rho$ in $T^*X$ such that for any $A\in \Psi_h^{comp}(X)$ with $WF(a)\subset V$, we have $Au=O(h^\infty)$.
\subsection{Lagrangian distributions and Fourier Integral Operators}\label{DSK}
\paragraph{Phase functions}
Let $\phi(x,\theta)$ be a smooth real-valued function on some open subset $U_\phi$ of $X\times \mathbb{R}^L$, for some $L\in \mathbb{N}$. We call $x$ the \emph{base variables} and $\theta$ the \emph{oscillatory variables}. We say that $\phi$ is a \emph{nondegenerate phase function} if the differentials $d(\partial_{\theta_1} \phi)...d(\partial_{\theta_L}\phi)$ are linearly independent on the \emph{critical set }
\begin{equation*}
C_\phi:=\{ (x,\theta); \partial_\theta \phi =0 \} \subset U_\phi.
\end{equation*}
In this case
\begin{equation*}
\Lambda_\phi:= \{(x,\partial_x \phi(x,\theta)); (x,\theta)\in C_\phi \} \subset T^*X
\end{equation*}
is an immersed Lagrangian manifold. By shrinking the domain of $\phi$, we can make it an embedded Lagrangian manifold. We say that $\phi$ \emph{generates } $\Lambda_\phi$.

\paragraph{Lagrangian distributions}
Given a phase function $\phi$ and a symbol $a\in S^{comp}(U_\phi)$, consider the $h$-dependent family of functions
\begin{equation}\label{massai}
u(x;h)= h^{-L/2} \int_{\mathbb{R}^L} e^{i\phi(x,\theta)/h} a(x,\theta;h) \mathrm{d}\theta.
\end{equation}
We call $u=(u(h))$ a \emph{Lagrangian distribution}, (or a \emph{Lagrangian state}) generated by $\phi$. By the method of non-stationary phase, if $\spt  (a)$ is contained in some $h$-independent compact set $K\subset U_\phi$, then
\begin{equation*}
WF_h(u)\subset \{(x,\partial_x \phi(x,\theta)); (x,\theta)\in C_\phi\cap K\}\subset \Lambda_\phi.
\end{equation*}

\begin{definition}\label{Grenoble}
Let $\Lambda\subset T^*X$ be an embedded Lagrangian submanifold. We say that an $h$-dependent family of functions $u(x;h)\in C_c^\infty(X)$ is a (compactly supported and compactly microlocalized) \emph{Lagrangian distribution associated to $\Lambda$}, if it can be written as a sum of finitely many functions of the form (\ref{massai}), for different phase functions $\phi$ parametrizing open subsets of $\Lambda$, plus an $O(h^\infty)$ remainder. We will denote by $I^{comp}(\Lambda)$ the space of all such functions.
\end{definition}

\paragraph{Fourier integral operators}
Let $X, X'$ be two manifolds of the same dimension $d$, and let $\kappa$ be a symplectomorphism from an open subset of $T^*X$ to an open subset of $T^*X'$. Consider the Lagrangian
\begin{equation*}
\Lambda_\kappa =\{(x,\nu;x',-\nu'); \kappa(x,\nu)=(x',\nu')\}\subset T^*X\times T^*X'= T^*(X\times X').
\end{equation*}
A compactly supported operator $U:\mathcal{D}'(X)\rightarrow C_c^\infty(X')$ is called a (semiclassical) \emph{Fourier integral operator} associated to $\kappa$ if its Schwartz kernel $K_U(x,x')$ lies in $h^{-d/2}I^{comp}(\Lambda_\kappa)$. We write $U\in I^{comp}(\kappa)$. The $h^{-d/2}$ factor is explained as follows: the normalization for Lagrangian distributions is chosen so that $\|u\|_{L^2}\sim 1$, while the normalization for Fourier integral operators is chosen so that $\|U\|_{L^2(X)\rightarrow L^2(X')} \sim 1$.

Note that if $\kappa\circ \kappa'$ is well defined, and if $U\in I^{comp}(\kappa)$ and $U'\in I^{comp}(\kappa')$, then $U\circ U'\in I^{comp} (\kappa\circ \kappa')$.

If $U\in I^{comp}(\kappa)$ and $O\subset T^*X$ is an open bounded set, we shall say that $U$ is \emph{microlocally unitary }near $O$ if $U^* U \equiv I_{L^2(X)\rightarrow L^2(X)}$ microlocally near $O\times \kappa (O)$.

\subsection{Local properties of Fourier integral operators}\label{chaussette}
In this section we shall see that, if we work locally, we may describe many Fourier integral operators without the help of oscillatory coordinates. In particular, following \cite[4.1]{NZ}, we will recall the effect of a Fourier integral operator on a Lagrangian distribution which has no caustics. 

Let $\kappa :T^*\mathbb{R}^d\rightarrow T^*\mathbb{R}^d$ be a local symplectic diffeomorphism. By performing phase-space translations, we may assume that $\kappa$ is defined in a neighbourhood of $(0,0)$ and that $\kappa(0,0)=(0,0)$. We furthermore assume that $\kappa$
is such that the projection from the graph of $\kappa$
\begin{equation}\label{tantine}
T^*\mathbb{R}^d\times T^*\mathbb{R}^d \ni (x^1,\xi^1; x^0,\xi^0)\mapsto (x^1,\xi^0)\in \mathbb{R}^d\times \mathbb{R}^d,~~ (x^1,\xi^1)= \kappa(x^0,\xi^0),
\end{equation}
is a diffeomorphism near the origin. Note that this is equivalent to asking that
\begin{equation}\label{minhlien}
\text{the } n\times n \text{ block } (\partial x^1/\partial x^0) \text{ in the tangent map } d\kappa(0,0) \text{ is invertible}.
\end{equation}

 It then follows that there exists a unique function ${\psi}\in C^\infty(\mathbb{R}^d\times \mathbb{R}^d)$ such that for $(x^1,\xi^0)$ near $(0,0)$,
\begin{equation*}
\kappa(\partial_\xi\psi(x^1,\xi^0),\xi^0)=(x^1,\partial_x\psi(x^1,\xi^0)),~~ \det \partial^2_{x,\xi}\psi\neq 0 \text{ and } \psi(0,0)=0.
\end{equation*}
The function $\psi$ is said to \emph{generate} the transformation $\kappa$ near $(0,0)$.

Thanks to assumption (\ref{tantine}), a Fourier integral operator $T\in I^{comp}(\kappa)$ may then be written in the form
\begin{equation}\label{cookie}
T u(x^1):= \frac{1}{(2\pi h)^d}\int \int_{\mathbb{R}^{2n}} e^{i(\psi(x^1,\xi^0)-\langle x^0,\xi^0 \rangle/h} \alpha (x^1,\xi^0;h) u(x^0) \mathrm{d}x\mathrm{d}x^0 \mathrm{d}\xi^0,
\end{equation}
with $\alpha\in S^{comp}(\mathbb{R}^{2d})$.

Now, let us state a lemma which was proven in \cite[Lemma 4.1]{NZ}, and which describes the effect of a Fourier integral operator of the form (\ref{cookie}) on a Lagrangian distribution which projects on the base manifold without caustics.

\begin{lemme}\label{pacal}
Consider a Lagrangian $\Lambda_0=\{(x_0,\phi_0'(x_0)); x\in \Omega_0\},
\phi_0\in C_b^\infty(\Omega_0)$, contained in a small neighbourhood
$V\subset T^*\mathbb{R}^d$ such that $\kappa$ is generated by $\psi$ near
$V$. We assume that
\begin{equation}\label{beignet}
\kappa(\Lambda_0)=\Lambda_1 = \{(x,\phi_1'(x)); x\in
\Omega_1\},~~\phi_1\in C_b^\infty(\Omega_1).
\end{equation}
Then, for any symbol $a\in S^{comp}(\Omega_0)$,
 the application of a Fourier integral operator $T$ of the form (\ref{cookie}) to the Lagrangian state
\begin{equation*}
a(x) e^{i\phi_0(x)/h}
\end{equation*}
associated with $\Lambda_0$ can be expanded, for any $L \bel 0$, into
\begin{equation*}
T (a e^{i\phi_0/h})(x) = e^{i\phi_1(x)/h} \Big{(} \sum_{j=0}^{L-1} b_j(x)
h^j+ h^L r_L(x,h) \Big{)},
\end{equation*}
where $b_j\in S^{comp}$, and for any $\ell\in \mathbb{N}$, we have
\begin{equation*}
\begin{aligned}
\|b_j\|_{C^\ell(\Omega_1)}&\leq C_{\ell,j}
\|a\|_{C^{\ell+2j}(\Omega_0)},~~~~0\leq j\leq L-1,\\
\|r_L(\cdot,h)\|_{C^\ell(\Omega_1)}&\leq C_{\ell,L}
\|a\|_{C^{\ell+2L+n}(\Omega_0)}.
\end{aligned}
\end{equation*}
The constants $C_{\ell,j}$ depend only on $\kappa$, $\alpha$ and
$\sup_{\Omega_0} |\partial^\beta \phi_0|$ for $0<|\beta|\leq 2\ell +j$.
Furthermore, if we write $g :\Omega_1\ni x \mapsto g(x):= \pi\circ \kappa^{-1} (x,\phi_1'(x))\in \Omega_0$, the principal symbol $b_0$ satisfies
\begin{equation}\label{shotgun}
b_0(x^1)= e^{i\theta/h}\frac{\alpha_0(x^1,\xi^0)}{|\det \psi_{x,\xi}(x^1,\xi^0)|^{1/2}}|\det dg(x^1)|^{1/2} a\circ g(x^1),
\end{equation}
where $\xi_0= \phi_0'\circ g(x^1)$ and where $\theta\in \mathbb{R}$.
\end{lemme}

\section{Hyperbolic near infinity manifolds}\label{portugal}

In this appendix, we will explain how the results of \cite{Ing} and of the present paper apply in the case of hyperbolic near infinity manifolds. Namely, we have to check that the hypotheses of Section \ref{assumptions} which depend on the structure of the manifold at infinity are satisfied in the case of Hyperbolic near infinity manifolds. More precisely, we will have to build distorted plane waves $E_h$ so that Hypothesis \ref{pied} is satisfied. This section partly follows \cite[\S 7]{DG}.

\begin{definition} We say that $X$ is \emph{hyperbolic near infinity} if it fulfils
the Hypothesis \ref{Guepard}, and if furthermore, in a collar
neighbourhood of $\partial \overline{X}$, the metric $g$ has sectional
curvature $-1$ and can be put in the form
\begin{equation*}g=\frac{db^2+h(b)}{b^2},
\end{equation*}
where $b$ is the boundary defining function introduced in section
\ref{Hector}, and where $h(b)$ is a smooth 1-parameter family of metrics on $\partial\overline{X}$ for $b\in [0,\epsilon)$.
\end{definition}

\paragraph{Construction of $E_h^0$}
Let us fix a $\xi \in \partial \overline{X}$. By definition of a hyperbolic near
infinity manifold, there exists a neighbourhood $\mathcal{V}_\xi$ of $\xi$
in $\overline{X}$ and an isometric diffeomorphism $\psi_\xi$ from
$\mathcal{V}_\xi\cap X$ into a neighbourhood $V_{q_0,\delta}$ of the
north pole $q_0$ in the unit ball $\mathbb{B}:=\{q\in \mathbb{R}^d;
|q|<1\}$ equipped with the metric $g_0$:
\begin{equation*}V_{q_0,\delta} :=\{q\in \mathbb{B}; |q-q_0|<\delta\},~~~~ g_0=\frac{4
dq^2}{(1-|q|^2)^2},
\end{equation*}
where $\psi_\xi(\xi)=q_0$, and $|\cdot|$ denotes the Euclidean length. We
shall choose the boundary defining function on the ball $\mathbb{B}$ to be
\begin{equation}\label{clovis}
b_0=2\frac{1-|q|}{1+|q|},
\end{equation}
and the induced metric $b_0^2 g_0|_{\mathbb{S}^d}$ on $\mathbb{S}^d=
\partial \mathbb{B}$ is the usual one with curvature $+1$. The function
$b_\xi:= b_0\circ \psi_\xi^{-1}$ can be viewed locally as a boundary defining function on $X$.

For each $p\in \mathbb{S}^d$, we define the Busemann function on $\mathbb{B}$
\begin{equation*}\phi_p^{\mathbb{B}}(q)=\log \Big{(}\frac{1-|q|^2}{|q-p|^2}\Big{)}.
\end{equation*}

Note that we have $\lim\limits_{p\rightarrow q}\phi_p^{\mathbb{B}}(q)= +\infty$.

There exists an $\epsilon>0$ such that the set
\begin{equation*}U_\xi := \{x\in \overline{X}; d_{\overline{g}(x,\xi)<\epsilon}\}
\end{equation*}
lies inside $\mathcal{V}_{\xi}$, where $\overline{g}= \bdf^2g$ is the
compactified
metric.
We define the function
\begin{equation*}\phi_\xi(x):=
\phi^{\mathbb{B}}_{q_0}
\big{(}\psi_{\xi}(x)\big{)}, ~~\text{    for   } x\in U_{\xi}, ~~ 0 \text{ otherwise}.
\end{equation*}

Let $\hat{\chi}:\overline{X}\longrightarrow [0,1]$ be a smooth function which
vanishes outside of $U_\xi$, which is equal to one in a neighbourhood
of $\xi$.

The incoming wave is then defined as
\begin{equation*}E_h^0(x;\xi) := \tilde{\chi}(x) e^{((d-1)/2+i/h)\phi_\xi(x)}~~
\text{ if } x\in U_\xi ,~~~~0 \text{ otherwise.} 
\end{equation*}

$E_h^0$ is then a Lagrangian state associated to the Lagrangian manifold 
\begin{equation*}
\Lag':= \{ (x,\partial_x \phi_\xi(x)),  x\in U_\xi\},
\end{equation*}
represented on figure \ref{varietehyp}.
$\Lag'$ satisfies the invariance hypothesis (\ref{invfut}), as can be easily seen by working in $\mathbb{B}$, but it will not satisfy hypothesis (\ref{invpass}). To obtain a manifold $\Lag$ which satisfies hypothesis (\ref{invpass}) from  $\Lag'$, we just have to continue propagating the points of $\Lag'$ which are already in $\mathcal{DE}_+$, that is to say, which go directly to infinity in the future, as follows :

$$\Lag:= \Lag' \cup \bigcup_{t\geq 0} \bigcup_{\rho\in \Lag'\cap \mathcal{DE}_+} \Phi^t(\rho).$$

If $U_\xi$ has been chosen small enough, then $\Lag$ will be included in  $\mathcal{V}_\xi$, and by working in $\mathbb{B}$, we can check that the hypotheses \ref{puissanceplus} and \ref{puissanceplus2} are satisfied.

\begin{figure}
    \center
   \includegraphics[scale=0.9]{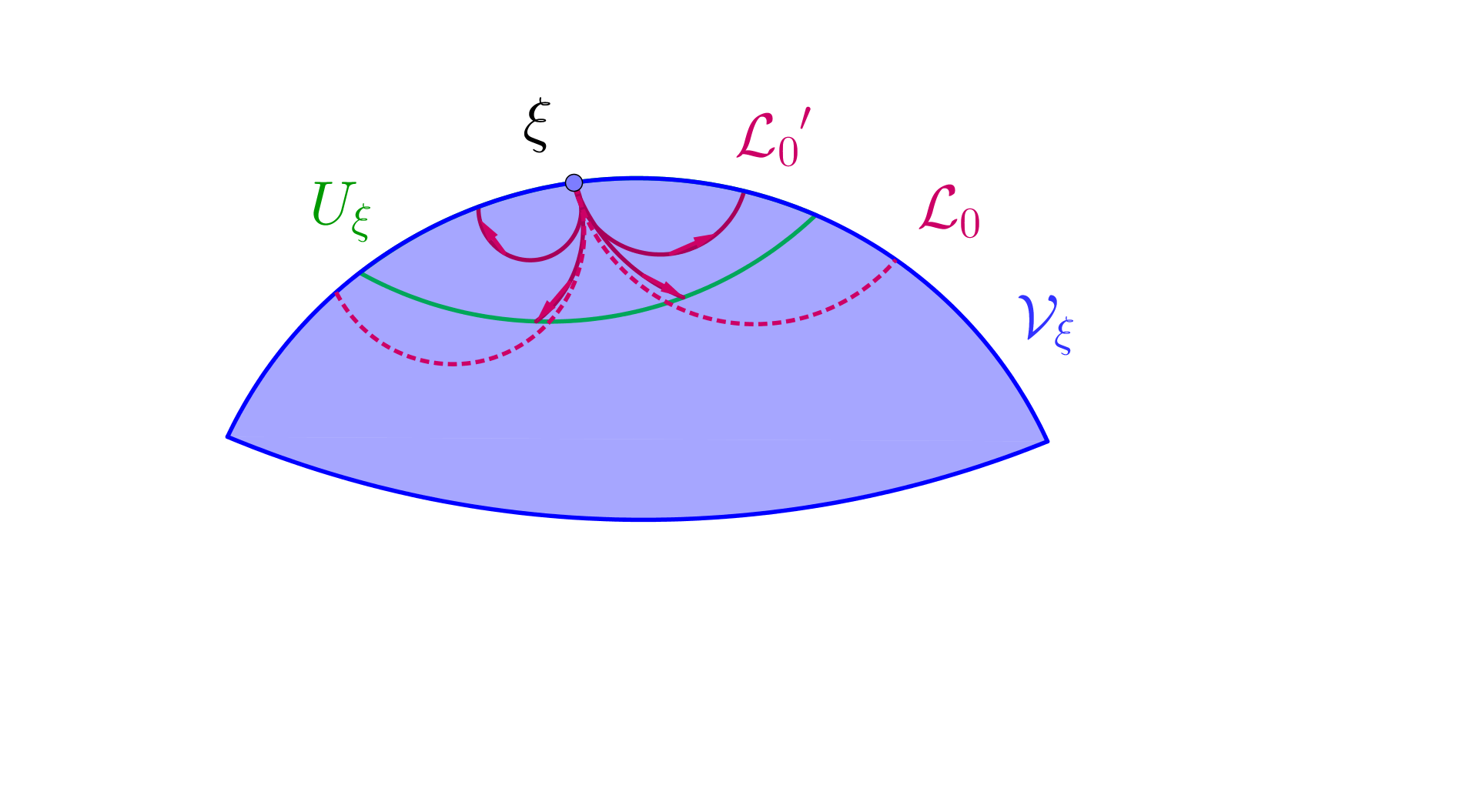}
    \caption{The Lagrangian manifolds $\Lag'$ and $\Lag$ when $(X,g)$ is hyperbolic near infinity.} \label{varietehyp}
\end{figure}

\paragraph{Construction of $E_h^1$}
We set $E_h^1:= -R_h F_h$, where $R_h$ is the outgoing resolvent $$(-h^2\Delta -\frac{(d-1)^2}{4} h^2 - 1 -i0)^{-1},$$ and $F_h:= [h^2\Delta, \tilde{\chi}] e^{((d-1)/2+i/h)\phi_\xi(x)}$. 

Let us now check that $E_h^1$ is a tempered distribution. As explained in \cite[\S 7.2]{DG}, we have 
\begin{equation}\label{grece}
\|b^{-1}F_h\|_{L^2(X)}= O(h).
\end{equation}
We must thus check that for any $\chi\in C_c^\infty(X)$, there exists $C,N\bel 0$ such that
\begin{equation}\label{montaigne}
\|\chi R_h b\|_{L^2\rightarrow L^2}\leq C h^{-N}.
\end{equation}
To prove such an estimate, we want to use the results of \cite{DV}. However, their main theorem does not apply directly here, and we have to adapt it a little.
Let us write $\hat{X}_0:=\{x\in X; b(x) \geq \epsilon_0/4\}$ and $\hat{X}_1:=\{x\in X; b(x) < \epsilon_0/2\}$, where $\epsilon_0$ is as in Hypothesis \ref{Guepard}. These manifolds are represented on figure \ref{gluingfigure}.

\begin{figure}
    \center
   \includegraphics[scale=0.9]{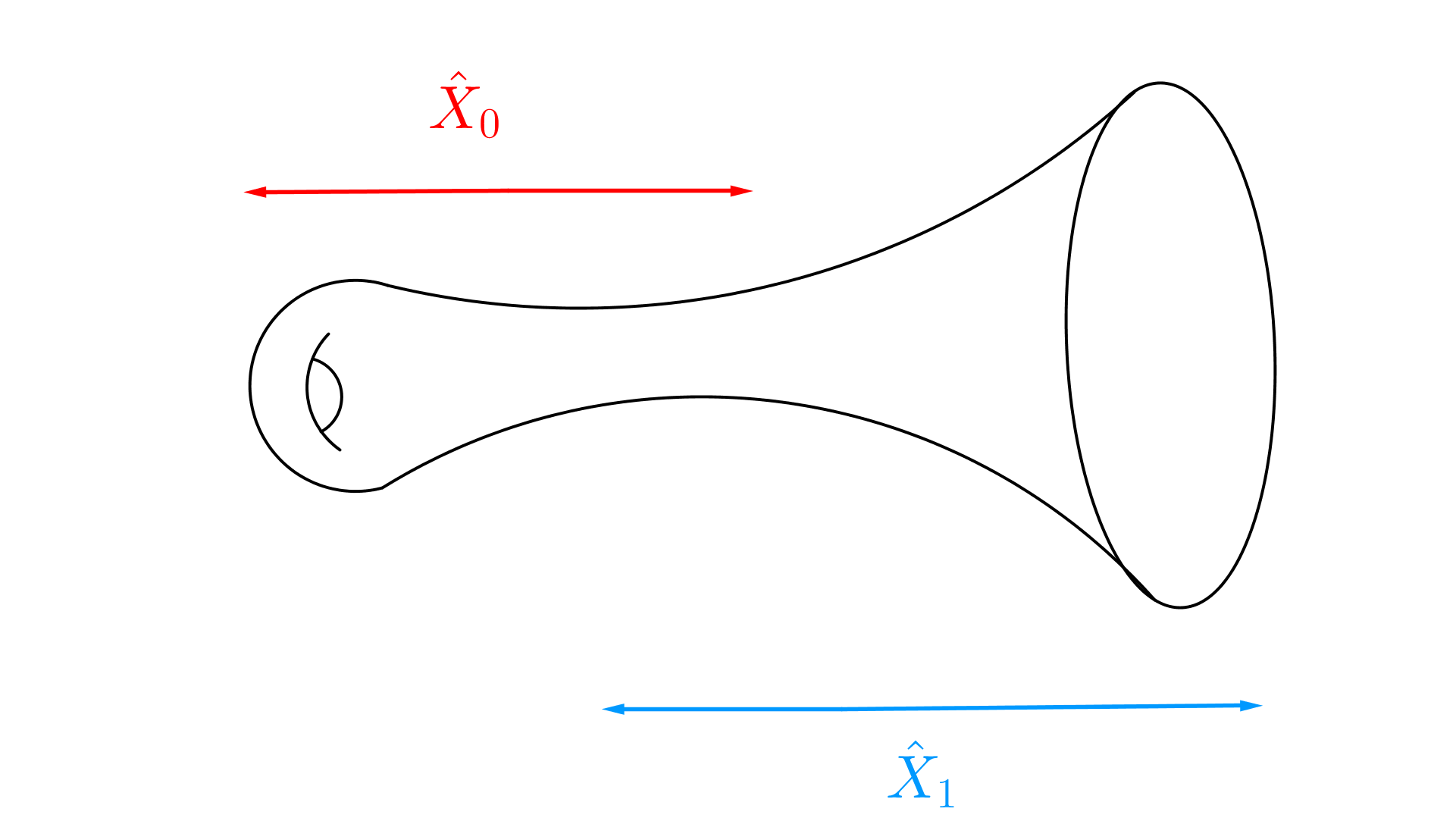}
    \caption{The manifolds $\hat{X}_0$ and $\hat{X}_1$} \label{gluingfigure}
\end{figure}

Let $\rho_0$,$\rho_1$, $\rho$, $\tilde{\rho}_0$, $\tilde{\rho}_1$, $\tilde{\rho}$ be bounded functions in $C^\infty(X)$ such that
\begin{itemize}
\item $\rho_0,\tilde{\rho}_0=1$ on $\hat{X}_0$, and $\rho_0,\tilde{\rho}_0$ vanish outside of $\{x\in X; b(x)\leq \epsilon_0/8\}$.
\item $\rho_1,\tilde{\rho}_1=1$ on $\hat{X}_0\cap \hat{X}_1$, $\rho_1,\tilde{\rho}_1$ vanish in $\{x\in X; b(x)\geq \epsilon_0\}$. We do not specify yet the behaviour of $\rho_1$ and $\tilde{\rho}_1$ in $X_1\backslash X_0$: actually, in the sequel, these two functions will have a different behaviour in this region of $X$.
\item  $\rho,\tilde{\rho}\in C^\infty(X)$ are such that $\rho=\rho_1$ on $X_1$, $\tilde{\rho}=\tilde{\rho}_1$ on $X_1$ and $\rho=\tilde{\rho}=1$ on $X\backslash X_1$.
\end{itemize}

Suppose that there exist constants $C,N\bel 0$ such that for $j=0,1$, we have
\begin{equation}\label{fureter}
\|\rho_j R_h\tilde{\rho}_j\|_{L^2(X)\rightarrow L^2(X)}\leq Ch^{-N}.
\end{equation}
We are going to show that this implies the existence of constants $C',N'\bel 0$ such that
\begin{equation}\label{musil}
\|\rho R_h\tilde{\rho}\|_{L^2(X)\rightarrow L^2(X)}\leq C'h^{-N'}.
\end{equation}
In \cite{DV}, the authors prove (\ref{musil}) only in the case where $\rho_0=\tilde{\rho}_0$, $\rho_1=\tilde{\rho}_1$, $\rho=\tilde{\rho}$. However, the proof of (\ref{musil}) follows exactly the same lines. Let us recall them, for completeness.

Let $\chi_0\in C^\infty(\mathbb{R};[0,1])$ be such that $\chi_0(s)=1$ if $s\geq 5\epsilon_0/12$ and $\chi_0(s)=0$ if $s\leq \epsilon_0/3$ and let $\chi_1=1-\chi_0$.

We then define a parametrix for $(P_h-1)$ as follows. Let
\begin{equation*}
F:= \chi_0(b(x) +\epsilon_0/12) R_h \chi_0(b(x)) + \chi_1(b(x)-\epsilon_0/12) R_h\chi_1(b(x)).
\end{equation*}
We set $A_0:=[P_h, \chi_0\big{(}b(\cdot)+\epsilon_0/12\big{)}] R_h (\chi_0\circ b )$ and $A_1:=[P_h, \chi_1\big{(}b(\cdot)-\epsilon_0/12\big{)}] R_h (\chi_1\circ b)$. By the same algebraic computations as in \cite[\S 3]{DV}, we obtain that
\begin{equation*}
(P_h-1)(F-FA_0-FA_1+FA_0A_1)=Id - A_1A_0+A_1A_0A_1= Id + O_{L^2\rightarrow L^2}(h^\infty),
\end{equation*}
where the second inequality comes from (3.3) in \cite{DV}.

Consequently, for $h$ small enough, we have that
\begin{equation}\label{Gremeau}
\begin{aligned}
\|\rho R_h \tilde{\rho}\|_{L^2\rightarrow L^2} &\leq C \big{\|}\rho(F-FA_0-FA_1+FA_0A_1)\tilde{\rho}\big{\|}_{L^2\rightarrow L^2}\\
& = C \big{\|}\rho ( F-\chi_0\big{(}b(\cdot)+ \epsilon_0/12\big{)} R_h (\chi_0\circ b) A_1\\
&+ \chi_1\big{(}b(\cdot)-\epsilon_0/12\big{)} R_h (\chi_1\circ b) (-A_0+A_0A_1))\tilde{\rho}\big{\|}.
\end{aligned}
\end{equation}
Let us use the triangular inequality, and bound each of the terms. We have
$$\|\rho F \tilde{\rho}\|\leq \|\rho_0 R_h \tilde{\rho}_0\|+\|\rho_1 R_h \tilde{\rho}_1\|\leq C h^{-N}$$
by (\ref{fureter}).

Next, we have 
\begin{equation*}
\begin{aligned}
\big{\|}\rho(\chi_0\big{(}b(\cdot)+ \epsilon_0/12\big{)} R_h (\chi_0\circ b) A_1)\tilde{\rho}\big{\|} &\leq \big{\|}\rho_0 R_h (\chi_0\circ b) \big{[}P_h, \chi_1\big{(}b(\cdot)-\epsilon_0/12\big{)}\big{]}R_h \tilde{\rho}_1\big{\|}\\
&\leq C \|\rho_0 R_h \tilde{\rho}_0\| \|\rho_1 R_h \tilde{\rho}_1\|\\
&\leq C h^{-2N},
\end{aligned}
\end{equation*}
where, to go from the first line to the second, we used the fact that $\rho_1=1$ on $X_0\cap X_1$ and that, by considerations on the supports, we have
\begin{equation}\label{supportborne}
\begin{aligned}
\big{[}P_h, \chi_1\big{(}b(\cdot)-\epsilon_0/12\big{)}\big{]} (1-\rho_1) &\equiv 0\\
 (1-\tilde{\rho}_0)\big{[}P_h, \chi_1\big{(}b(\cdot)-\epsilon_0/12\big{)}\big{]}&\equiv 0.
 \end{aligned}
 \end{equation}

As to the last term of  (\ref{Gremeau}), we may bound it by
\begin{equation*}
\begin{aligned}
&\big{\|}\rho\chi_1(x-\epsilon_0/12) R_h \chi_1(-A_0+A_0A_1)\tilde{\rho}\big{\|}\\
&\leq \big{\|}\rho_1 R_h (\chi_1\circ b) \big{(}-\big{[}P_h, \chi_0\big{(}b(\cdot)\epsilon_0/12\big{)]} R_h (\chi_0\circ b) \\
&+ \big{[}P_h, \chi_0\big{(}b(\cdot)+\epsilon_0/12\big{)]} R_h (\chi_0\circ b) \big{[}P_h, \chi_1\big{(}b(\cdot)+\epsilon_0/12\big{)]} R_h (\chi_1\circ b) \big{)} \tilde{\rho}\big{\|}\\
&\leq C \big{(} \| \rho_1 R_h \tilde{\rho_1} \rho_0 R_h \tilde{\rho_0}\| + \|\rho_0 R_h \tilde{\rho}_0 \rho_1 R_h \tilde{\rho}_1 \|\big{)}\\
&\leq C h^{-2N},
\end{aligned}
\end{equation*}
where, to obtain the second inequality, we used again (\ref{supportborne}).

This concludes the proof of (\ref{musil}).

To prove (\ref{montaigne}), we use  \cite[Theorem 5.1]{vasy2011microlocal} (for $s=1$), which says that if $(X,g)$ is asymptotically hyperbolic and has no trapped set, then for any $\sigma\in \mathbb{R}$, we have
\begin{equation*}
\|b^{-(d-1)/2 +i\sigma} (-\Delta-(\frac{d-1}{2})^2-\sigma^2)^{-1} f\|_{H^1_{|\sigma|^{-1}}(X_{0,even})}\leq \frac{C}{|\sigma|}\|b^{-(d+3)/2+i\sigma} f\|_{H^0_{|\sigma|^{-1}}(X_{0,even})}.
\end{equation*}
In the notations of \cite{vasy2011microlocal}, we have $\|f\|_{H^0_{|\sigma|^{-1}}(X_{0,even})}\sim \|b^{(d+1)/2} f\|_{L^2(X)}$.
Since, furthermore, the $L^2$ norm in a compact set may be bounded by the $H^1_{|\sigma|^{-1}(X_{0,even})}$ norm, we obtain that for any $\chi \in C_c^\infty(X)$ we have \begin{equation*}
\|\chi (-h^2\Delta - h^2 \Big{(}\frac{d-1}{2}\Big{)}^2 - 1)^{-1}f\|_{L^2(X)}\leq \frac{C'}{h}\|b^{-1}f\|_{L^2(X)}.
\end{equation*}

Consequently, for any $f\in L^2(X)$, we have
\begin{equation*}
\|\chi (-h^2\Delta - h^2 \Big{(}\frac{d-1}{2}\Big{)}^2 - 1)^{-1}b f\|_{L^2(X)}\leq \frac{C'}{h}\|f\|_{L^2(X)},
\end{equation*}
which gives us (\ref{montaigne}) in the non-trapping case. 

If the trapped sed is non-empty, we glue together the resolvent estimates as in the proof of Theorem 6.1 in \cite{DV}, but by using (\ref{musil}) instead of \cite[Theorem 2.1]{DV}. More precisely, we take $\rho_0,\rho_1$ and $\rho'_0$ compactly supported, and $\rho'_1=b$. This gives us  (\ref{montaigne}).

By combining this with (\ref{grece}), we have shown that $E_h^1$ is a tempered distribution. We may then check easily that

\begin{equation*}
\Big{(}-h^2\Delta - h^2\frac{(d-1)^2}{4}-1\Big{)}E_h=0. 
\end{equation*}

\paragraph{Wave-front set of $E_h^1$}
Let us now show that $E_h^1$ is such that there exists $\epsilon_2\bel 0$ such that for any $(x,\xi)\in S^*X$ such that $b(x)<\epsilon_2$, we have
\begin{equation}\label{outgoing2}
\rho\in WF_h (E_h^1) \Rightarrow \rho\in\mathcal{DE}_+.
\end{equation}

In \cite[\S 7]{DG}, the authors proved that for all $\rho\in WF_h(E_h^1)$, we have $\rho\in S^*X$ and we have either:

(i) $\rho\in \Gamma^+$, that is, $\rho$ is in the outgoing tail, or

(ii) There exists $t\geq 0$ such that $\Phi^{-t}(\rho)\in \{(x,\partial_x\phi_\xi(x)); x\in \spt(\partial \tilde{\chi})\}$, where $\tilde{\chi}$ and $\phi_\xi$ are as in the construction of $E_h^0$.

Suppose that $\rho\in  WF_h(E_h^1)\cap \mathcal{DE}_-$ is such that there exists $t\geq 0$ such that $\Phi^{-t}(\rho)\in \{(x,\partial_x \phi_\xi(x)); x\in \spt(\partial \tilde{\chi})\}$.
 Then, $d_{\overline{g}}(\Phi^{-t'}(\rho),\xi)$ will be a decreasing function going to zero as $t'\rightarrow +\infty$. 

Let us denote by $U'_\xi\subset U_\xi$ a neighbourhood of $\xi$ in $\overline{X}$ on which $\tilde{\chi}$ is equal to one. 
 
But then, as explained in \cite[Assumption (A7)]{DG}, if we take $\epsilon_2$ small enough, we may suppose that for any $\rho'\in S^*X$ such that $b(\rho') \leq \epsilon_2$ and $d_{\overline{g}}(\Phi^{-t'}(\rho),\xi)$ will be a decreasing function going to zero as $t'\rightarrow +\infty$, we have $\rho'\in \pi_X(U'_\xi)$. Therefore, $\pi_X(\rho)\in U'_\xi$, and $\pi_X(\Phi^{-t}(\rho))\in U'_\xi$. This is absurd, since $\pi_X(\Phi^{-t}(\rho))\in \spt(\partial \tilde{\chi})$, and $\tilde{\chi}\equiv 1$ on $U'_\xi$. This proves (\ref{outgoing2}).
\end{appendices}

\bibliographystyle{alpha}
\bibliography{references}
\end{document}